%% file: main.tex
\documentclass[11pt]{article}
\input{mary_macros}

\input{Dorsa_macros}

\usepackage{thm-restate}
\newcommand{\MPJ}[2]{\lfloor #1 \rfloor_{2^{#2}}  - 2^{#2}}    

\newcommand{\seed}{\mathrm{sd}}
\newcommand{\defn}[1]{\textbf{#1}}
\newcommand{\MP}{\mathrm{MP}}
\newcommand{\MPone}{\mathrm{MP1}}
\newcommand{\MPtwo}{\mathrm{MP2}}
\newcommand{\MPthree}{\mathrm{MP3}}
\newcommand{\MPi}{\mathrm{MPi}}
\newcommand{\BVC}{\mathrm{BVC}}
\newcommand{\constbvc}{\gamma} 
\newcommand{\vone}{v^{(1)}}
\newcommand{\vtwo}{v^{(2)}}
\newcommand{\vthree}{v^{(3)}}
\newcommand{\vi}[1]{v^{(#1)}}

\newcommand{\MPall}{\mathrm{MP}^{\{1,2,3\}}}
\newcommand{\MPalljl}{\mathrm{MP}^{\{1,2,3\}}}
\newcommand{\MPonetwojl}{\mathrm{MP}^{\{1,2\}}}
\newcommand{\MPonetwo}{\mathrm{MP}^{\{1,2\}}}
\newcommand{\M}{\mathtt{AvMPs}} 
\newcommand{\Mem}{\mathtt{Memory}} 
\newcommand{\megastate}{\mathtt{p}} 
\newcommand{\ms}[1]{\mathtt{ms}(#1)} 
\newcommand{\var}{\mathtt{var}}
\newcommand{\pv}{\mathtt{p.v}} 
\newcommand{\pdepth}{\mathtt{p.depth}} 
\newcommand{\tjmp}{t_{\mathrm{jump}}} 
\newcommand{\phash}{\operatorname{\mathtt{p.prev-hash}}}
\newcommand{\pseed}{\operatorname{\mathtt{p.prev-seed}}}
\newcommand{\piter}{\mathtt{p.iter}}

\newcommand{\qihash}{\ensuremath{\operatorname{\mathtt{q}_i\texttt{.prev-hash}}}} 
\newcommand{\qiseed}{\ensuremath{\mathtt{q}_i\operatorname{\mathtt{.prev-seed}}}}
\newcommand{\qiiter}{\ensuremath{\mathtt{q}_i\mathtt{.iter}} }
\newcommand{\qit}{\ensuremath{\mathtt{q}_i\mathtt{.T}} }
\newcommand{\paiter}{\operatorname{\mathtt{\megastate}}^*_A\operatorname{\mathtt{.iter}}}

\newcommand{\pbiter}{\operatorname{\mathtt{\megastate}}^*_B\operatorname{\mathtt{.iter}}}
\newcommand{\pahash}{\operatorname{\mathtt{\megastate}}^*_A\operatorname{\mathtt{.prev-hash}}}
\newcommand{\pbhash}{\operatorname{\mathtt{\megastate}}^*_B\operatorname{\mathtt{.prev-hash}}}
\newcommand{\paseed}{\operatorname{\mathtt{\megastate}}^*_A\operatorname{\mathtt{.prev-seed}}}
\newcommand{\pbseed}{\operatorname{\mathtt{\megastate}}^*_B\operatorname{\mathtt{.prev-seed}}}
\newcommand{\pat}{\operatorname{\mathtt{\megastate}}^*_A\operatorname{\mathtt{.T}}}
\newcommand{\pbt}{\operatorname{\mathtt{\megastate}}^*_B\operatorname{\mathtt{.T}}}
\newcommand{\tcmp}{t_{\mathrm{cmp}}}

\newcommand{\IHA}{I_{H_A}}
\newcommand{\IHB}{I_{H_B}}

\newcommand{\Btotal}{B_{\mathrm{total}}} 
\newcommand{\Itotal}{I_{\mathrm{total}}} 
\newcommand{\Iblock}{I_{\mathrm{perBlock}}} 
\newcommand{\Aiter}{\ensuremath{\ell_A}} 
\newcommand{\Biter}{\ensuremath{\ell_B}} 
\newcommand{\simPath}{\ensuremath{\mathcal{T}}} 

\newcommand{\Icount}{\ensuremath{I_{\mathrm{current}}}}
\newcommand{\Icnt}{\ensuremath{I_{\mathrm{cnt}}}}
\newcommand{\HC}{\ensuremath{\mathsf{HC}}} 
\newcommand{\protosneaky}{sneaky attack\ }
\newcommand{\protosneakyf}{sneaky attack}
\newcommand{\Protosneaky}{Sneaky attack}
\newcommand{\specialpt}{jumpable point\ }
\newcommand{\specialptf}{jumpable point}
\newcommand{\specialscale}{jumpable scale\ }
\newcommand{\Specialpts}{Jumpable points}
\newcommand{\sneakyjmp}{sneaky jump\ }
\newcommand{\sneakyjmpf}{sneaky jump}
\newcommand{\tsp}{t^\dagger}
\newcommand{\Isp}{I^\dagger}
\newcommand{\jsp}{j^\dagger}

\newcommand{\csneak}{c^*} 
\newcommand{\clower}{c_l}
\newcommand{\cupper}{c_u} 
\newcommand{\Rblock}{R_{\mathrm{block}}}

\newcommand{\Rtblock}{\tilde{R}_{\mathrm{block}}}
\newcommand{\Riter}{R_{\mathrm{iter}}}

\newcommand{\Enc}{\mathrm{Enc}}
\newcommand{\Dec}{\mathrm{Dec}}
\newcommand{\divewindow}{\mathcal{I}_{\mathrm{dive}}}
\newcommand{\divewindows}[1]{\mathcal{I}_{\mathrm{dive},#1}}
\newcommand{\Chash}{C_{\mathrm{hash}}}

\newcommand{\Chashtwo}{{K}_{\mathrm{hash}}}
\newcommand{\tcq}{t_{c_q}} 
\newcommand{\tcqone}{t_{c_{q_1}}}

\newcommand{\tphatone}{t_{\hat{p}_1}}
\newcommand{\tphattwo}{t_{\hat{p}_2}}

\newcommand{\tqhattwo}{t_{\hat{q}_2}}

\newcommand{\Icorr}{I_{\mathrm{corr}}}
\newcommand{\sd}{\mathrm{sd}}
\newcommand{\ttq}{\mathtt{q}}
\newcommand{\obighash}{C_b \log d }

\newcommand{\phat}{\hat{p}}
\newcommand{\qhat}{\hat{q}}
\newcommand{\tphat}{t_{\hat{p}}}

\newcommand{\Bcorr}{B_{\mathrm{corrupted}} }
\newcommand{\mseq}{\cong}
\newcommand{\Bsim}{T}
\newcommand{\Bsimp}{T^{\leq \mathtt{p}}}
\newcommand{\pt}{\mathtt{p.T}}

\title{Interactive Coding with Small Memory and Improved Rate}
\author{Dorsa Fathollahi\thanks{Stanford University. Research supported in part by NSF grants CCF-1844628, CCF-2231157 and CCF-2133154. dorsafth@stanford.edu } \and Bernhard Haeupler\thanks{INSAIT, Sofia University "St Klimemt Ohridski", Bulgaria \& ETH Zurich, Switzerland. Partially funded by the European Union's Horizon 2020 ERC grant 949272 and partially funded by the Ministry of Education and Science of Bulgaria (support for INSAIT, part of the Bulgarian National Roadmap for Research Infrastructure). bernhard.haeupler@inf.ethz.ch } \\ \and Nicolas Resch\thanks{  Informatics' Institute, University of Amsterdam. Research supported in part by an NWO Veni grant (VI.Veni.222.347).n.a.resch@uva.nl } \and Mary Wootters\thanks{ Stanford University. Research supported in part by NSF grants CCF-1844628, CCF-2231157 and CCF-2133154. marykw@stanford.edu  }}
\date{\today}

\begin{document}

\maketitle

\begin{abstract}
    In this work, we study two-party interactive coding for adversarial noise, when both parties have limited memory.  We show how to convert any adaptive protocol $\Pi$ into a protocol $\Pi'$ that is robust to an $\eps$-fraction of adversarial corruptions, not too much longer than $\Pi$, and which uses small space. More precisely, if $\Pi$ requires space $\log(s)$ and has $|\Pi|$ rounds of communication, then $\Pi'$ requires $O_\epsilon(\log s \log |\Pi|)$ memory, and has $$|\Pi'| = |\Pi|\cdot\left( 1 + O\left( \sqrt{ \eps \log \log 1/\eps } \right)\right)$$ rounds of communication.  The above matches the best known communication rate, even for protocols with no space restrictions.
\end{abstract}

\section{Introduction}\label{sec:intro}
\input{intro}

\section{Technical Overview}\label{sec:tech}
\input{tech}

\section{Definitions and Preliminaries}\label{sec:prelim}
\input{prelim}

\section{Protocol}\label{sec:protocol}
\input{protocol}

\section{Analysis}\label{sec:analysis}

\input{analysis_start}

\subsection{Useful Lemmas and Definitions}\label{subsec:def-lem}

\input{defs_lems}

\subsection{Sneaky Attack and Main Technical Lemma}\label{sec:maintech}

\input{sneaky}

\subsection{Potential Function Analysis}\label{sec:progressProof}
\input{potential}

\subsection{Hash Function Analysis}\label{sec:hash}

\input{hash}

\subsection{Proof of Theorem~\ref{thm:main}}\label{sec:done}
\input{rest_of_analysis}

\section{Proof of Main Technical Lemma}\label{sec:grossproof}

\input{maintechrestate}

\bibliographystyle{alpha}
\bibliography{refs.bib}

\end{document}

%% file: mary_macros.tex
\usepackage[margin=1.125in]{geometry}
\usepackage[table]{xcolor}
\definecolor{DarkGreen}{rgb}{0.1,0.5,0.1}
\definecolor{DarkRed}{rgb}{0.5,0.1,0.1}
\definecolor{DarkBlue}{rgb}{0.1,0.1,0.5}

\usepackage[small]{caption}
\usepackage[pdftex]{hyperref}
\hypersetup{
    unicode=false,          
    pdftoolbar=true,        
    pdfmenubar=true,        
    pdffitwindow=false,      
    pdfnewwindow=true,      
    colorlinks=true,       
    linkcolor=DarkBlue,          
    citecolor=DarkGreen,        
    filecolor=DarkGreen,      
    urlcolor=DarkBlue,          
    pdftitle={},
    pdfauthor={},    
}
\usepackage{amssymb,amsmath,amsthm,amsfonts,enumitem}
\usepackage{fullpage,nicefrac,comment}
\usepackage{tikz}
\usetikzlibrary{arrows,decorations.pathmorphing,decorations.shapes,snakes,patterns,positioning, shapes.callouts, shapes.arrows}
\usepackage[noend]{algpseudocode}
\usepackage{algorithm}

\makeatletter
\newenvironment{breakablealgorithm}
  {
   \begin{center}
     \refstepcounter{algorithm}
     \hrule height.8pt depth0pt \kern2pt
     \renewcommand{\caption}[2][\relax]{
       {\raggedright\textbf{\ALG@name~\thealgorithm} ##2\par}%
       \ifx\relax##1\relax 
         \addcontentsline{loa}{algorithm}{\protect\numberline{\thealgorithm}##2}%
       \else 
         \addcontentsline{loa}{algorithm}{\protect\numberline{\thealgorithm}##1}%
       \fi
       \kern2pt\hrule\kern2pt
     }
  }{
     \kern2pt\hrule\relax
   \end{center}
  }
\makeatother

\newcommand{\I}{\ensuremath{\mathcal{I}}}

\newcommand{\cS}{\ensuremath{\mathcal{S}}}

\newcommand{\F}{{\mathbb F}}

\newcommand{\PR}[1]{{\mathbb{P}}\left\{ #1\right\}}

\newcommand{\EE}{\mathbb{E}}

\newcommand{\inparen}[1]{\left(#1\right)}
\newcommand{\inbrak}[1]{\left[#1\right]}

\newcommand{\poly}{\mathrm{poly}}

\newcommand{\eps}{\varepsilon}
\renewcommand{\epsilon}{\varepsilon}

\newtheorem{theorem}{Theorem} 
\newtheorem{lemma}[theorem]{Lemma} 
\newtheorem*{lemma*}{Lemma}
\newtheorem{definition}{Definition}

\newtheorem{observation}[theorem]{Observation}

\newtheorem{corollary}[theorem]{Corollary} 

\newtheorem{remark}{Remark}
\newtheorem{asm}{Assumption}
\newtheorem{claim}[theorem]{Claim}
\newtheorem{subclaim}[theorem]{Sub-Claim}

%% file: Dorsa_macros.tex
\definecolor{orange(colorwheel)}{rgb}{1.0, 0.5, 0.0}
\definecolor{darkpastelblue}{rgb}{0.17, 0.12, 0.8}
\definecolor{yellowgreen}{HTML}{98CC70}
\definecolor{emerald}{rgb}{0.31, 0.78, 0.47}

\definecolor{ferngreen}{rgb}{0.11, 0.44, 0.26}
\definecolor{iris}{rgb}{0.35, 0.31, 0.81}
\definecolor{BurntOrange}{HTML}{F7921D}

%% file: intro.tex
We study the problem of \emph{interactive communication over a noisy channel}.  Suppose that two parties, Alice and Bob, would like to carry out an interactive protocol $\Pi$. Formally, $\Pi$ is represented by a DAG and a transition function (see Section~\ref{sec:prelim}); informally, $\Pi$ contains instructions for how Alice and Bob should pass messages back and forth, for example to compute some function of interest.

However, Alice and Bob cannot carry out $\Pi$ directly, because  the channel between Alice and Bob is noisy, implying that the messages they pass back and forth might not be received correctly.  We study an adversarial model of corruption:  An adversary (with full knowledge of $\Pi$ and any inputs that Alice and Bob may have) is allowed to corrupt up to an $\eps$-fraction of the bits that are sent in either direction, over the course of the entire protocol.  The adversary is \emph{adaptive}, meaning that whether they choose to introduce a corruption can depend on what has happened so far. The goal is then to transform $\Pi$ into a \emph{robust} protocol, $\Pi'$, which allows Alice and Bob to {simulate} $\Pi$, even in the presence of such an adversary.

Our approach works in two different models.  In the first model, both the original protocol $\Pi$ and the robust protocol $\Pi'$ have \emph{alternating speaking order,} meaning that Alice and Bob take turns speaking.  In the second model, Alice and Bob's speaking order need not be fixed in advance (for either $\Pi$ or $\Pi'$).  That is, whether or not Alice or Bob transmit in iteration $i$ can depend not just on $i$, but also on their inputs, the transcript so far, as well as on any private randomness that Alice and Bob use.\footnote{We note that in this second model, it is possible that in the robust protocol $\Pi'$, both or neither of Alice and Bob may try to transmit in the same iteration.  We work in the ``speak-or-listen'' model of \cite{ghaffari2014optimal}, which was also used by other works in the same parameter regime we work in, \emph{e.g.} \cite{H14}. In this model, if both Alice and Bob transmit, then neither hears anything; and if neither transmit, then the adversary may make them hear anything, and this does not count toward the corruption budget.} See Section~\ref{sec:prelim} for more details on this second model.  

The problem of interactive communication over a noisy channel was first studied by Schulman~\cite{schulman1992communication,sch96}, and since then there has been a huge body of work on it; we refer the reader to~\cite{ran_survey} for an excellent survey.  Traditionally, work on the two-party problem against an adaptive adversary has focused on the trade-off between the following three quantities:
\begin{itemize}
    \item The \defn{rate} of the scheme, which captures how much communication overhead is required to make $\Pi$ robust.  Formally, this is defined as $|\Pi|/|\Pi'|$, where $|\Pi|$ denotes the number of rounds of communication in $\Pi$.\footnote{We assume that $\Pi$ has a fixed length that is known ahead of time; given this assumption, our robust protocol $\Pi'$ will also have a fixed length.} Thus, the rate is always at most $1$, and the goal is to make it as close to $1$ as possible.
    \item The \defn{corruption budget} of the adversary, which is the ``$\eps$'' above; the adversary is allowed to corrupt an $\eps$-fraction of communications, and the goal is to make $\eps$ as large as possible.
  
    \item The \defn{computational efficiency}, which is the total running time for Alice and Bob to do the simulation.  The goal is for this to be polynomial (or even linear) in the running time required for the original protocol $\Pi$.
\end{itemize}
After decades of work, there are now protocols that check all three boxes in a variety of parameter regimes.  In our work, we focus on the extremely high-rate regime, so we take $\eps$ to be a small constant, and want rate very close to $1$.  In this parameter regime, the protocol of \cite{H14} checks all three boxes: it is efficient, and obtains rate 
\begin{equation}\label{eq:desiredrate} 
R = 1 - O\bigl(\sqrt{ \eps \log\log(1/\eps) }\bigr) \tag{$\star$}
\end{equation}
when the corruption budget is $\eps$.  The rate \eqref{eq:desiredrate} is conjectured to be optimal in this setting~\cite{H14}. 

\paragraph{Small space interactive coding.} In our work, we add one more criterion to the list of desiderata:
\begin{itemize}
    \item The \defn{memory requirements} for Alice and Bob.  That is, if $\Pi$ can be run in small space, the robust version $\Pi'$ should also use small space.
\end{itemize}
In this work, as with previous work, we quantify the space that $\Pi$ requires by the number of states $s$ in the DAG that represents $\Pi$ (see Section~\ref{sec:prelim}).  If $\Pi$ has $s$ states, then the memory requirement is $O(\log s)$ bits.  Thus, our goal is for $\Pi'$ to use only a small multiple of $\log(s)$ bits.

Space-efficiency is an important step towards bringing interactive coding schemes closer to practical applicability.  
Existing interactive coding schemes that do \emph{not} explicitly take memory requirements into account (including the result of \cite{H14} mentioned above) require the parties to (at least) remember the entire transcript at every stage of the protocol $\Pi$.  This can be wasteful if $\Pi$ is very long, especially if $\Pi$ can itself be run in small space.

Small-space interactive communication over noisy channels  was first studied in the unpublished manuscript~\cite{HR18},\footnote{We note that this manuscript has been retracted, but we use several of its ideas in our work. Indeed, that manuscript is by a subset of the authors of the current work, and we view \cite{HR18} as a preliminary version of this work.} and subsequently by Chan, Liang, Polychroniadou and Shi in \cite{chan2020small} and by Efremenko, Haeupler, Kol, Resch, Saxena and Kalai in \cite{EHKKRS23}. We summarize their results in Table~\ref{table:prevwork}, and describe them in more detail below.

\begin{table}[h!]
\centering
\begin{tabular}{|c|c|c|c|}
\hline
\textbf{Reference} & \textbf{Rate} & \textbf{Space} & \textbf{Adversary}  \\
\hline\hline
\cite{H14} &  \cellcolor{green!20} $1 - O\left(\sqrt{\eps \log\log(1/\eps)}\right)$ & $\Omega(\log(s)\cdot |\Pi|)$ &  \cellcolor{green!20}Adaptive  \\
\hline
\cite{chan2020small} & \cellcolor{green!20} $1 - O(\sqrt{\eps})$ & \cellcolor{green!20} $O(\log(s) \cdot \log|\Pi|)$ & Oblivious \\
\hline
\cite{EHKKRS23} & $1 - O\left( \sqrt[3]{\eps \log(1/\eps) } \right)$ & \cellcolor{green!20} $O(\log(s) \cdot \log|\Pi|)$ & \cellcolor{green!20} Adaptive \\
\hline
Our work & \cellcolor{green!20} $1 - O\left(\sqrt{\eps \log\log(1/\eps)}\right)$ & \cellcolor{green!20} $O(\log(s)\cdot  \log|\Pi|)$  & \cellcolor{green!20} Adaptive \\
\hline
\end{tabular}
    \caption{Work on space-bounded noisy interactive communication in the high-rate regime against an adversary with a corruption budget of $\eps$.  We note that \cite{H14} does not try to minimize space, but it attains the best known rate against and adaptive adversary.  
    Above, the original protocol $\Pi$ has \emph{size} $s$ (see Section~\ref{sec:prelim} for formal definitions), meaning that it requires $O(\log s)$ space. 
 We have highlighted ``good'' (meaning, desired for this paper) values in green.}\label{table:prevwork}
\end{table}

The work of Chan et al.~\cite{chan2020small} designs a protocol for the case of an \emph{oblivious adversary},\footnote{That is, the adversary's decisions about which transmissions to corrupt can only depend on the round number, and not on the prior communication.} rather than the \emph{adaptive adversary} that we consider. 
They achieve small space---only $O(\log s \log|\Pi|)$, where as above $s$ is number of states in the DAG that represents $\Pi$.
While their protocol only applies to oblivious adversaries, they do obtain a very good rate of $1 - O(\sqrt{\eps})$, which is conjectured to be optimal for oblivious adversaries~\cite{H14}. 

Like our work, the work of Efremenko et al.~\cite{EHKKRS23}  considers an adaptive adversary, and also obtains space $O(\log s \log|\Pi|)$.  However, while the rate does approach $1$ as $\eps \to 0$, it is of the form
\[ R = 1 - O\left( \sqrt[3]{\eps \log(1/\eps) } \right),\]
which is smaller than the rate \eqref{eq:desiredrate} that \cite{H14} achieves, both because of the cube root instead of the square root, and also because of the $\log(1/\eps)$ rather than $\log\log(1/\eps)$.

\paragraph{Our result.} In this work, we show that one does not need to sacrifice in rate in order to obtain a small space protocol, even for adaptive adversaries.  More precisely, we give a protocol $\Pi'$ that:
\begin{itemize}
    \item has rate matching \eqref{eq:desiredrate}, the best-known rate for protocols against an adaptive adversary \emph{even with no space requirements};
    \item correctly simulates $\Pi$ even in the presence of an adaptive adversary with a corruption budget of $\eps$; 
    \item is computationally efficient; and
    \item uses space at most $O(\log(s) \cdot \log|\Pi|)$, matching the space bound of \cite{chan2020small,EHKKRS23}. 
\end{itemize}

Formally, our main theorem is as follows.

\begin{restatable}{theorem}{mainThm}\label{thm:main}
Fix $\eps > 0$.
Let $\Pi$ be a two-party interactive protocol that requires space $\log s$.  Then there is a randomized protocol $\Pi'$ (with private randomness) that, with probability at least $1 - 1/\poly(|\Pi|)$, correctly simulates $\Pi$, in the presence of an adaptive adversary who may corrupt an $\eps$ fraction of the bits sent in either direction.  Moreover, the rate of the protocol is at least
\[ R \geq 1 - O\left( \sqrt{ \eps \log\log(1/\eps) }\right) ,\]
and the amount of space required for each of Alice and Bob is at most $O_\eps(\log s \cdot \log |\Pi|)$.\footnote{The proof of Theorem~\ref{thm:main} establishes that the space required is $O\inparen{\log(s) \log|\Pi| + \sqrt{ \frac{\log\log(1/\eps)}{\eps}}  \log|\Pi| } = O_\eps(\log s \cdot \log |\Pi|)$.}
Finally, if the running time of $\Pi$ is $T$, then the running time of $\Pi'$ is at most $T \cdot \poly(|\Pi|/\eps)$. 
\end{restatable}

\begin{remark}[Speaking order in Theorem~\ref{thm:main}]\label{rem:speaking_order}
As noted at the beginning of the paper, our result holds in two different models: in the first, the speaking order of both $\Pi$ and $\Pi'$ are alternating; in the second, the speaking orders of both $\Pi$ and $\Pi'$ are unrestricted.

We prove our result in this second model (see Section~\ref{sec:prelim} for more on the model).  However, an inspection of our algorithm shows that if the original protocol $\Pi$ is alternating, then the robust protocol $\Pi'$ has a fixed and periodic speaking order, which can furthermore easily be made to be alternating (see Remark~\ref{rem:alternating}); this proves the result in the first model as well.
\end{remark}

One of our technical contributions is the introduction of a new (to us) style of analysis for noisy interactive communication.
 Typically, the approach is to define a potential function $\Phi$, and show that it stays ``well-behaved'' throughout the execution of $\Pi'$.  We do define such a $\Phi$, but our $\Phi$ is \emph{not} always ``well-behaved.''  Thus, we augment the potential function analysis with another, more global analysis, which keeps track of how often $\Phi$ can behave poorly over the course of the entire run of $\Pi'$. 
 We discuss our protocol---which is an adaptation of the protocol of \cite{H14}---and our proof techniques further in Section~\ref{sec:tech}.

\paragraph{Related work.}
Coding for interactive communication dates back to the work of Schulman~\cite{schulman1992communication,sch96}. Since then, a long line of works have given constructions of robust protocols, considering many desiderata such as communication rate, tolerable error rate, error model, relaxed decoding notions (e.g. list-decoding), time complexity, and so on, for example~\cite{gelles2011efficient,braverman2011towards,braverman2012towards,brakerski2012efficient,kol2013interactive,brakerski2014fast,ghaffari2014optimal,ghaffari2014optimal2,
HaeuplerTransInf16p4588,
braverman2017coding,braverman2017constant,braverman2017list,haeupler2017synchronization,haeupler2017bridging,efremenko2020binary,chan2020small,EHKKRS23,efremenko2023rate}.
For further details, we recommend the excellent survey of Gelles~\cite{ran_survey}. 

In our work, we focus on interactive coding schemes whose rate approach $1$ as the error rate $\eps$ tends to $0$. The first progress on this question was achieved by Kol and Raz~\cite{kol2013interactive}, 
who provided a scheme of rate $1-O\big(\sqrt{\eps\log\tfrac1\eps}\big)$, which tolerated an $\eps$ fraction of \emph{random} errors. 

This protocol, and any (rate-efficient) protocol thereafter, applied the ``rewind-if-error'' paradigm, initiated by Schulman~\cite{schulman1992communication}: namely, the parties simulate the protocol for a while, but if they become convinced that errors have derailed the computation they can rewind to a previous point and restart the computation from there. To enable both parties to simulate different parts of the original $\Pi$ at the same time during $\Pi'$ (which is what happens when the parties are out of synch regarding what the current part of $\Pi$ to continue from is), the speaking order of the original $\Pi$ needs to be alternating (or have small period), in order for the robust protocol $\Pi'$ to have a fixed speaking order; 

Kol and Raz~\cite{kol2013interactive} also developed powerful min-entropy techniques to prove impossibility results for high-rate protocols with irregular enough (and even periodic protocols with a large enough period). Later, the impossibility results were refined and made more precise by several conjectures in Haeupler~\cite{H14} and the recent work by Efremenko, Kol, Paramonov, Saxena~\cite{efremenko2023rate}, which used min-entropy to formally prove a (large-alphabet version) of one of these conjectures, namely that without any restrictions on the input protocol being alternating (or periodic with low period) a rate approaching $1$ is impossible even against a single random erasure.

Inspired by the work by Kol and Raz~\cite{kol2013interactive}, Haeupler~\cite{H14} designed rate-efficient coding schemes that worked against adaptive adversaries and featured improved rates. With the same assumption on alternating protocols, Haeupler~\cite{H14} obtained rate $1-O(\sqrt{\eps})$ if the errors are selected by an oblivious adversary; and rate $1-O\big(\sqrt{\eps\log\log\tfrac1\eps}\big)$ if the errors are selected by an adaptive adversary. These remain the best protocols in this regime in terms of rate and are conjectured to be optimal~\cite{H14}. Haeupler~\cite{H14} also pointed out that as an alternative to considering alternating protocols one can allow for arbitrary original protocols $\Pi$ if one allows the robust simulation to have an adaptive speaking order (in the model of \cite{ghaffari2014optimal}, which is the second setting we consider).

Our work studying \emph{space-bounded} interactive coding schemes is most closely connected to the works \cite{HR18,chan2020small,EHKKRS23}, which we discussed above (see Table~\ref{table:prevwork}).
In all these works the protocols are designed such that the parties only require space $O(\log(s)\cdot \log |\Pi|)$.  
(We are not aware of other works studying such memory-bounded schemes, although we mention that the problem of space-bounded communication complexity was studied by Brody et al~\cite{brody2013space}.)

The protocol of~\cite{chan2020small} applies to the case of oblivious adversaries and achieves rate $1-O(\sqrt{\eps})$ (i.e., the conjectured to be optimal rate from~\cite{H14} for this setting). This work in fact applies to a more general scenario where a single ``Alice'' (i.e., a server) wishes to communicate with $m$ ``Bobs'' (i.e., $m$ clients); in this case the constants hidden in the big-$O$ notation for the space and rate depend on $m$. A main conceptual contribution of this work is to have the memory-bounded parties chain hashes of previous points they could rewind to, an approach we use as well. 

The protocol of~\cite{EHKKRS23} does work for arbitrary adversaries (as we consider), but achieves lower rate $1-O\big(\sqrt[3]{\eps\log1/\eps}\big)$.  We note that \cite{EHKKRS23} focuses on a communication model where both $\Pi$ and $\Pi'$ are alternating.  As noted in Remark~\ref{rem:speaking_order}, if $\Pi$ is alternating then our protocol $\Pi'$ can be made alternating, so our result improves the rate in that model.

The main result of \cite{EHKKRS23} is a ``compiler,'' which works in a black-box way, taking any (not-necessarily-space-efficient\footnote{In fact, it need not even have non-trivial time complexity.}) interactive coding scheme as input, and outputting a new interactive coding scheme that is space-efficient. The high-level strategy employed in this work is to apply a form of ``concatenation'' scheme. That is, the protocol is simulated in short blocks that are made noise-resilient by using an ``inner'' interactive coding scheme. These ``blocks'' can then be viewed as larger alphabet symbols.  Thus, the problem is essentially reduced to designing a noise-resilient protocol \emph{over a larger alphabet}. The compiler then essentially applies an ``outer'' interactive coding scheme that works over $\log|\Pi|$-bit alphabets; such a scheme is given in, e.g.,~\cite{H14}. The authors can choose the (necessarily deterministic) inner coding scheme to have rate $1-O(\sqrt{\eps\log1/\eps})$ (using a scheme of Cohen and Samocha~\cite{cohen2020palette}) and the outer coding scheme to have rate $1-O(\sqrt{\eps})$; this leads to a concatenated scheme with rate $1-O(\sqrt[3]{\eps\log1/\eps})$. 

Another closely related work is due to Efremenko et al~\cite{efremenko2022circuits}, which studies the task of constructing error resilient circuits. The approach taken therein is to translate this task into the problem of constructing space efficient interactive coding schemes for a certain, non-standard communication model.\footnote{This is a generalization of a result of Kalai, Rao and Lewko~\cite{kalai2012formulas}, which translated the problem of constructing robust \emph{formulas} into a certain interactive coding task. } In particular, the communication model gives a certain amount of ``feedback'' to the parties (i.e., they learn what the other party received, even if it was corrupted); however the adversary is given the additional power to tamper with the parties' memories. The space blow-up in this protocol is also $O(\log|\Pi|)$, as is the case for our protocol and those of \cite{chan2020small,EHKKRS23}. 

Lastly, given that all the above protocols blow up the space complexity by an $O(\log |\Pi|)$ factor, it is natural to wonder if that is indeed the best possible, or if it can be improved. As we discuss more in Section~\ref{sec:conclusion}, it seems like this blow-up factor is inherent, at least in the absence of substantially new ideas.

%% file: tech.tex
In this section, we give a technical overview of the proof of Theorem~\ref{thm:main}.  The starting point for our work is the protocol of \cite{H14}, so we describe it in Section~\ref{sec:tech1}.  Then in Section~\ref{sec:tech2}, we describe some challenges in adapting this protocol to be space-efficient, and how we overcome them. 
Finally, in Section~\ref{sec:tech3} we explain in more detail the structure of the proof, with some pointers to key lemmas.  We begin however in Section~\ref{sec:tech0} with an extremely high-level overview of what we view as the most interesting part of our analysis.

\subsection{Sneaky Attacks and a New Flavor of Analysis}\label{sec:tech0}
As mentioned in Section~\ref{sec:intro}, we view one of our main technical contributions to be a new flavor of analysis.  As we will see, adapting \cite{H14} to the small-space-setting opens us up to an attack, which we call a \textbf{sneaky attack}, discussed more below.  However, this attack ends up being an attack only on the ``standard'' potential function analysis, rather than on the scheme itself.  

In more detail, the analysis of \cite{H14} and other works proceed by analyzing a potential function $\Phi$.  Intuitively, $\Phi$ keeps track of how much progress Alice and Bob make, relative to how much error the adversary introduces.  A good rule of thumb for the success of the protocol, in the standard analysis, is the following:
\begin{quote}\textbf{Rule of Thumb:} If Alice and Bob have to backtrack a lot, then the adversary should have to introduce a commensurate number of errors.
\end{quote}
Since the adversary has a limited budget, this rule of thumb means that Alice and Bob don't have to backtrack a lot over the course of the protocol, and the rate is good.  If the potential function $\Phi$ stays ``well-behaved'' throughout the protocol, then this is rule of thumb is maintained; thus the traditional approach is to show that $\Phi$ is always well-behaved.

We do define a potential function $\Phi$, but it turns out that $\Phi$ may not stay well-behaved: If the adversary executes a \emph{sneaky attack}, then Alice and Bob may rewind a lot even if the adversary hasn't introduced many corruptions at the time.  However, we show that if such a sneaky attack occurs, the adversary must have ``set up'' for it by introducing corruptions at some \emph{other} time in the algorithm; and moreover that these corruptions can only be used for a single sneaky attack.  Thus, even though the potential function $\Phi$ may be poorly behaved sometimes, we can bound how much that happens.

To that end, our analysis involves two components.  First, we do the traditional potential function analysis.  Second, we augment that with an analysis that bounds the number of sneaky attacks that can occur throughout the execution of $\Pi'$, and argue that the only way that $\Phi$ is \emph{not} well-behaved is if a sneaky attack occurs.

To the best of our knowledge, this sort of proof structure is novel to our work.  We hope that this style of analysis may inspire simpler protocols going forward.

\subsection{The Protocol of \cite{H14}}\label{sec:tech1}
In order to describe our protocol in more detail, we first describe the protocol of \cite{H14}, which allows for high-rate robust interactive communication, but without small space.
At a high-level, the protocol of \cite{H14} falls within the ``rewind-if'' paradigm of~\cite{schulman1992communication}.  The basic idea is that Alice and Bob will simulate the protocol $\Pi$ in iterations, where each iteration consists of a \emph{Verification Phase}, a \emph{Computation Phase}, and a \emph{Transition Phase}.  During the Verification Phase, the two parties communicate with each other to make sure that their transcripts so far match.  If everything seems on track, then during the \emph{Computation Phase} they simulate $r$ rounds of $\Pi$.  Finally, in the \emph{Transition Phase}, if things do not seem to be on track, the parties may each ``rewind'' to an earlier place in $\Pi$ to attempt to get back in sync.

The instantiation of the ``rewind-if'' paradigm in \cite{H14} relied on a set of cleverly chosen \emph{meeting points}.  More precisely, for a given depth $\ell$ in the protocol $\Pi$, \cite{H14} defines two ``scale-$j$'' meeting points for $\ell$, $\MPone$ and $\MPtwo$, that are roughly $2^j$ iterations earlier than $\ell$ in $\Pi$.  If a party (say, Alice) decides that she needs to rewind, she will iterate over all the scales $j = 0, 1, 2, \ldots$, and for each scale she will attempt to see if she and Bob have a shared meeting point at scale $j$.  She will do this by using roughly $2^j$ rounds of ``voting'' for each of the scale-$j$ points.  Once Alice identifies a shared meeting point (and hopefully, so does Bob), they will rewind back to that point and begin simulating again. The reason that \cite{H14} uses \emph{two} meeting points is because, if these points are chosen carefully, it can be shown that most of the time Alice and Bob's sets of two meeting points will overlap, so they have a common meeting point to rewind to. The reason that they vote for $2^j$ rounds before rewinding to a point approximately $2^j$ away is to make sure that the adversary has to invest corruptions at roughly the same scale that Alice and Bob rewind.
This is the ``rule of thumb'' discussed in Section~\ref{sec:tech0}, and enables \cite{H14} to use the potential function analysis described above.

\subsection{Challenges in Making \cite{H14} Small-Space, and How We Overcome Them}\label{sec:tech2}
The approach of \cite{H14} may require much more space than $\Pi$.  The reason is that Alice and Bob need to store the whole history of what they have simulated so far, rather than just the state of $\Pi$ that they are in.  This leads to our zero'th attempt to adapt the protocol of \cite{H14} to small space: Instead of storing the whole history, Alice and Bob each just store the state $v$ of $\Pi$ that they are meant to be at; call these states $v_A$ and $v_B$, respectively.  This zero'th attempt immediately runs into several challenges.  We discuss five of them in turn, and explain how we overcome them in our protocol.  Briefly, the challenges are the following:
\begin{itemize}
    \item \textbf{Challenges 1 and 2} have to do with the basic mechanics of deciding when to rewind and where to rewind to.  These also come up and were resolved in \cite{EHKKRS23}, and we resolve them in a similar (though not identical) way.  We discuss these only briefly below.
    \item After resolving Challenges 1 and 2, we end up with a protocol that is susceptible to \textbf{Challenge~3}, which is that Alice and Bob may rewind too much, resulting in a bad rate.  To address this, we do have to modify the protocol further, which we describe below.
    \item After resolving Challenge 3, we are left with \textbf{Challenge 4}, which is the \emph{sneaky attack} alluded to above in Section~\ref{sec:tech0}.  This is not actually a problem with the protocol but rather a problem with the analysis, so we describe how to fix the analysis.
    \item There are no more challenges.  Proving this is a challenge in and of itself, which we discuss below as \textbf{Challenge 5}.
\end{itemize}

\paragraph{Challenge 1: Deciding if the parties are ``in sync''.}  In \cite{H14}, the parties check if they are ``in sync'' during the Verification phase by exchanging hashes of their simulated transcripts.  However, since Alice and Bob no longer remember their transcripts, they can no longer do this.  A first attempt would be for Alice and Bob to simply send hashes of $v_A$ and $v_B$, their current states in $\Pi$, instead of hashes of the entire transcripts. However, this does not work: it is possible that the adversary could maneuver Alice and Bob into a situation where they agree on the state $v$ in $\Pi$, but that $v$ is not the correct state.  If that happens, then Alice and Bob would happily continue on simulating---in sync, but down the wrong path---while the adversary sits back and laughs.\footnote{We note that laughing is not a necessary part of this attack; the important thing is that the adversary does not need to invest in any more corruptions.}  Instead, we use a similar solution as in~\cite{chan2020small,EHKKRS23}, which is \emph{hash chaining}. That is, along with maintaining their current states $v_A$ and $v_B$ in $\Pi$, Alice and Bob maintain a hash value $H$.  
Every so often, they update $H$ by hashing together their most recently simulated rounds of $\Pi$, along with the \emph{previous} value of $H$ (as well as some additional information). This way, Alice and Bob have some record of the entire past without having to store the whole transcript.

However there is one more wrinkle.  Because the purpose of this chained hash $H$ is to maintain a history of the whole protocol, it can \emph{never} suffer a hash collision.  If it did even once, Alice and Bob could be off-track for the rest of the protocol.  In order to guarantee that the probability of a hash collision is sufficiently small, the output of this chained hash must be relatively large; for that reason we call it the \emph{big} hash.  In more detail, the parameters are such that Alice and Bob have space to \emph{store} one copy of $H$ at a time, but don't have bandwidth to \emph{exchange} it and compare their copies.  If they did that, then the rate of the protocol would not match the desired rate \eqref{eq:desiredrate}.  
Instead, Alice and Bob also employ an additional \emph{small} hash, which has a much higher chance of collision and which they do exchange.   This small hash is used to exchange information both about their current state, as well as their copies of $H$, the randomness used to generate it, and when that randomness was itself generated.  By carefully choosing what information is hashed and exchanged, and controlling the number of small hash collisions, we are able to guarantee that Alice and Bob stay on track with high probability.
 We remark that \cite{EHKKRS23} don't face this challenge, as the simulation is done in chunks of length $\Theta(\log |\Pi|)$, so they can afford to send hashes of length $\Theta(\log|\Pi|)$ while maintaining rate close to $1$. However, we simulate in chunks of length $O_\eps(1)$ in order to obtain the rate \eqref{eq:desiredrate}, so we can only afford to exchange hashes of constant length, forcing us to use these small hashes. (We remark such small hashes are also used in \cite{H14} and other protocols achieving \eqref{eq:desiredrate}, such as \cite{chan2020small}.)

\paragraph{Challenge 2: Remembering the ``meeting points'' to jump back to.}
In \cite{H14}, Alice and Bob can remember the whole histories, so once they decide that they want to rewind to an earlier meeting point, they remember all the state information associated with it.  However, in our setting Alice and Bob cannot remember every possible meeting point that they might like to jump to.  Instead, we take the same approach as \cite{chan2020small,EHKKRS23} and allow Alice and Bob to remember only about $\log|\Pi|$ meeting points (and the associated state information) at a time. These meeting points are chosen to be roughly geometrically spaced apart. If a party (say Alice) simulates very deep in the protocol $\Pi$, she will have a fairly spotty memory of the earlier part of her simulation, and a more robust memory of the more recent part.

\paragraph{Challenge 3: Spotty memories can lead to over-jumping.}  As mentioned above, because Alice and Bob keep adding meeting points as they simulate $\Pi$, they need to forget meeting points from earlier iterations of $\Pi$.  This causes spotty memory of earlier in the simulation.  In particular, if we use the two-meeting-point voting system of \cite{H14}, something like the process pictured in Figure~\ref{fig:attack1} could occur.

\begin{figure}[h!]
\begin{center}
\begin{tikzpicture}[yscale=.5]
    \begin{scope}
        \draw[red, ->, thick] (0,0) to (0,-5.2);
        \draw[blue,->, thick, dashed] (.4, 0) to (0.4, -5.2);
        \foreach \x in {0,-3,-4, -4.5,-4.75}{
        \draw[red,fill=red] (0,\x) circle (0.1);
        \draw[blue,fill=blue] (.4,\x) circle (0.1);
        }
        \node[anchor=south] at (0.2, .3) {\begin{minipage}{3cm} \begin{center}\footnotesize 1. Alice and Bob are simulating $\Pi$ in sync with each other.  They remember very few meeting points toward the beginning of the protocol.\end{center}\end{minipage}};
    \end{scope}
        \begin{scope}[xshift=3.5cm]
        \draw[red, thick] (0,0) to (0,-5);
        \draw[blue, thick, dashed] (.4, 0) to (0.4, -5);

        \draw[red, thick](0,-5) to (-.5, -5.5);
        \draw[red, thick, ->] (-.5, -5.5) to (-.5, -10);
        \draw[blue, thick, dashed ](.4,-5) to (.9, -5.5);
        \draw[blue, dashed, thick, ->] (.9, -5.5) to (.9, -10);
              \foreach \x in {0,-4}{
        \draw[red,fill=red] (0,\x) circle (0.1);
        \draw[blue,fill=blue] (.4,\x) circle (0.1);
        }
        \foreach \x in {-6, -8, -9, -9.5}{
        \draw[red,fill=red] (-.5,\x) circle (0.1);
        \draw[blue,fill=blue] (.9,\x) circle (0.1);
        }
        \node[anchor=south] at (0.2, .3) {\begin{minipage}{3cm} \begin{center}\footnotesize 2. The adversary causes Alice and Bob to go on different paths for a while.  Now they remember even fewer points toward the beginning of the protocol.\end{center}\end{minipage}};
    \end{scope}
           \begin{scope}[xshift=7cm]
        \draw[red, thick] (0,0) to (0,-4);
        \draw[blue, thick, dashed] (.4, 0) to (0.4, -4);

              \foreach \x in {0,-4}{
        \draw[red,fill=red] (0,\x) circle (0.1);
        \draw[blue,fill=blue] (.4,\x) circle (0.1);
        }

        \node[anchor=south] at (0.2, .3) {\begin{minipage}{3cm} \begin{center}\footnotesize 3. Alice and Bob catch the error and rewind to resolve it.\end{center}\end{minipage}};
    \end{scope}

            \begin{scope}[xshift=10.5cm]
        \draw[red, thick] (0,0) to (0,-4);
        \draw[blue, thick, dashed] (.4, 0) to (0.4, -4);

        \draw[red, thick](0,-4) to (-.5, -4.5);
        \draw[red, thick, ->] (-.5, -4.5) to (-.5, -5);
        \draw[blue, thick, dashed ](.4,-4) to (.9, -4.5);
        \draw[blue, dashed, thick, ->] (.9, -4.5) to (.9, -5);
              \foreach \x in {0,-4}{
        \draw[red,fill=red] (0,\x) circle (0.1);
        \draw[blue,fill=blue] (.4,\x) circle (0.1);
        }

        \node[anchor=south] at (0.2, .3) {\begin{minipage}{3cm} \begin{center}\footnotesize 4. The adversary again makes Alice and Bob diverge, but just a little bit.\end{center}\end{minipage}};
    \end{scope}

                \begin{scope}[xshift=14cm]

              \foreach \x in {0}{
        \draw[red,fill=red] (0,\x) circle (0.1);
        \draw[blue,fill=blue] (.4,\x) circle (0.1);
        }

        \node[anchor=south] at (0.2, .3) {\begin{minipage}{3cm} \begin{center}\footnotesize 5. Alice and Bob catch the error, and rewind.  The adversary fools them into skipping the only close meeting point, and then they rewind a lot!\end{center}\end{minipage}};
    \end{scope}
\end{tikzpicture}
\end{center}
\caption{Steps in an attack that the adversary might make when Alice and Bob use only $\MPone$ and $\MPtwo$ as in \cite{H14}, exploiting Alice and Bob's limited memory.  The red solid line represents Alice's simulated path; the blue dashed line represents Bob's simulated path.  The ovals represent the meeting points that Alice and Bob remember.  In steps 2 and 3, the adversary invests an amount of corruptions that is commensurate with the amount that Alice and Bob rewind, which is good.  But in steps 4 and 5, the adversary invests few corruptions, and can make Alice and Bob rewind very far back, which is not good.}\label{fig:attack1}
\end{figure}
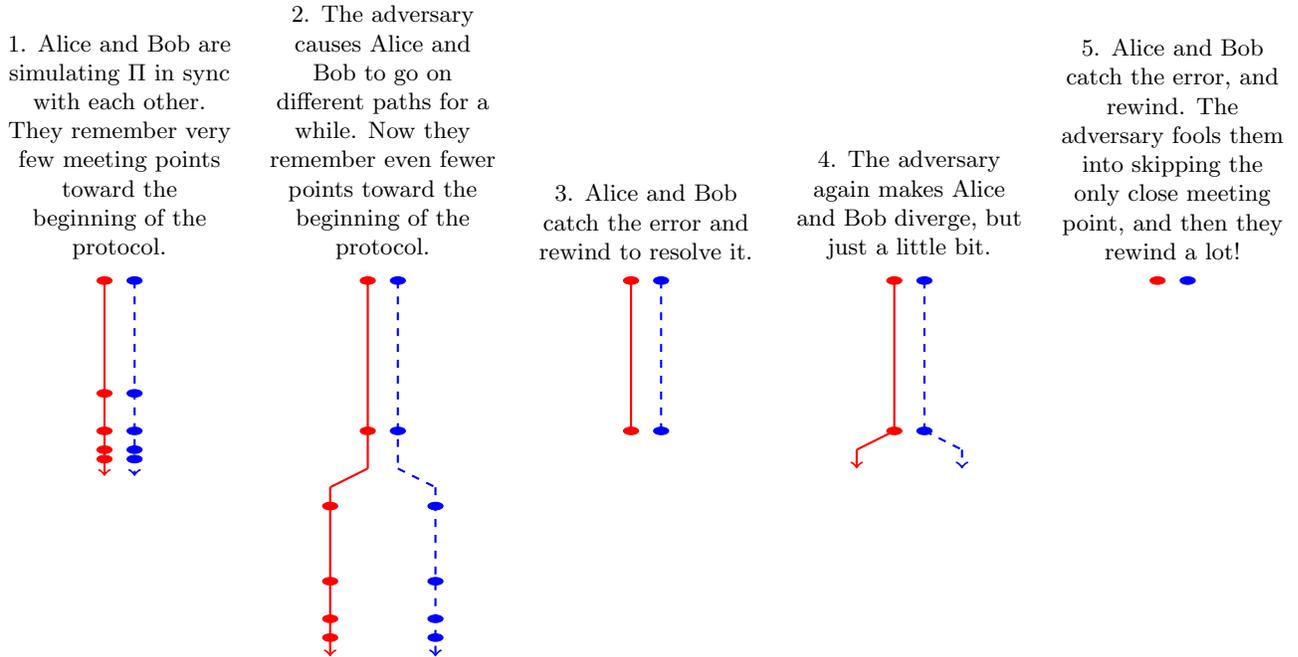

What goes wrong in Figure~\ref{fig:attack1} is that the ``rule of thumb'' discussed in Section~\ref{sec:tech0} is violated: It is possible for the adversary to make Alice and Bob rewind a lot by introducing relatively few corruptions.  In particular, after Alice and Bob have rewound once (at the start of Step 3 in Figure~\ref{fig:attack1}), because they have forgotten most of their meeting points from early in the protocol, even a small amount of corruption (Step 4) will make them rewind a very long way back (Step 5).   
Further, the adversary can repeat this attack without investing many additional corruptions.  Indeed, imagine that the timelines in Figure~\ref{fig:attack1} go further up the page; the adversary could use very few corruptions after Step 5 (essentially repeating Step 4), to cause them to rewind even higher.  Thus, the adversary can cause Alice and Bob to rewind back to the beginning of the protocol, while using very few corruptions.  

Because this attack can be repeated as described above, it's a problem with the protocol, not just the analysis.
Our solution is to modify the protocol to add a \emph{third} meeting point, which we (creatively) call  $\MPthree$.  This third meeting point may be close to Alice and Bob, but as far as the voting scheme is concerned, it ``counts'' like a meeting point that is much further away, in the sense that Alice and Bob will take many more votes before deciding to rewind to it.  Formally, if Alice and Bob are currently voting on candidates $\MPone$ and $\MPtwo$ that are approximately $2^j$ iterations back, they will also throw in the point $\MPthree$, which is the deepest meeting point that they remember that is divisible by $2^j$; but they will still vote $2^j$ times on $\MPthree$, just as they do with $\MPone$ and $\MPtwo$.  Thus, $\MPthree$ acts as a barrier to prevent an attack like the one in Figure~\ref{fig:attack1} from occurring: In Step 4 of Figure~\ref{fig:attack1}, the point that it was previously cheap for the adversary to make Alice and Bob skip will now become very expensive.

\paragraph{Challenge 4: Sneaky attacks: Spotty memories can \emph{still} lead to over-jumping.}  Adding $\MPthree$ fixes the problem in Figure~\ref{fig:attack1}.  However, there is still a way that the adversary can make Alice and Bob rewind very far back without introducing too many corruptions.  We refer to this as a \emph{\protosneakyf}, as it essentially allows the adversary to ``sneak'' past the barrier $\MPthree$.  This attack is shown in Figure~\ref{fig:sneaky}.  The basic idea is that (in Step 2 of Figure~\ref{fig:sneaky}), the adversary can drive one party (say, Alice) very far down in the protocol, so that she forgets most of her early meeting points.  But instead of driving Bob down as well, the adversary just stalls Bob; Bob knows that something is wrong, but Alice does not realize that yet, and Bob waits for Alice to figure it out.  When Alice finally realizes that something is amiss, she and Bob finally do rewind to a shared meeting point (Step 3).  At this point, Bob's memory is much better than Alice's.  Next (Step 4), the adversary tricks Bob into thinking that something is wrong, and he makes a short rewind to a point that Alice doesn't remember.  The adversary did not have to pay very much to make this happen, because Bob's rewind was short, and (except for skipping the special $\MPthree$ points), the adversary does not have to pay very much to make Alice or Bob rewind a short way.  But at this point, we are in trouble: Bob may have rewound past the barrier point $\MPthree$!  Then even though this barrier point is available for Alice to rewind to, it is no longer available for Bob.  In the next step (Step 5), Alice and Bob may then rewind quite a long way, even though the adversary did not invest many corruptions.
In this way, a \protosneaky ``{sneaks}'' around the barrier point $\MPthree$.

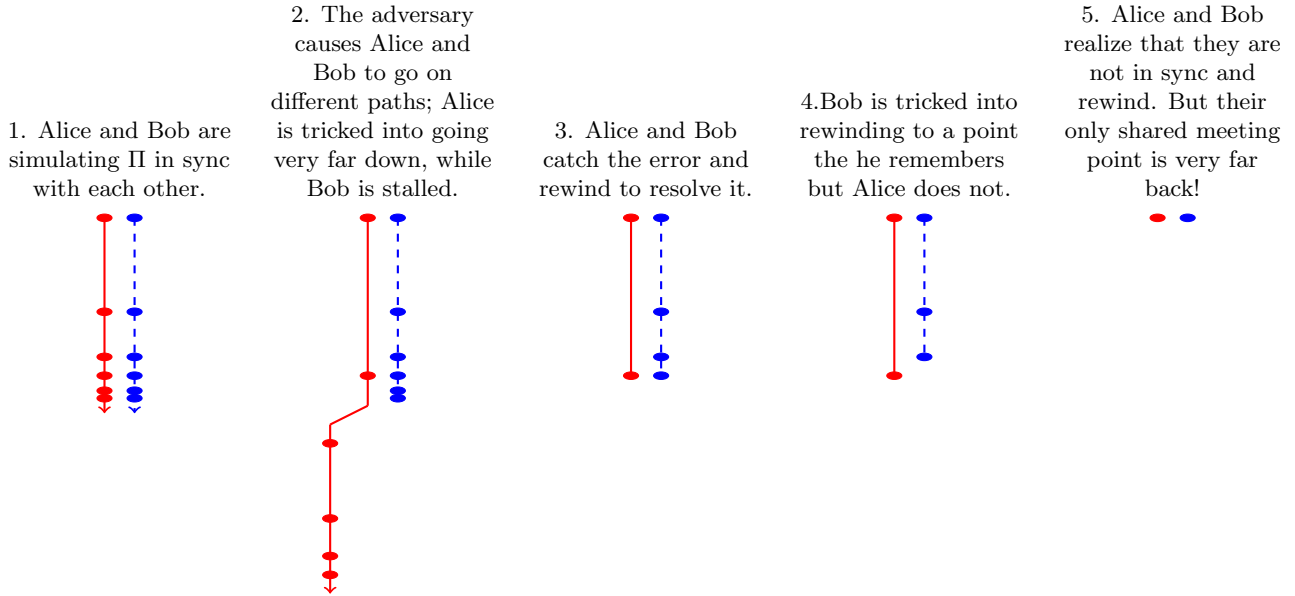
\begin{figure}[h!]
\begin{center}
\begin{tikzpicture}[yscale=.5]
    \begin{scope}
        \draw[red, ->, thick] (0,0) to (0,-5.2);
        \draw[blue,->, thick, dashed] (.4, 0) to (0.4, -5.2);
        \foreach \x in {0,-2.5,-3.7, -4.2, -4.6,-4.8}{
        \draw[red,fill=red] (0,\x) circle (0.1);
        \draw[blue,fill=blue] (.4,\x) circle (0.1);
        }
        \node[anchor=south] at (0.2, .3) {\begin{minipage}{3cm} \begin{center}\footnotesize 1. Alice and Bob are simulating $\Pi$ in sync with each other. \end{center}\end{minipage}};
    \end{scope}
    
        \begin{scope}[xshift=3.5cm]
        \draw[red, thick] (0,0) to (0,-5);
        \draw[blue, thick, dashed] (.4, 0) to (0.4, -5);

        \draw[red, thick](0,-5) to (-.5, -5.5);
        \draw[red, thick, ->] (-.5, -5.5) to (-.5, -10);

              \foreach \x in {0,-4.2}{
        \draw[red,fill=red] (0,\x) circle (0.1);
        }
        \foreach \x in {-6, -8, -9, -9.5}{
        \draw[red,fill=red] (-.5,\x) circle (0.1);
        }
        \foreach \x in {0,-2.5,-3.7, -4.2, -4.6,-4.8}{
        \draw[blue,fill=blue] (.4,\x) circle (0.1);
        }

        \node[anchor=south] at (0.2, .3) {\begin{minipage}{3cm} \begin{center}\footnotesize 2. The adversary causes Alice and Bob to go on different paths; Alice is tricked into going very far down, while Bob is stalled.\end{center}\end{minipage}};
    \end{scope}
           \begin{scope}[xshift=7cm]
        \draw[red, thick] (0,0) to (0,-4.2);
        \draw[blue, thick, dashed] (.4, 0) to (0.4, -4.2);

              \foreach \x in {0,-4.2}{
        \draw[red,fill=red] (0,\x) circle (0.1);
        \draw[blue,fill=blue] (.4,\x) circle (0.1);
        }
        \draw[blue,fill=blue] (.4,-2.5) circle (0.1);
        \draw[blue,fill=blue] (.4,-3.7) circle (0.1);

        \node[anchor=south] at (0.2, .3) {\begin{minipage}{3cm} \begin{center}\footnotesize 3. Alice and Bob catch the error and rewind to resolve it.\end{center}\end{minipage}};
    \end{scope}

            \begin{scope}[xshift=10.5cm]
        \draw[red, thick] (0,0) to (0,-4.2);
        \draw[blue, thick, dashed] (.4, 0) to (0.4, -3.7);

              \foreach \x in {0}{
        \draw[red,fill=red] (0,\x) circle (0.1);
        \draw[blue,fill=blue] (.4,\x) circle (0.1);
        }
        \draw[blue,fill=blue] (.4,-2.5) circle (0.1);
        \draw[blue,fill=blue] (.4,-3.7) circle (0.1);
        \draw[red,fill=red] (0,-4.2) circle (0.1);

        \node[anchor=south] at (0.2, .3) {\begin{minipage}{3cm} \begin{center}\footnotesize 4.Bob is tricked into rewinding to a point the he remembers but Alice does not.\end{center}\end{minipage}};
    \end{scope}

                \begin{scope}[xshift=14cm]

              \foreach \x in {0}{
        \draw[red,fill=red] (0,\x) circle (0.1);
        \draw[blue,fill=blue] (.4,\x) circle (0.1);
        }

        \node[anchor=south] at (0.2, .3) {\begin{minipage}{3cm} \begin{center}\footnotesize 5. Alice and Bob realize that they are not in sync and rewind.  But their only shared meeting point is very far back!\end{center}\end{minipage}};
    \end{scope}
\end{tikzpicture}
\end{center}
\caption{Steps in a \protosneakyf. The red solid line represents Alice's simulated path; the blue dashed line represents Bob's simulated path.  The ovals represent the meeting points that Alice and Bob remember. In Steps 2 and 3, the adversary invests an amount of corruptions that is commensurate with the amount that Alice has to rewind, which is good.  But in Step 4, the adversary only has to invest a few corruptions (enough to make Bob think that he and Alice are out of sync and make a short rewind); and then in Step 5, Alice and Bob may rewind a lot. }\label{fig:sneaky}
\end{figure}

However, as mentioned above, a sneaky attack is actually an attack on the potential-function-based analysis, \emph{not} on the protocol itself.  In more detail, we show that, unlike the attack demonstrated in Figure~\ref{fig:attack1}, a \protosneaky is not repeatable.  Every time the adversary perpetrates a \protosneakyf, they have to invest the corruption in Step 2 (needed to make Alice go so far down in her simulation).  We show that, unlike in Figure~\ref{fig:attack1}, the adversary cannot make the Step 2 corruptions once, and then use them repeatedly to perpetrate many \protosneakyf{s}; this means that the ``damage'' done by sneaky attacks can be bounded.  

Thus, our approach is the following.  We define a potential function $\Phi$ similar to the potential function studied in \cite{H14}.  However, instead arguing that $\Phi$ is \emph{always} well-behaved (which is not true, because when the final rewind---Step 5---of a \protosneaky happens, $\Phi$ may behave very poorly), we argue that $\Phi$ is well-behaved in all of the iterations where a \protosneaky does not occur.  Then in a separate argument we show that, because the adversary must pay for each \protosneaky individually, not too many \protosneakyf{s} can occur, and over the course of the whole simulation, $\Phi$ will still make progress.

\paragraph{Challenge 5: There is no Challenge 5.}

At this point, the reader may reasonably ask where all this is going.  Are we going to keep introducing ad hoc fixes until we can't think of any more attacks?  Fear not!  \textbf{In our analysis, we show that a \protosneaky is the \emph{only} way that the adversary can make Alice and Bob rewind a lot by investing in relatively few corruptions.}  

In Lemma~\ref{lem:maintech} (our main technical lemma), we essentially show that if Alice and Bob's counters $j$ get very large (meaning that they are about to rewind a long way), then \emph{either} the adversary had to invest a lot of corruptions; \emph{or} a \protosneaky is in progress.  This essentially says that the only exception to our intuition that ``the adversary should pay for long rewinds'' comes from \protosneakyf{s}.

\subsection{Overview of the Proof Structure}\label{sec:tech3}
Formalizing the above intuition requires delicate analysis.  Here, we give a high-level idea of the structure of the proof, with pointers to some key sections/lemmas.

\paragraph{The potential function.}
 
We begin by defining our potential function, called $\Phi$, in \eqref{eq:potential} in Section~\ref{subsec:def-lem}.  This potential function is quite similar to the one from~\cite{H14}.  Intuitively, $\Phi$ increases when things are going well, and if it becomes large enough then we can conclude that Alice and Bob have succeeded in simulating $\Pi$.   Thus, our goal is to show that it increases.

\paragraph{Sneaky attacks, and how they are ``paid'' for.}
After proving a few auxiliary lemmas in Section~\ref{sec:usefullemmas},  we formally define a \protosneaky in Section~\ref{sec:maintech}, and prove a few useful statements about it.  These include our formalization of the intuition that ``the adversary has to pay for each sneaky attack individually.'' To that end, we define different parts of a sneaky attack: a \emph{diving window}, the time window depicted in Step 2 of Figure~\ref{fig:sneaky}, where Alice is diving down; and a \emph{voting window}, the time window depicted in Step 3 of Figure~\ref{fig:sneaky}, where Alice and Bob are voting where to go next.  We show that all of these windows are disjoint, even \emph{between} sneaky attacks.  Since these windows are where the adversary introduces corruptions, this means that the adversary has to pay separately for each sneaky attack.

\paragraph{Main technical lemma: sneaky attacks are the only thing that go wrong.}
In Section~\ref{sec:maintech}, we also give the statement of our main technical lemma (Lemma~\ref{lem:maintech}), which essentially says that the only thing that can go wrong is a \protosneakyf.  The proof of the main technical lemma is rather long and technical (as the name suggests), and we defer the proof to Section~\ref{sec:grossproof}.

 \paragraph{The potential function is well-behaved, except during sneaky attacks.} In Section~\ref{sec:progressProof}, we prove that the potential function $\Phi$ is well-behaved, except during a \protosneakyf.  In more detail, Lemma~\ref{lem:progress} shows that, \emph{unless} a \protosneaky occurs, either $\Phi$ increases by $1$, or else there was a corruption or hash collision, in which case it decreases by a controllable amount.  Since the number of corruptions is bounded, and hash collisions are unlikely, this implies that $\Phi$ makes steady progress except for sneaky attacks.

\paragraph{Upper and lower bounds on $\Phi$.}
Next, in Sections~\ref{sec:upperBd} and \ref{sec:lowerBound}, we prove upper and lower bounds on $\Phi$, respectively.  The lower bound (Lemma~\ref{lem:philb}) captures the fact that progress is being made: it uses the fact that $\Phi$ makes steady progress except for sneaky attacks; and then uses the earlier analysis (that the adversary has to ``pay'' for sneaky attacks separately) in order to bound the aggregate damage across all the \protosneakyf{s}.  
The upper bound (Lemma~\ref{lem:phiub}) captures the fact that ``progress'' (in terms of $\Phi$ getting large) actually means ``progress'' (in terms of Alice and Bob succeeding).  

\paragraph{Putting it all together.}
We take a brief detour in Section~\ref{sec:hash} to prove some lemmas about how (in)frequently hash collisions occur, and then finally put it all together to prove Theorem~\ref{thm:main} in Section~\ref{sec:done}.

\paragraph{Other organization.}
After a quick discussion of open problems and future directions, in Section~\ref{sec:prelim}, we give some definitions and set up notation.  In that section we also define the meeting points that we will use, and prove some basic facts about them.
In Section~\ref{sec:protocol}, we present our protocol $\Pi'$.  The actual protocol (Algorithm~\ref{alg:adaptive}) is quite long; as mentioned above, the details are delicate, and as such a fair amount of of bookkeeping is required in the final protocol.  To that end, we give a simplified high-level version of the protocol in Algorithm~\ref{alg:adaptive_short}, which the reader may want to go over first before diving into the details.

\subsection{Open Problems and Future Directions}\label{sec:conclusion}
Before we get into the details, we record a few open problems.  

An obvious question is whether the space bound of $O(\log s \cdot \log|\Pi|)$ is optimal.  We conjecture that it is, at least in the context of rewind-if schemes.
In more detail, the $O(\log|\Pi|)$ blow-up in the space complexity comes from needing to remember $O(\log|\Pi|)$ meeting points. 
In any rewind-if approach, it seems infeasible to design a robust protocol only remembering $o(\log |\Pi|)$ meeting points. Indeed, in such a case, if the parties Alice and Bob are at depth $p$, then for any large constant $C$ there would exist a length $\ell$ for which the parties have no saved meeting points between depth $p-\ell$ and depth $p-C\ell$. Then, if the adversary corrupts all rounds between $p-\ell$ and $p$, the parties would need to rewind to the closest correct meeting point, which is depth less than $p-C\ell$. That is, with $\ell$ corruptions the adversary causes $C\ell$ wasted steps, implying that the protocol cannot handle more than a $1/C$ fraction of corruptions. As $C$ was arbitrary, this shows we could not hope for a positive error rate. Such considerations also motivate the $O(\log |\Pi|)$-blow-ups in \cite{chan2020small,EHKKRS23}.

Finally, we want to highlight the question of applications.  We believe that space-bounded interactive protocols are well-motivated from a practical perspective alone: for example, as elucidated in~\cite{chan2020small}, particularly in client-server models, weak clients may not be able to afford massive blow-ups in space complexity in robustifying a protocol. However, there is also potential for other applications. For example, as discussed above,~\cite{efremenko2022circuits} demonstrated that interactive coding schemes can be used to obtain circuits resilient to short-circuit errors, and that the space-complexity of the interactive coding scheme is directly tied to the size of the resulting resilient circuit. This is done via an adaptation of the Karchmer-Wigderson transformation. While we do caveat that the model of noise which arises via this transformation is incomparable to ours (the parties get feedback, making it easier, but the adversary may also interfere with the parties' memories, making it harder), we believe that any tools and techniques developed in the context of space-efficient interactive coding could prove effective in this area.

%% file: prelim.tex
In this section, we formally define our model and introduce definitions and notation that we will need to define and analyze our protocol.

We begin by giving the formal model for the original protocol $\Pi$ (Section~\ref{sec:orig}), and for our adaptive adversary (Section~\ref{sec:adv}).
In Sections~\ref{sec:robust}, \ref{sec:subscripts} and \ref{sec:MPs}, we outline the structure of our robust protocol $\Pi'$ and introduce several definitions we will use.  In Section~\ref{sec:hashprelims}, we introduce some hash functions and relevant theorems that we will need.

Throughout, we consider two parties, Alice and Bob, simulating a protocol $\Pi$. We assume that Alice and Bob have private randomness.

\subsection{The original protocol $\Pi$} \label{sec:orig}
We begin with a formal definition of the original protocol $\Pi$.
Formally, the protocol $\Pi$ is represented by a rooted directed acyclic graph (DAG), with a designated root node \texttt{start} and a designated set of terminal nodes that have no outgoing edges, along with two \defn{transition functions} $\tau_A$ and $\tau_B$ discussed more below. 
Each node $v$ in $\Pi$ is a \defn{state}; throughout the paper, we let $s$ denote the number of states in $\Pi$.  Each state (except for the terminal nodes) has two outgoing edges, labeled ``$0$'' and ``$1$.''
Each state \defn{belongs} to exactly one of Alice or Bob; let $\mathcal{A}$ be the set of states owned by Alice and $\mathcal{B}$ be the set of states owned by Bob. 
Then the transition functions $\tau_A$ and $\tau_B$ map each party's states to a bit: $\tau_A: \mathcal{A} \to \{0,1\}$ and $\tau_B: \mathcal{B} \to \{0,1\}$.  
We interpret the transition functions as instructions for what Alice and Bob should do in each state.

With this notation, we can view $\Pi$ as a \defn{pebble game} for Alice and Bob on the underlying DAG.  That is, we imagine that there is a pebble that begins on the node \texttt{start}.  At each timestep, the party who owns the pebble's state applies their transition function to decide what to transmit, and moves the pebble accordingly.  For example, if the pebble is currently on state $v \in \mathcal{A}$, then Alice will compute $b = \tau_A(v)$.  She will then send the bit $b$ to Bob, and she will move the pebble to the $b$'th child of $v$.  The game ends when the pebble reaches a terminal node.  At this point, the transition functions $\tau_A$ and $\tau_B$ encode what Alice and/or Bob should output at the end of the protocol.

We assume that Alice and Bob have oracle access to these transition functions and the structure of the DAG, and that these don't count towards their space complexity.\footnote{If the reader would like the transition functions and underlying DAG to count towards the space complexity in the model, then our results still hold: If the transition functions and DAG description require an additional $s'$ bits of space, then our robust protocol $\Pi'$ will take an additional $s'$ bits of space as well, for a total of $s' + O(\log s \cdot \log|\Pi|)$.  In particular, the space overhead is still at most a multiplicative factor of $O(\log |\Pi|)$.}
Thus, the total space that is needed for the protocol is $\log(s)$, the amount of space needed to remember which state the pebble is on.

 Alice's \defn{transcript} of $\Pi$ is a list of all of the bits that she sends and receives. 
The \defn{depth} of a state $v$ in $\Pi$ is its depth in the DAG from \texttt{start}.  The \defn{depth} of $\Pi$ is the depth of the deepest state in $\Pi$.
Throughout this paper, we use $d$ to denote the depth of $\Pi$.  
 Notice that the communication complexity $|\Pi|$ of $\Pi$ is equal to the depth $d$.

\begin{remark}[Where are the inputs?]
In the definition above, there are no explicit inputs for Alice and Bob, even though we often think of the protocol $\Pi$ as computing functions of such inputs. 
However, we can think of the inputs as being hard-coded into the transition function.  (That is, Alice knows $\tau_A$ and her input, but not $\tau_B$ or Bob's input, and vice versa; so we may as well collapse $\tau_A$ and Alice's input into one function, and vice versa).

This simplifies the notation and does not affect our arguments or results.
    This input-free definition is perhaps reminiscent of the model of (read-once) branching programs in the context of randomized algorithms using low-space: there, for fixed inputs the algorithm execution can also be viewed as a DAG, and the transitions are now only determined by the random coins (and not the inputs).
\end{remark} 

\begin{remark}[What about speaking order?]\label{rem:speaking_order_redux}
    In the Introduction (Remark~\ref{rem:speaking_order}), we remarked that the speaking order of $\Pi$ is not necessarily fixed in the model above.  That is, whether or not Alice or Bob speak at a round is determined by what has happened so far in the protocol, not just the transcript of that round.  This is captured in the DAG model, as there is no stipulation about which state is owned by which party.
    In the original (noiseless) protocol $\Pi$, exactly one of Alice and Bob will speak at any given round.  However, when we move to the noisy setting, it might be that Alice and Bob are no longer on the same page about which state of $\Pi$ they are simulating, and so it could be that both or neither of them speak at once.   We address what happens in this case in the next section when we describe the power of the adversary.

    However, as noted in Remark~\ref{rem:speaking_order}, if the original protocol $\Pi$ has \emph{alternating} speaking order, then our robust protocol $\Pi'$ will also have fixed speaking order.  Moreover, it is simple to modify $\Pi'$ so that it is alternating, without changing the rate asymptotically.  We address this more in Remark~\ref{rem:alternating} after we present our robust protocol.
\end{remark}

\subsection{The adversary}\label{sec:adv}
We consider a \defn{adaptive adversary} with a \defn{corruption-fraction budget} $\eps$.  This means that the adversary may flip any bit transmitted in either direction, adaptively (based on the transcript of $\Pi'$ so far), provided that when $\Pi'$ has finished, at most an $\epsilon$-fraction of the transmitted bits were flipped.  When the adversary chooses to flip a bit, we call this a \defn{corruption} in that round.  We assume that the adversary knows (a bound on) $|\Pi'|$.

As mentioned above, we do not assume a fixed speaking order for $\Pi$, and we will in general not have a fixed speaking order in $\Pi'$.  Thus, it is possible that in $\Pi'$, neither or both of Alice and Bob speak during a round.
In our model (which is the same as the model introduced in \cite{ghaffari2014optimal} and used in \cite{H14}), 
if neither Alice nor Bob speak during a round, then the adversary may make Alice or Bob hear anything, without counting towards its corruption-fraction budget.
Furthermore, if both parties speak in a round, then neither party receives a transmission (which is natural, as neither party is listening), and this also does not count towards the adversary's corruption budget. 

\subsection{The robust protocol $\Pi'$}\label{sec:robust}  In order to protect $\Pi$ from errors, our main algorithm (Algorithm~\ref{alg:adaptive}) transforms $\Pi$ into a robust protocol $\Pi'$.

The new protocol $\Pi'$ will proceed in a series of $\Btotal$ \defn{blocks}.  Each block consists of $\Iblock$ \defn{iterations}.  The total number of iterations is $\Itotal = \Btotal \cdot \Iblock$.  In each iteration, each party will simulate $r$ \defn{rounds} of the original protocol (or perform $r$ ``dummy rounds'' in some cases);  will do some bookkeeping; and then will either rewind to a previous point in their simulated path (in multiples of $r$ rounds) or will stay where they are.

We will use the following terminology for the elements of this process.
\begin{itemize}
\item The \defn{simulated path for Alice}, denoted $\simPath_A$, is the full transcript that Alice has simulated so far in $\Pi$, along with the iteration numbers in $\Pi'$ that she simulated them.  For example, suppose Alice simulates $5r$ rounds $(\sigma_1, \sigma_2, \sigma_3, \sigma_4, \sigma_5)$ over five iterations $I = 1,2,3,4,5$, where each $\sigma_i$ is the transcript of $r$ rounds of $\Phi$.  If she then jumps back to the point in her simulation after $3r$ rounds, or $3$ iterations, her simulated path would be 
\[ \simPath_A = (\sigma_1 \circ 1, \sigma_2 \circ 2, \sigma_3 \circ 3),\] 
where ``$\circ$'' denotes concatenation.  If she then went on to re-simulate $\sigma_4'$ in iteration $I=6$, her simulated path would be 
\[ \simPath_A = (\sigma_1 \circ 1, \sigma_2 \circ 2, \sigma_3 \circ 3,  \sigma'_4 \circ 6)\]
We define the \defn{simulated path for Bob}, denoted $\simPath_B$, analogously.  

 Because everything in the algorithm proceeds in increments of $r$ rounds, we will treat the simulated path as a sequence of \emph{iterations}, rather than of rounds, as in the example above.

 \item When the protocol $\Pi'$ is in a particular block, we define $\Bsim_A$ and $\Bsim_B$ as the substrings of $\simPath_A$ and $\simPath_B$, respectively, corresponding to the iterations in that current block.  We use $\Bsim$ (with no subscript) to refer to this substring for an arbitrary party.

\item We define $\Aiter$ (resp. $\Biter$) to denote the length of Alice's (resp. Bob's) simulated path, measured in iterations. 

\item A \defn{point} is an integer $p \in \{0, 1, \ldots, \max\{\Aiter, \Biter\}\}$, interpreted as a depth on a simulated path.  
For example, if Alice's simulated path is $\simPath_A = ( \sigma_1 \circ 1, \sigma_2 \circ 6 )$, then $p=1$ would correspond to the entry $\simPath_A[1] = \sigma_2 \circ 6$ at depth $1$ on her simulated path.  Below in Definition~\ref{def:megastate}, we will associate each point $p$ (which is an integer) with a \emph{mega-state} $\ms{p}$ for each party, whose depth is the integer $p$.  This mega-state will include some information about that party's simulated computation at that depth.

\item  If $\simPath_A \neq \simPath_B$ or $\ms{\ell_A} \neq \ms{\ell_B}$, then there is a latest point where they still agree: that point is called the \defn{divergent point}.
Formally, the divergent point (which will usually be denoted ``$b$'') is 
\[ \max\{ p : \simPath_A^{(\leq p)} = \simPath_B^{(\leq p)} \text{ and } \forall p' \leq p,~\ms{p'}_A = \ms{p'}_B\} \]
where the notation $\simPath^{(\leq p)}$ denotes the restriction of $\simPath$ to the points $1, 2, \ldots, p$.
If $\simPath_A = \simPath_B$ and $\ms{\ell_A} = \ms{\ell_B}$,  we say that the divergent point doesn't exist.
Notice that if $\Aiter > \Biter$ and Alice and Bob agree up to iteration $\Biter$, then the divergent point is $b = \Biter$. 

\item When Alice or Bob rewind to a previous point in $\Pi$ (to the end of a previous iteration), we call it a \defn{jump}.\footnote{Technically, the ``jump'' occurs on Line~\ref{line:back-jump} in Algorithm~\ref{alg:adaptive} below.}  When Alice and Bob simultaneously jump to the same point, we call it a \defn{successful jump}.  When Alice and Bob have a successful jump above the divergent point $b$, we call it a \defn{successful jump that resolves $b$.}    

\item When the divergent point is $b$, we define the \defn{correct simulated path} to be $\simPath_A^{(\leq b)} = \simPath_B^{(\leq b)}$.  If there is no divergent point, we define the correct simulated path to be $\simPath_A = \simPath_B$.

Notice that if $\simPath_A = \simPath_B$, then Alice and Bob's computations agree (that is, they have correctly simulated part of the original protocol $\Pi$).

\item We define $\ell^+$ to be the length (in iterations) of the correct simulated path.
\item We define $\ell^-_A =\Aiter - \ell^+ $    and analogously $\ell^-_B =\Biter - \ell^+  $. Further, $\ell^- = \ell^-_A + \ell^-_B$.  That is, $\ell^-$ is the depth that Alice and Bob have gone while being out of sync, measured in iterations.  
We call a window of time starting when $\ell^-$ becomes positive to when $\ell^-$ next becomes zero a \defn{bad spell.}  

\item If $\ell^- > 0$ (that is, if we are in a bad spell), we define $L^- $ to be the maximum that $\ell^-$ 
has been during this bad spell. 
\end{itemize}

\begin{remark}
    The main reason for including the iteration number in the definition of $\simPath_A$ and $\simPath_B$ is technical: We would like it to be the case that the only way that a bad spell can end is if Alice and Bob end it ``on purpose,'' by successfully jumping (see Lemma~\ref{lem:noweirdcoincidences}).  If $\simPath_{A,B}$ did not include iteration numbers, the adversary could ``trick'' Alice and Bob into accidentally getting back on track.  Of course it is fine for the correctness of the protocol if the adversary helps Alice and Bob, but it is a hassle in the analysis.

    We note that $\simPath_A$ and $\simPath_B$ are not explicitly stored by Alice and Bob during $\Pi'$; they are used only in the analysis.  On the other hand, $\Bsim_A$ and $\Bsim_B$ (which are short substrings of $\simPath_A$ and $\simPath_B$) are stored.
\end{remark}

\begin{remark}
    Observe that in a given bad spell, the divergent point could change. In fact, it could move from a point $b$ to a point $b'$ with $b'<b$; we spell this out formally in Lemma~\ref{lem:noweirdcoincidences}. 
\end{remark}

\subsection{Alice/Bob Subscripts and Time/Iteration Indexing}\label{sec:notationABtime}\label{sec:subscripts}

\paragraph{$A$ and $B$ subscripts:} Our main protocol (Algorithm~\ref{alg:adaptive}) is from the point of view of either party, and includes variables like $k, E, \MPone, \Mem$, and so on.  For a variable $x$ that appears in Algorithm~\ref{alg:adaptive}, we use $x_A$ to denote Alice's copy, $x_B$ to denote Bob's copy, and $x_{AB}$ to denote $x_A + x_B$.

\paragraph{Times and Iteration Numbers:}
By a \defn{time} $t$, we mean a specific moment in the execution of the protocol $\Pi'$. For example, we might refer to the ``time $t$ right before Line~\ref{line:back-jump} is executed in iteration $I$ of Algorithm~\ref{alg:adaptive}.''

For $I \in \{0, 1, \ldots, \Itotal\}$, 
when we refer to an \defn{iteration $I$} in Algorithm~\ref{alg:adaptive}, we mean the iteration when the variable $\Icount$ in Algorithm~\ref{alg:adaptive} is incremented to $I$ (in particular, we don't reset the iteration count at each block).  By convention, ``Iteration $0$'' refers to the time before the Algorithm begins.
 
 At a particular time $t$ in the execution of Algorithm~\ref{alg:adaptive}, we use $\I(t) \in \{1, 2, \ldots, \Itotal\}$ to denote the iteration that $\Pi'$ was on at time $t$.

\paragraph{Specifying time/iteration with $x(t)$ and $x(I)$:} 
As all of the variables change with time, when the time is not clear from context, we will sometimes write $x(t)$ to denote the value of the variable $x$ at time $t$.  For example, $k_A(t)$ will denote Alice's value of $k$ in Algorithm~\ref{alg:adaptive} at time $t$.  For an iteration number $I \in \{1, \ldots, \Itotal\}$, we sometimes use $x(I)$ to denote the value of $x$ at the \emph{end} of iteration $I$ (Line~\ref{line:iterend} in Algorithm~\ref{alg:adaptive}).

\subsection{Meeting points}\label{sec:MPs}

At each iteration, Alice and Bob will maintain a set of ``meeting points,'' which are points that they are able to jump back to.  
In order to define meeting points, we first introduce some notation:
For non-negative integers $x$ and $y$, define \[\lfloor x \rfloor_y = \max_{j \in \mathbb{Z}} \{ j\cdot y \,:\, j \cdot y \leq x \}.\]  
That is, $\lfloor x \rfloor_y$ is $x$ rounded down to a multiple of $y$.
With this, we have the following definitions. 

Throughout the protocol, Alice and Bob will maintain a \emph{$j$-value}, $j_A$ and $j_B$, respectively, which represent how far they are willing to jump.
Sometimes we need to refer to the hypothetical transition candidates if a party was at the point $\ell$ and had $j$-value $j$.  To that end, we have the following definition.

\begin{definition}[MP set, scale-$j$ MPs]\label{def:MP}
Let $a \in \{1, \ldots, \lceil d/r \rceil \}$, where we think of $a$ as representing the depth of a partial simulation of $\Pi$.  The \defn{meeting point (MP) set for $a$}, denoted $M_a$, is defined as
\[ M_a = \{ b \geq 0 \,:\, \exists j \geq 0 \text{ so that }  \lfloor a \rfloor_{2^j} - 2^j = b\}. \]
We call $\lfloor a \rfloor_{2^j} - 2^j$ the \defn{scale-$j$ meeting point (MP) for $a$.}  
That is, the scale-$j$ MP for $a$ is the \emph{second}-closest multiple of $2^j$ below $a$.
\end{definition} 
We remark that this is the same meeting point set that was considered in prior work, including \cite{H14} and \cite{EHKKRS23}.

In our protocol (Algorithm~\ref{alg:adaptive}), Alice and Bob will only consider three special meeting points at a time, which we call the \defn{transition candidates}.
\begin{definition}[Scale-$j$ transition candidates]\label{def:transition-candidates}
    Fix $j$ and $\ell$.  Define the \defn{scale-$j$ transition candidates for point $\ell$}
to be the set 
$\MPalljl(j, \ell) $ that contains the following three points:
\begin{itemize}
    \item $\MPone(j,\ell) := \lfloor \ell \rfloor_{2^{j+1}} - 2^{j+1}$.  (That is, $\MPone(j,\ell)$ is the scale-$(j+1)$ MP for $\ell$).
\item $\MPtwo(j,\ell) := \lfloor \ell \rfloor_{2^j} - 2^j$.  (That is, $\MPtwo(j,\ell)$ is the scale-$j$ MP for $\ell$).
\item $\MPthree(j,\ell)$, which is the largest $p \in M_{\ell}$ so that $2^{j}$ divides $p$. 
\end{itemize}

We will use the notation $\MPonetwojl(j, \ell) = \{\MPone(j,\ell), \MPtwo(j,\ell)\}.$
\end{definition}

At a particular time $t$ in the protocol, when Alice is at position $\ell_A$ and has $j$-value equal to $j_A$, we will have 
\[ \MPone_A = \MPone(j_A, \ell_A), \MPtwo_A = \MPtwo(j_A, \ell_A), \MPthree_A = \MPthree(j_A, \ell_A),\]
and we refer to these as the \defn{transition points for Alice} at time $t$ (and similarly for Bob).  Similarly, we will use $\MPall_A$ to denote $\MPalljl(j_A, \ell_A)$ and $\MPonetwo_A$ to denote $\MPonetwo(j_A, \ell_A)$ (and similarly for Bob).

\paragraph{How meeting points and their corresponding states are represented in the protocol.}

Throughout the protocol (Algorithm~\ref{alg:adaptive}), Alice and Bob will store information about a \emph{subset} of the meeting points in $M_{\Aiter}$ and $M_{\Biter}$, respectively.

In more detail, each party stores:

\begin{itemize}
\item A set $\M \subseteq \{0,1,\ldots, \lceil d/r \rceil \}$ of points.  The set $\M$  contains a subset of $M_{\Aiter}$ (for Alice) and $M_{\Biter}$ (for Bob).

\item The string $\Bsim$, which we recall is the substring of the simulated path $\simPath$ restricted to the iterations that occurred during the current block.  

\item A memory $\mathtt{Memory}$ that stores \defn{mega-states} $\mathtt{p}$ corresponding to the points in $\M$.  
We formally define a mega-state as follows.

\begin{definition}\label{def:megastate}
A \defn{mega-state} is the information stored in $\Mem$ that corresponds to a state in $\Pi$ at the end of an iteration.  Each mega-state $\mathtt{p}$ includes:
\begin{itemize}
	\item A state $\mathtt{p.v}$ of $\Pi$, which is the last state simulated in that party's simulated path. 
	\item An integer $\mathtt{p.depth}$, which is the {iteration}-depth of the state $\mathtt{p.v}$ in the simulated path.  (Recall that this is the depth of $\mathtt{p.v}$ in $\Pi$, divided by $r$).
	\item A hash value and its corresponding random seed. In the pseudocode, these are denoted as $\phash$ and $\pseed$, respectively. (Outside of the pseudocode, we often refer to them as $H$ and $R$, respectively, so that the equations don't get bogged down with variable names.) 
  \item An integer $\piter$, which is the iteration of $\Pi'$ that was the last time $(\phash, \pseed)$ were generated.  Similarly, $\piter$ is commonly denoted as $I_{H}$ outside of the pseudocode. 
    \item In the pseudocode, we also use a variable called $\pt$. This variable is not actually stored separately, rather it refers to the substring of $\Bsim$ that was simulated up to the iteration $\piter$.  Thus, $\pt$ is just shorthand for information that is stored elsewhere.  We use $\Bsimp$ to refer to the this string, both in the pseudocode and an in the analysis.

\end{itemize} 
For an integer $q$, we use $\ms{q}$ to refer to the mega-state that has depth $q$, if it exists.  (We use $\ms{q}_A$ and $\ms{q}_B$ to specify Alice or Bob's copy of memory when needed.)
\end{definition}
\end{itemize}
 As one would expect, a point $q$ will be in $\M$ if and only if there is some $\texttt{p} \in \Mem$ so that $\texttt{p.depth}=q$, that is, so that $\texttt{p} = \ms{q}$.  Thus, Alice's memory contains mega-states that line up with the points in $\M$.  

We say that a meeting point $q \in M_{\Aiter}$  is \defn{available} for Alice if $q \in \M_A$ (where $\M_A$ is Alice's copy  of $\M$).  We define available meeting points for Bob similarly. 

\paragraph{A few useful facts about meeting points.}  We record a few useful lemmas about meeting points.  First, we have a useful definition. 

\begin{definition}\label{def:stable}
We say that a point $p$ is \defn{$j$-stable} if $2^j | p$ but $2^{j+1} \nmid p$.
\end{definition}
Intuitively, points that are very stable will take a long time for Alice or Bob to forget, in the sense that a very stable point $p$ will remain in $M_a$ for very large values of $a$.  To quantify this, we have the following lemma.  (We note that this has essentially appeared in previous work, but we provide a proof for completeness.)

\begin{lemma}[\cite{HR18}]\label{lem:inMa}
Suppose that $p$ is $j$-stable and let $a$ be a non-negative integer.  Then $p \in M_a$ if and only if $a \in \{p, p+1, p+2, \ldots, p + 2^{j+1} - 1 \}$.
\end{lemma}
\begin{proof}
Suppose that $p$ is $j$-stable, and fix $a$.  First, suppose that $p \leq a \leq p + 2^{j+1} - 1$.
Let $w = \max\{w' \,:\, p + 2^{w'} \leq a\}$, and notice that $w \leq j$, since $p + 2^{j+1} > a$ by assumption.  Then $p = \lfloor a \rfloor_{2^w} - 2^w$, and hence $p \in M_a$.

On the other hand, suppose that $a < p$.  Then it is clear that $p \notin M_a$.  Now suppose that $a \geq 2^{j+1}$.  Suppose towards a contradiction that $p \in M_a$, so there is some $w$ so that $p = \lfloor a \rfloor_{2^w} - 2^w$.  If $w > j$, then $2^w$ divides $\lfloor a \rfloor_{2^w} - 2^w$, but by the definition of $j$-stability does not divide $p$, a contradiction.  But if $w \leq j$, then 
\[ \lfloor a \rfloor_{2^w} - 2^w \geq \lfloor a \rfloor_{2^j} - 2^j = p + 2^j \neq p,\]
another contradiction.  Thus, $p \notin M_a$.  This completes the proof.
\end{proof}

Our next helpful lemma says that if $L^-$ isn't too large at time $t$, relative to $2^j$, then Alice and Bob have a common scale-$j$ transition candidate such that they agree on the mega-state corresponding to this point. Moreover, we can take this candidate to be $\MPone$ for whichever party is deeper in their simulated path.

\begin{lemma}\label{lem:commonpt}
Fix a time $t$ and an integer $j$, and suppose that $2^j \geq 8 L^-(t)$.   Suppose that $\Aiter(t) \geq \Biter(t)$. Then there is a point $p$ so that:
\begin{itemize}
    \item $p = \MPone(j,\Aiter(t))$; and
    \item $p \in \MPonetwojl(j, \Biter(t)).$
    \item $\texttt{p} = \ms{p}_A = \ms{p}_B$ 
\end{itemize}
\end{lemma}
\begin{proof}
    Let $p$ be the scale-($j+1$) MP for $\Aiter(t)$.  Then by definition, $p = \MPone(j,\Aiter(t))$.  Let $\hat{p} = p + 2^{j+1}$.  Notice that by the definition of $p$, 
    \[ \hat{p} \leq \Aiter(t) < \hat{p} + 2^{j+1}.\]  Since $\Biter(t) \leq \Aiter(t)$ and since 
    \[ \Aiter(t) - \Biter(t) \leq L^-(t) \leq 2^{j-3}, \]
    we have 
    \[ \hat{p} - 2^{j-2} < \Biter(t) < \phat + 2^{j+1}. \] 
    This implies that either $p$ is a scale-($j+1)$ MP for Bob (if $\Biter(t) \geq \hat{p}$) or a scale-$j$ MP for Bob (if $\Biter(t) < \hat{p}$) since $p + 2^j \leq \hat{p} - 2^{j-2}$.
    In the first case, we have
    $p = \MPone(j, \Biter(t))$, and in the second we have $p = \MPtwo(j,\Biter(t))$. Further, as $p - \ell_A \geq 2^j$ and $p-\ell_B \geq 2^j$ and $2^j > L^-$, then $p$ is above the divergent point $b$, and as a result $\ms{p}_A = \ms{p}_B$, this proves the lemma.  
\end{proof} 

\paragraph{\Specialpts.}
Finally, we define a \defn{\specialptf} to be one that Alice and Bob could have jumped to if their $j_A$, $j_B$ parameters were appropriate:
\begin{definition}[\Specialpts] \label{def:specialpts}
    Fix a time $\tsp$, and suppose that $k_A(\tsp) = k_B(\tsp)$.  (We remark that in Algorithm~\ref{alg:adaptive}, this will imply that $j_A(\tsp) = j_B(\tsp)$.)  
    \begin{itemize}
    
    \item We say that a point $p$ is a \defn{\specialpt with scale $j$} at time $\tsp$ if 
    \[ p \in \MPalljl(j, \ell_A(\tsp)) \cap \MPalljl(j, \ell_B(\tsp)) \ \]
    and 
     \[ \ms{p}_A = \ms{p}_B.\]
    That is, $p$ is a \specialpt with scale $j$ if $p$ would appear as a common meeting point candidate for Alice and Bob and they both agree on the corresponding mega-state, provided that $j_A = j_B = j$. 
    \item Suppose that a meeting point $p$ is available for both Alice and Bob at time $\tsp$, meaning $p\in \M_A(\tsp)$ and $p\in \M_B(\tsp)$. Then,  we call $p$ an \defn{available \specialpt with scale~$j$}.  

    \item We will sometimes refer to a \defn{\specialptf} at time $\tsp$ (with no scale specified); by this we mean a \specialpt with scale $j_A(\tsp) = j_B(\tsp)$.

    \item Let $\jsp = j_A(\tsp) = j_B(\tsp)$.
    The \defn{next \specialscale} at time $\tsp$ is the smallest $j \geq \jsp$ so that there is an available \specialpt with scale $j$.  Such a point is called a \defn{next available \specialpt}.
    \end{itemize} 
\end{definition}
\begin{remark}
    In Definition~\ref{def:specialpts}, we are assuming that $k_A(\tsp) = k_B(\tsp)$, and we will use this definition in settings when $k_A, k_B$ are growing large (in particular, larger than one).  When this happens, Alice and Bob are not moving in the protocol $\Pi$: they are simulating dummy rounds, and voting to attempt to agree on a meeting point to jump back to.  Thus, the next available \specialpt at a time $\tsp$ is the point they will next have an opportunity to jump to as $k_A, k_B$ grow (unless the adversary tricks one or both of them into jumping early).
\end{remark}
\begin{remark}\label{rem:jumpable}
Suppose that $b$ is the divergent point at time $\tsp$.  Suppose that $p$ is a jumpable point with scale $j$ at time $\tsp$, and further that $p < b$.
Then in order to establish that $p$ is an \emph{available} jumpable point (as in Definition~\ref{def:specialpts}),  it is sufficient for $p$ to be included in both the meeting point sets, $p\in \M_A \cap \M_B$. The reason for this is that, since $p< b$, we have $\simPath_A[p] = \simPath_B[p]$, so the point $p$ lies on the correct simulation path. Further, from the definition divergent point, we know that $ \ms{p}_A = \ms{p}_B $, which means the mega-states corresponding to $p$ in the memories of both parties match each other.
\end{remark}

\subsection{Hash Functions}\label{sec:hashprelims}
The protocol will make use of three hash functions, $h_1, h_2$, $h_3$.  We build them out of the following ingredients.

As in \cite{H14}, we will use the following hash functions, from \cite{NN93}:
\begin{theorem}[\cite{NN93}]\label{thm:hash}

There is a constant $\Chashtwo$ so that the following holds.  For all integers $o$ and $t$, there is an integer $\seed = \Chashtwo (o + \log t)$ and a function $h:\{0,1\}^t \times \{0,1\}^\seed \to \{0,1\}^o$, computable in time $\poly(o,t)$, so that for any $x \neq y \in \{0,1\}^t$,
  \[ \Pr[ h(x,R) = h(y,R) ] \leq 2^{-o/\Chashtwo},\]
  where the probability is over a seed $R$ drawn uniformly at random from $\{0,1\}^\seed$.
\end{theorem}

Following \cite{H14},
we will also need to ``stretch'' a uniformly random binary string of length $O(\log \ell + \log(1/\delta))$ into a string of length $\ell$, that is $\delta$-biased.\footnote{We say that a random string $X \in \{0,1\}^\ell$ is $\delta$-biased if for all sets $S \subseteq [\ell]$, $\left| \Pr[ \sum_{i \in S} X_i = 0 ] - \Pr[ \sum_{i \in S} X_i = 1 ] \right| \leq \delta$, where the sums are mod 2.  What will be useful for us (as it was for \cite{H14}) is that any $\delta$-biased $X \in \{0,1\}^\ell$ is $\eps = \poly(\delta)$ close (in total variation distance) to a distribution that is $O(\log (1/\delta))$-wise independent.} The idea is then to chop up this pseudorandom string to use as seeds across an entire block's worth of hash functions. 

In particular, we will make use of the following theorem, adapted from \cite{H14}, which uses the $\delta$-biased random variables from \cite{NN93} along with the \emph{inner product hash family}:   
\begin{theorem}[Follows from \cite{NN93}, \cite{H14}]\label{thm:pseudorand}
    Fix integers $m$, $t$ and $o$, and let $\delta \in (0,1)$. 
    Let $\seed = 2 \cdot t \cdot o$, and let $\ell = m \cdot \seed$.
    Then there is a number $\ell' = O( \log(\ell) + \log(1/\delta))$ and functions 
    \[ h:\{0,1\}^t \times \{0,1\}^{\seed} \to \{0,1\}^o\] and 
    \[ \mathtt{extend}: \{0,1\}^{\ell'} \to \{0,1\}^{\ell} \]
    so that the following holds.
    Let $(x_1, y_1), \ldots, (x_m, y_m)$ be pairs of distinct binary strings, so that $x_i, y_i \in \{0,1\}^t$.  Choose $R \in \{0,1\}^{\ell'}$ uniformly at random, and let $\tilde{R} \in \{0,1\}^{\ell}$ be $\tilde{R} = \mathtt{extend}(R)$.  Divide up $\tilde{R}$ into $m$ blocks of length $\seed$, and for $i \in [m]$, let $\tilde{R}[i]$ denote the $i$'th block.  
    Let 
    \[Z_i = \mathbf{1}\bigl[ h(x_i, \tilde{R}[i]) = h(y_i, \tilde{R}[i] ) \bigr].\]
    Then the distribution $Z = (Z_1, \ldots, Z_m)$ satisfies
    \[ \| Z - W \|_{TV} \leq \delta, \]
    where $W = (W_1, \ldots, W_m)$ are i.i.d. Bernoulli-$2^{-o}$ random variables, and $\|\cdot\|_{TV}$ denotes the total variation distance. Furthermore the function $\mathtt{extend}$ can be computed in time $\poly(\ell)$.
\end{theorem}

With these ingredients, we will define three hash functions, $h_1, h_2,h_3 $, and from them we will build  two hash functions, $h^s$ (the \defn{small hash function}) and $h^b$ (the \defn{big hash functions}).

The hash function $h_1$, along with a method \texttt{extend}, will come from Theorem~\ref{thm:pseudorand}.  The hash functions $h_2$ and $h_3$  will come from Theorem~\ref{thm:hash}. 

We formally say how to construct each of these hash functions in Algorithm~\ref{alg:init} (\textsc{Initialization}) below, along with all the parameters, but here we describe each of the hash functions in words, to give some intuition about their purpose and why they are constructed the way they are.

For the \defn{small hash function}, we 
 define
\[ h^s(x, R_1, R_2) = h_2( h_1(x, R_1), R_2). \]  
The reason to combine two hash functions is the following.  We will use the small hash function every iteration, and we cannot afford to have Alice and Bob swapping large seeds in each iteration.  Instead, they swap a small seed during each \emph{iteration} (this is the one used by $h_2$), and a larger  seed once at the beginning of each \emph{block}, which they then pseudorandomly extend via the \texttt{extend} function from Theorem~\ref{thm:pseudorand}; chopping up this pseudorandom string gives the seeds used by $h_1$.  Because there are roughly $1/\log(d)$ as many blocks as there are iterations, Alice and Bob can afford a longer seed for $h_1$.  The catch is that, because the seed for the pseudorandom string needs to be shared in advance so that the \texttt{extend} function can be applied, the adversary knows what it is; thus we will have to union bound over all decisions the adversary might make in the analysis (see Section~\ref{sec:hash}).

Recall that $h_3$ comes from Theorem~\ref{thm:hash}. We define the \defn{big hash function} $h^b$ as $h^b (x,R) = h_3(x,R)$ (the parameters will be set in Algorithm~\ref{alg:init} below).
This big hash function is used at the end of every block to update the chained hash stored inside the megastates simulated during that block. The seed is shared right before the big hash computation and is protected with an error correcting code (see Algorithm~\ref{alg:adaptive} and Algorithm~\ref{alg:Rand-exchange} for more details).

Crucially, Alice and Bob never directly exchange hashes from the big hash function, as they do not have bandwidth to do so.
Rather, as discussed in the introduction, the point is that the big hash will be ``chained,'' to make sure that Alice and Bob have some record of the past.  They will then use the \emph{small} hash function (which has a higher collision probability) to exchange hashes of their big hashes, among other things.

\subsection{Other Notation}\label{sec:othernotation}

\paragraph{Bad Vote Count ($\BVC$):} We also define a new parameter for each of Alice and Bob that does not appear in Algorithm~\ref{alg:adaptive}, called $\BVC$ (Bad Vote Count).  Each party maintains three voting counters, $\vi{1}, \vi{2}, \vi{3}$.  Alice increments  $\vi{i}_A$ when she has reason to believe that $\MPi_A$ also appears as one of $\MPone_B, \MPtwo_B, \MPthree_B$, and Bob behaves similarly.  
Informally, Alice's \defn{Bad Vote Count} $\BVC_A$ increases by one whenever she either erroneously increments $\vi{i}_A$, or when she fails to increment $\vi{i}_A$ when she should have.  (And similarly for Bob.)  We formally define $\BVC$ in Definition~\ref{def:BVC} after we present the pseudocode for Algorithm~\ref{alg:adaptive}. 

\paragraph{Remaining notational conventions:} For two strings $x,y$, we use $x \circ y$ to denote the concatenation of these strings; for an integer $n$, we use $[n]$ to denote the set $\{1, \ldots, n\}$.

%% file: protocol.tex
\newcommand{\headerComment}[1]{\textbf{$\blacktriangleright$ #1}}
\newcommand{\leftComment}[1]{\State \textcolor{black!50}{$\triangleright$ \emph{#1}}}
\newcommand{\itComment}[1]{\Comment{\textcolor{black!50}{\emph{#1}}}}
\subsection{Protocol Overview}\label{sec:protocol_overview}

As discussed above, the algorithm is broken into blocks, which are then broken into iterations.   At the top of each block, Alice and Bob exchange randomness, which will be pseudorandomly extended and used for the small hashes throughout the block.    Each block is broken into $\Iblock$ iterations, and within each iteration there are three more phases: The \defn{Verification Phase}, in which Alice and Bob exchange small hashes to verify consistency, and vote for places to jump back to if the verification failed; the \defn{Computation Phase}, in which Alice and Bob simulate $\Pi$, assuming nothing seems amiss; and the \defn{Transition Phase}, in which Alice and Bob may jump to an earlier point.  Finally, at the end of each block, there is a \defn{Big Hash Computation Phase}, where Alice and Bob exchange fresh randomness to update their chained big hashes. 
The high level pseudocode is given in Algorithm~\ref{alg:adaptive_short}. 

\begin{algorithm}
\small
\caption{\textbf{Birds-Eye View} of Algorithm~\ref{alg:adaptive} (See Alg.~\ref{alg:adaptive} for more detail)}
\label{alg:adaptive_short}
\begin{algorithmic}[1]

\State \textbf{Parameters:} count $k$, error estimate $E$, big hash value $\mathcal{H}$, current mega-state $\texttt{p}$

\;

\For{$\Btotal$ blocks}
   \State Exchange randomness for small hashes in this block.
    \For{ $\Iblock$ iterations } 

        \State $k \leftarrow k + 1 $   
        \State $j = \lfloor \log(k) \rfloor $
        \State Choose $\{\MPone, \MPtwo, \MPthree\} = \MPalljl(j, \pdepth)$ as in Section~\ref{sec:MPs} 
      
        \State $T = \emptyset$
        
        \; 
        
\color{blue!70!black}
        \State \headerComment{Verification Phase}  
        \State Exchange additional randomness for small hash
        \State Hash many parameters using the small hash function $h^s$
        \State Transmit these hash values, and receive corresponding hash values.

        \If{ things seem fishy } 
		 \State Increment error estimate: $E \leftarrow E+1 $ 
        \Else 
	
 \State For each $i = 1,2,3$, if  $h^s(\MPi)$ agrees with one of the other party's MP hashes:
    \State \ \ \ \ Increment vote $\vi{i}$ for MP$i$.
            	
        \EndIf
        
               \; 
               
\color{green!40!black}
        \State\headerComment{Computation Phase}

         \If{ everything seems on track } 
        	\State Simulate $\Pi$ for $r$ rounds to update current mega-state $\megastate$.
              
                \State Append the simulation transcript and the iteration number to $T$.
                \State Update $\Mem$ to forget states that are no longer in $M_{\Aiter}$ or $M_{\Biter}$.
         
        	 \State Reset Status: $ k, E , \vone , \vtwo ,\vthree \leftarrow  0$ \label{line-simple:reset1}
        \Else
        	\State Simulate $r$ dummy rounds
        \EndIf
        
\;

\color{violet!80!black}
        \State \headerComment{Transition Phase}
        \If{ $2E \geq k$ } 
		\State Reset Status: $ k, E , \vone , \vtwo, \vthree  \leftarrow 0$ \label{line-simple:reset2} 
        \ElsIf{ $ k = 2^{j+1} -1$ and there is an available $\MPi$ with $\vi{i} \geq 0.4 k$ and something was fishy}
            \State Jump to $\MPi$ and update $\Mem$
            \State Update the partial simulation path $T$ for the new megastate
            \State Reset Votes: $\vone , \vtwo , \vthree \leftarrow  0$
        \EndIf

\; \label{line-simple:iterend}
\EndFor
\color{orange}
     \State\headerComment{Big Hash Computation}
     \State Exchange randomness for big hash.
     \For{every $\megastate \in \M$ that is simulated in the current block}
     \State Compute the chained hash using $h_b$ and update $\megastate$.
     \EndFor
\color{black}

\EndFor

\;
\State\headerComment{All done!}
\State Celebrate!  We have successfully simulated $\Pi$!
\end{algorithmic}
\end{algorithm}
\normalsize

\subsection{The full protocol}\label{sec:full_protocol}
Now, we define the full protocol in detail.  Our final algorithm is given in Algorithm~\ref{alg:adaptive}, but first we define several auxiliary protocols.

In the protocol, we will make use of an efficient asymptotically good error correcting code.
\begin{theorem}[See, e.g., \cite{essentialcodingtheory}]\label{thm:goodECC}
Let $\ell$ be an integer.  Then there is a value $\hat{\ell} = \Theta( \ell + \Iblock )$ so that the following holds.
There is binary error-correcting code $\mathcal{C} \subset \F_2^{\hat{\ell}}$ with block length $\hat{\ell}$, message length $\ell$, and minimum distance at least $4 \Iblock$.  Moreover, there are polynomial-time encoding and decoding maps $\Enc: \F_2^{\ell} \to \F_2^{\hat{\ell}}$ and $\Dec: \F_2^{\hat{\ell}} \to \F_2^{\ell}$ so that 
\[ \Dec( \Enc( x + e ) ) = x \]
for any $e \in \F_2^{\hat{\ell}}$ with Hamming weight at most $2 \Iblock$. 
\end{theorem}

We begin in Algorithm~\ref{alg:Rand-exchange} with the protocol that Alice and Bob use to exchange randomness. 
The protocol either simply encodes randomness with an error correcting code and sends it; or receives such an encoding from the other party and decodes it.

\begin{algorithm}
\small
\caption{\textsc{RandomnessExchange}: Protocol for Alice and Bob to exchange random bits} 
\label{alg:Rand-exchange}
\begin{algorithmic}[1]
\Require{Parameter $\ell$ (integer), \texttt{send} (boolean)} 
\State Let $\mathcal{C}$, $\Enc$ and $\Dec$ be as in Theorem~\ref{thm:goodECC}.
\If{\texttt{send} = True} 
\State $R \gets  $ uniform bit string of length $\ell$
\State Send $\Enc(R)$ to the other party.
\ElsIf{\texttt{send} = False}
\State Receive $c$ from other party, 
\State $R \gets \Dec ( c )$
\EndIf
\end{algorithmic}
\end{algorithm} 

Next in Algorithm~\ref{alg:PRand-exchange} we give an algorithm for exchanging pseudorandom bits, which simply exchanges a short random seed as in Algorithm~\ref{alg:Rand-exchange}, and then both parties pseudorandomly extend that seed.

\begin{breakablealgorithm}
\caption{\textsc{PseudoRandExchange}: Protocol for Alice and Bob to exchange pseudorandom bits}
\label{alg:PRand-exchange}
\begin{algorithmic}[1]
\Require{Parameters $\ell$, $\delta$ and \texttt{send} (boolean)}
\State $\ell^\prime \gets O(\log ( 1 / \delta ) + \log\ell )$  \itComment{$\ell'$ is as in Theorem~\ref{thm:pseudorand}}
\State $R' \gets$ \textsc{RandomnessExchange($\ell'$,} \texttt{send})
\itComment{Run Algorithm~\ref{alg:Rand-exchange}}
\State $R \gets \texttt{extend}(R')$ \itComment{ \texttt{extend} is as in Theorem~\ref{thm:pseudorand}}
\end{algorithmic}
\end{breakablealgorithm}

Next, we have a method to initialize all of the parameters in the main protocol, Algorithm~\ref{alg:adaptive}.  The pseudocode and comments in Algorithm~\ref{alg:init} also includes the interpretation of these parameters.

\begin{algorithm}
\caption{\textsc{Initialize}: Initialization for Algorithm~\ref{alg:adaptive}}
\label{alg:init}
\begin{algorithmic}[1]
\Require{Parameters $d,\eps$}
\State \headerComment{Choose Parameters:}\\
$r_c = \Theta\bigl(\log\log\frac{1}{\eps}\bigr)$ \itComment{$r_c$ is a bound on ``extra'' communication per iteration---}\\
\itComment{we show in the proof of Theorem~\ref{thm:main} that we may take $r_c = \Theta(\log\log(1/\eps)).$}\\
$r = \bigl\lceil \sqrt{\frac{r_c}{\eps}}\bigr\rceil $  \itComment{$r$ is the number of rounds in an iteration} \\
$\Iblock  = \lceil \log d \rceil$ \hfill \itComment{$\Iblock$ is the number of iterations in a block.}\\
$\Itotal = \lceil d / r \rceil + \Theta( d \epsilon )$ \hfill \itComment{$\Itotal$ is the  total number of iterations.}\\
$\Btotal  = \lceil \Itotal / \Iblock \rceil$\label{line:Btot} \hfill \itComment{$\Btotal$ is the total  number of blocks.}\\\\

\headerComment{Initialize Memory and Meeting points:}\\
$\Bsim = \emptyset $ \itComment{$\Bsim$ stores the partial simulated path from the current block}\\
$\M \gets \emptyset$ \itComment{$\M$ is the meeting point set}

\leftComment{Intialize a ``mega-state'' \texttt{p}:}\\

$\ \ \ \ \mathtt{p.v} \gets $ start state of $\Pi$ \itComment{$\mathtt{p.v}$ stores the current state in $\Pi$.} \\
$\ \ \ \ \phash, \pseed \gets \mathtt{None}$ \itComment{Hash of simulated path, seed used for that hash.}\\
$ \ \ \ \ \piter \gets 0$ \itComment{Iteration number of  the last time hash of the simulated path is computed}  \\
$\ \ \ \ \mathtt{p.depth} \gets 0$ \itComment{ $\mathtt{p.v}$ is $\mathtt{p.depth}$ \emph{iterations} deep in $\Pi$. }\\
{$\ \ \ \ \pt \gets \bot$} \itComment{$\pt$ will be a placeholder for $\Bsimp$ throughout the algorithm; however we only actually store $T$ because of space constraints.} \\
$\Mem \gets [\mathtt{p}]$ \itComment{$\Mem$ \ contains a list of remembered mega-states}\\

$\mathtt{Rew} \gets \texttt{False}$ \itComment{$\mathtt{Rew}$ says whether this party should rewind or not.}
\\\\

\headerComment{Set up hash families:}\\
\begin{tabular}{lll}
$t_1 \gets \log s + O(r \log d + {  \log ^2 d })$, &$o_1 \gets 2\log(1/\eps)$, &$\seed_1 \gets 2 t_1 o_1$ \\
$t_2 \gets 2\log(1/\eps)$, &$o_2 \gets \Chash$, &$\seed_2 \gets \Chashtwo(o_2 + \log(t_2)) = O(\log\log(1/\eps))$ \\

$t_3 \gets \Theta(r \log d + \log^2 d)$, &$o_3 \gets \obighash$, &$\seed_3 \gets { \Chashtwo ( o_3 + \log(t_3) )  = O( \obighash + \log \frac{1}{\eps} ) }$  \label{line:big-h-par}

\end{tabular}
\vspace{.2cm}
\leftComment{
 Note: $\Chash$ is an absolute constant that will be chosen later.  $C_b$ will be chosen as a sufficiently large constant, independent of $\eps$ (assuming $d$ is sufficiently large, as we do). $\Chashtwo$ is the constant from Theorem~\ref{thm:hash}.}
\leftComment{Note: The constant in the big-Theta on $t_3$ is set so that the hash $h^b$ can accommodate all of its inputs. 
 We verify that this can be done consistently, as well as verifying that that all of these parameter choices are consistent, in Claim~\ref{claim:bighashokay}.}

\\

 \itComment{Above, we note that these quantities may not necessarily be integers.  Adding ceilings to them does not change the analysis, so in a slight abuse of notation, we choose to omit the ceilings for notational simplicity.}

\State Let $h_1: \F_2^{t_1} \times \F_2^{\seed_1} \to \F_2^{o_1}$ be as guaranteed in Theorem~\ref{thm:pseudorand} \\
Let $h_2: \F_2^{t_2} \times \F_2^{\seed_2} \to \F_2^{o_2}$ be as guaranteed in Theorem~\ref{thm:hash}\\

Define $h^s: \F_2^{t_1} \times \F_2^{\seed_1} \times \F_2^{\seed_2} \to \F_2^{o_2}$ by
$h^s(x, s_1, s_2) = h_2( h_1(x,s_1), s_2).$ \itComment{Small hash}\\

Let $h^b: \F_2^{t_3} \times \F_2^{\seed_3} \to \F_2^{o_3}$ be as guaranteed in Theorem~\ref{thm:pseudorand} \itComment{Big hash}\\ 

\\

\headerComment{Set up good ECC:}
\State Let $\Enc$ and $\Dec$ denote the encoding and decoding functions from Theorem~\ref{thm:goodECC}.

\leftComment{We abuse notation slightly and use the same $\Enc$ and $\Dec$ notation to denote the encoder and decoder for different block lengths.}
\end{algorithmic}
\end{algorithm}


Next, in Algorithm~\ref{alg:Maintain-AvMPs} we include an algorithm for updating each party's memory (both the list $\M$ of points that are being remembered; and the memory $\Mem$ of mega-states to which those points correspond), given a current state $\texttt{p}$.
\begin{algorithm}[H]
\caption{\textsc{MaintainAvMPs}: Protocol to maintain the available meeting point set}
\label{alg:Maintain-AvMPs}
\begin{algorithmic}[1]
\Require{Parameters \texttt{p},\ $\M$, $\Mem$, $\texttt{add}$}
\State $a \gets \texttt{p.depth}$
\State $M_a \gets \{ b \geq 0 \,:\, \exists j \geq 0 \text{ s.t. } b = \lfloor a \rfloor_{2^j} - 2^j \}$

\State $\M' \gets (\M \cap M_a) \cup \{a\}$ 
\label{line:updateMem}
\State For each $b \in \M \setminus \M'$, remove the mega-state $\ms{b}$ from $\Mem$.  \label{line:forget}
\State \itComment{Recall from Def.~\ref{def:megastate} that $\ms{b}$ is the mega-state in $\Mem$ with depth $b$.}
\If{$\texttt{add} = True $}
 \State Add a copy of $\texttt{p}$ to $\Mem$.  
\EndIf 

\State $\M \gets \M'$
\end{algorithmic}
\end{algorithm}


\newpage
Finally, we can state our main protocol, which is given in Algorithm~\ref{alg:adaptive}.
\begin{breakablealgorithm}
\caption{\textsc{RobustSimulate}: Protocol to robustly simulate $\Pi$ against an adaptive adversary.}
\label{alg:adaptive}
\begin{algorithmic}[1]
\Require Noiseless protocol $\Pi$, with $d$ rounds; 
and the adversary's error budget $\eps \in (0,1)$

\;
\State\headerComment{Initialize parameters, bookkeeping, and hash functions}
\State\textsc{Initialize}($d,\eps$) \itComment{Run Algorithm~\ref{alg:init}} 
\State \Icount = 0; \Icnt = 0 \itComment{\Icount \ ranges from $0$ to $\Itotal$, while \Icnt \ gets reset each block.}

\;
\State\headerComment{Begin main protocol}
\For{$\Btotal$ blocks}
\State $\Icnt \gets 0$ \itComment{Re-Initialize the block-wide iteration counter}
    \State $\Rblock^s, \Rtblock^s  \gets \textsc{PseudoRandExchange}(\ell=\seed_1 \cdot \Iblock, \delta = 2^{ -C_\delta \cdot \Iblock} , \texttt{send=True}$) \label{line:randex} \\
\itComment{For Bob, \texttt{send = False}.}\\
\itComment{$C_\delta$ is a constant that will be chosen in the proof.}

    \leftComment{Note: Alice and Bob do not actually store all of  $\Rblock^s$.  They store $\Rtblock^{s}$, and generate chunks of $\Rblock^{s}$ as needed.  See the proof of Theorem~\ref{thm:main} for details.}

    \State $T=\emptyset$

    \For{ $\Iblock$ iterations \label{line:iterloop}} 
    
        \State $\Icount \gets \Icount + 1$; $\Icnt \gets \Icnt + 1 $ 
        \State $k \leftarrow k + 1 $   \label{line:iterstart}
        \State $j = \lfloor \log(k) \rfloor $

        \State $\MPone \leftarrow $ scale $j + 1$  meeting point for \texttt{p.depth} 
        \State $\MPtwo \leftarrow $ scale $j$ meeting point for \texttt{p.depth}
        \State $\MPthree \leftarrow  \max\{ p' \in M_{\pdepth}\,:\, 2^j \mid p' \}$ 
        \label{line:chooseMPs}
        \State $\Riter \gets \textsc{RandomnessExchange}(\ell = \seed_2)$
        
        \;

\color{blue!70!black}
        \State \headerComment{Verification Phase}  
        
        	\For{ \texttt{var} $\in \{k, \pv, (\phash, \pseed, \pt, \piter)\}$ \label{line:var}} 
		\State $H_{\texttt{var}} \gets h^s( \texttt{var}, \Riter, \Rblock^s[\Icnt] )$ \label{line:varhash} 
\State\itComment{Divide $\Rblock^s$ into $\Iblock$ chunks; let $\Rblock^s[I]$ denote the $I$'th chunk.}	
      \State Send $H_{\texttt{var}}$ to the other party 
      \State Receive $H_{\texttt{var}}'$ from the other party
	\EndFor

        \For{$i \in \{1,2,3\}$}
           
            \State $\texttt{q}_i \gets \ms{\MPi}$ \itComment{If $\MPi$ is not available, then $\texttt{q}_i = \texttt{None}$}
\label{line:qi}
            \State $H_{\texttt{q}_i.\texttt{v}} = h^s ( \texttt{q}_i.\texttt{v} , \Riter , \Rblock^s[\Icnt] ) $ \label{line:ms-hash-v}
            \State $H_{\texttt{q}_i, b} = h^s(  ( \qihash, \qiseed,\qit, \qiiter) , \Riter , \Rblock^s[\Icnt]  )$\label{line:ms-hash-small-big-hash}
            \State $H_{\texttt{q}_i \texttt{.depth} } = h^s ( \texttt{q}_i\texttt{.depth} , \Riter, \Rblock^s[\Icnt] )$\label{line:ms-hash-depth}
            \State Send $H_{\texttt{q}_i.\texttt{v}}, H_{\texttt{q}_i, b} ,H_{\texttt{q}_i \texttt{.depth} } $ to the other party 
             \State Receive $H^\prime_{\texttt{q}_i.\texttt{v}}, H^\prime_{\texttt{q}_i, b}  ,H^\prime_{\texttt{q}_i \texttt{.depth} } $ from the other party  \label{line:end-verifiction}
    
        \EndFor

        \;
        
        \If{  $H_k \neq H_k^\prime $ } \label{line:checkk}
		 \State $E \leftarrow E+1 $ \itComment{If hashes of $k$ don't agree, increase Error count}
        \Else 
		\For{$i = 1,2,3$}
  \leftComment{If the hashes related to $\MPi$ agree with at least one of the other party's MP hashes, increase vote for MP$i$:}
   
            \If  {  there is a $j\in \{1,2,3\}$  such that $   H_{\texttt{q}_i.\texttt{v}}  = H^\prime_{\texttt{q}_j.\texttt{v}}$ and  $H_{\texttt{q}_i, b} = H^\prime_{\texttt{q}_j, b}  $  and $H_{\texttt{q}_i \texttt{.depth} } = H^\prime_{\texttt{q}_i \texttt{.depth} }$} 
            \label{line:ugly}
            			\State $\vi{i} \leftarrow \vi{i}+1$ \label{line:updatevi} 
			\EndIf
		\EndFor
        \EndIf

       \; 
\color{green!40!black}
        \State\headerComment{Computation Phase} 
        \If{  $ k = 1 $ and $H_{(\phash, \pseed, \pt, \piter)} = H_{(\phash, \pseed, \pt, \piter)}^\prime$  and $E = 0 $ and $H_{\pv} = H_{\pv}'$ and \texttt{ Rew = False }}  
        \leftComment{Simulate $\Pi$ for $r$ rounds starting from $\pv$; get a transcript $\sigma$ and a new state $v$: }\label{line:checksmallhash}
        	\State  $ \sigma, v \leftarrow $ \textsc{Simulate}$(\Pi, \megastate, r)$ \label{line:simulate}

		\leftComment{Update \texttt{p}:}
   \State $T = T\circ (\sigma \circ \Icount)$
        
            \State  $\piter = \Icount$ \label{line:updateiter}
        	\State \texttt{p.v } $\leftarrow$  $v$ 
		\State\texttt{p.depth} $\gets \pdepth + 1$  
        \State \textsc{MaintainAvMPs}(\texttt{p}, $\M$, $\Mem$, $\texttt{add=True}$)  \label{line:maintain-comp}
        	 \State Reset Status: $  k, E , \vone , \vtwo ,\vthree \leftarrow  0$ \label{line:reset1}
        \Else
        	\State Send ``$0$'' for $r$ rounds.   \itComment{Run $r$ dummy rounds.} \label{line:dummy-rounds}
		
        	\State $\texttt{Rew} \leftarrow \texttt{True}$ \itComment{Record that there is disagreement and we need to rewind.} 
        \EndIf
\;\\
\color{violet!80!black}
        \State \headerComment{Transition Phase}
        \If{ $2E \geq k$ }  \label{line:sync-condition}
		\State Reset Status: $ k, E , \vone , \vtwo, \vthree  \leftarrow 0$ \label{line:reset2} 
        \ElsIf{ $ k = 2^{j+1} - 1$} 

        \If{$k> 1$}
		\For{i = 3,2,1} \label{line:backwardsFor}
			\If{ $\vi{i} \geq (0.4) \cdot 2^j $ and $\MPi$ is available } 
      			    \State   $\texttt{p} \gets \texttt{q}$, where $\texttt{q} \in \Mem$ has $\texttt{q.depth} = \MPi$ \itComment{Jump to $\MPi$} \label{line:back-jump}
                        \State $ \Bsim = \Bsimp$

                    \State \textsc{MaintainAvMPs}(\texttt{p}, $\M$, $\Mem$, $\texttt{add=False}$) 
                        \State \texttt{Rew} $\gets$ \texttt{False}
				\State Break
         		\EndIf
		\EndFor
        \EndIf
        \State Reset Votes: $ \vone , \vtwo , \vthree \gets 0 $ \label{line:reset-vi}
	\EndIf

 \itComment{(End of For-loop over $\Iblock$ iterations).}
\; \label{line:iterend}
\EndFor
\color{orange}
     \State\headerComment{Big Hash Computation} \label{line:big-hash-phase} 
    \State $\Rblock^{b,1}    \gets \textsc{RandomnessExchange}(\ell =  O(\log d + \log 1/\eps)   , \texttt{send=True})$ \label{line:alice-send-bigrand} 
    
    \itComment{For Bob, $\texttt{send = False}$}
    \State $\Rblock^{b,2} \gets \textsc{RandomnessExchange}(\ell =  O(\log d + \log 1/\eps)   , \texttt{send=False})$ \label{line:bob-send-bigrand} 
    
     \itComment{For Bob, $\texttt{send = True}$ 
     } 
    \State $\Rblock^b = \Rblock^{b,1} \circ \Rblock^{b,2}$\label{line:Rblock-append}

    \For{every $\texttt{p}\in \M$ simulated during the current block}
    \State $\phash= h^b( (\phash,\pseed, \pt,\piter ),\Rblock^{b}) $ \label{line:big-hash-comp}
    \State $\pseed = \Rblock^{b}$
    \State $\pt = \bot $ 
    \EndFor
 \color{black}

\EndFor

 \itComment{(End of For-loop over $\Btotal$ blocks).}
\;
\State\headerComment{All done!}
\State Celebrate!  We have successfully simulated $\Pi$!
\end{algorithmic}
\end{breakablealgorithm}

\normalsize

\begin{remark}[What if $r$ does not divide $d$?]\label{rem:padding-length}
    If the original protocol did not have length $d$ that is a multiple of $r$, we ``pad it out'' to have length $r \cdot \lceil d/r\rceil$ by having one party (say, Alice) send dummy symbols (say, $0$) for $r \cdot \lceil d/r\rceil-d$ rounds. Note that as $r = O(1/\sqrt{\eps})$ this will not prevent us from achieving our target rate. 
\end{remark}

    \begin{remark}[Making $\Pi'$ alternating if $\Pi$ is]\label{rem:alternating} 
    As noted in Remarks~\ref{rem:speaking_order} and \ref{rem:speaking_order_redux}, if the original protocol $\Pi$ is alternating, then our robust protocol can be made alternating without affecting the statement of Theorem~\ref{thm:main}.  
    To do this, if Alice is supposed to speak twice in a row, we will simply insert a dummy round for Bob in between, and vice versa.  To see this does not meaningfully affect the rate, notice that, if $\Pi$ is alternating, the only parts of $\Pi'$ that where extra rounds need to be inserted are the parts outside of simulating $\Pi$ (for example, in the randomness exchange routines, during which Alice does all of the talking).  However, our main result about the rate shows that the amount of communication taken up by these parts of the protocol is at most $O( d \sqrt{ \eps \log\log(1/\eps)})$.  Indeed, the total amount of communication is $d + O(d \sqrt{\eps \log\log(1/\eps)})$, and at least $d$ must be taken up by simulating $\Pi$, which is alternating by assumption.  But this means that the number of inserted rounds will be at most $O( d\sqrt{\eps \log\log(1/\eps)})$, and the total communication is still
    \[ d(1 + O(\sqrt{\eps \log\log(1/\eps)}), \]
    as claimed in Theorem~\ref{thm:main}.
    \end{remark}

Before we move on, we prove a few quick statements that follow immediately from the pseudocode.  The first is just the observation that our choice of parameters for the hash functions are consistent with how they are used and with Theorems~\ref{thm:hash} and \ref{thm:pseudorand}. 
\begin{claim}\label{claim:bighashokay}
    The parameters chosen in Algorithm~\ref{alg:init} for $t_i, o_i, \sd_i$ for $i=1,2,3$ and the corresponding definitions of $h^s$ and $h^b$ (using Theorems~\ref{thm:hash} and \ref{thm:pseudorand}) are consistent with how $h^s$ and $h^b$ are used in Algorithm~\ref{alg:adaptive}.
\end{claim}
\begin{proof}
We begin with the big hash, and then analyze the small hash.
\paragraph{The big hash $h^b$.} 

We recall that $t_3$ is set to $\Theta(r \log d + \log^2 d)$ in Line~\ref{line:big-h-par} in Algorithm~\ref{alg:init}.  We first verify that (a) this is large enough to accommodate  the inputs being hashed by $h^b$ in line~\ref{line:big-hash-comp}.  We will also be more explicit about the constants in the $\Theta(\cdot)$ notation.  Next, (b) we will verify that the choice of 
\begin{equation}\label{eq:need}
\sd_3 = \Chashtwo (o_3 + \log(t_3))
\end{equation}
in Algorithm~\ref{alg:init} (which is also what is required by Theorem~\ref{thm:hash}) is feasible.  This is not obvious, since $t_3$ will depend on $\sd_3$, as the seed from the previous round is hashed in $h^b$.

We begin with (a).
The inputs to $h^b$ include:
    \begin{itemize}
         \item $\pt$: A partial transcript of the megastate within the current block. The length of $\pt$ is at most $r \log d + \log^2d$ bits.  Indeed, this includes at most $\Iblock = \lceil \log d \rceil$ entries, each of which contains $r$ bits of transcript, along with $\log(\Itotal) \leq \log(d)$ bits to specify the iteration number.\footnote{We omit the floors and ceilings for notaitonal clarity.}
        \item $\Icount$: A counter with $\log(\Itotal) \leq O(\log(d))$ bits.
        \item $\phash$: A previous hash with $o_3$ bits.  As above, we choose $o_3 = C_b \log d$.
        \item $\pseed$:  This has length $\sd_3$, which as above is $O(\log d + \log 1/\eps)$.  Fix a constant $C$ so that $\sd_3 \leq C\log(d/\eps)$.  (As noted above, the reason to make this constant $C$ explicit while allowing big-Oh notation elsewhere is that $\sd_3$ depends on $t_3$ which depends on $\sd_3$, and we need to make sure that there are not circular dependencies in the constants.
        \item $\piter$: the iteration number of the last time the hash of the simulated path is computed, which requires $
        \Theta(\log d)$ bits.
    \end{itemize}
    Altogether, we need $t_3$ to satisfy 
     \[ t_3 \geq r \log d + \log^2d + O(\log d)  + C\log(d/\eps) = \Theta(r \log d + \log^2 d) + C \log(d/\eps).\]
We observe that this is indeed $\Theta(r \log d + \log^2 d)$, establishing our first goal (a).  Now we turn to (b), verifying that \eqref{eq:need} is feasible.  We need to satisfy
\[ \sd_3 = \Chashtwo( o_3 + \log(t_3) ),\]
and it suffices to show that
\[ C \log(d/\eps) \geq \Chashtwo\inbrak{C_b \log d + \log\inparen{O(r \log d + \log^2 d) + C\log(d/\eps) } },\]
where $C$ is the constant chosen above.
This simplifies to
\begin{align*}
    C \log(d/\eps) &\geq \Chashtwo \inbrak{ C_b \log d + \log r + O(\log\log d) + \log(C) + \log\log(1/\eps)}. \\
    C \log(d/\eps) &\geq \Chashtwo\inbrak{O(\log(d/\eps)) + \log(C)},
\end{align*}
where in the last line we have condensed any terms that do not depend on $C$ into the big-Oh notation.
By choosing $C$ sufficiently large (relative to $\Chashtwo$ and $C_b$), and taking $d$ to be large enough, we see that the above can be satisfied.  This establishes (b).

    \paragraph{The small hash $h^s$.}
    Next, we verify the parameters in the small hash $h^s$.  We begin with $h_1$, the ``inner'' part of the hash.  This hash function takes as input several possible things, so we need to choose $t_1$ to be as large as the largest of them.  The possible inputs are:
    \begin{itemize}
        \item $k$: A counter that is at most $\Iblock$, and so takes at most $\log(\Iblock) = \log\log d$ bits.
        \item $\pv$: A state from $\Pi$. This takes $\log s$ bits.
        \item The tuple $(\phash, \pseed,\pt , \piter)$ (for both the current mega-state $\megastate$ and also for the meeting point candidate mega-states $\texttt{q}_i$): The hash outputs and seeds come from the big hash, $h^b$.  By the analysis above, $\phash$ has $o_3 = O(r \log d)$ bits, $\pseed$ has $O(C_b \log d + \log 1/\eps)$ bits, and $\pt$ has length at most $r\log d + \log^2 d$ bits.
        The iteration counter $\piter$ has at most $\log \Itotal = O(\log d)$ bits, and so this whole tuple has $O(r \log d + \log^2  d )$ bits.
        \item The depth numbers $\texttt{q}_i.\texttt{depth}$.  These are integers in $[d/r]$, and thus have at most $\log(d)$ bits.
    \end{itemize}
    Thus, if we take $t_1 = \Theta(\log s + r \log d + \log ^ 2 d )$, we can accommodate all of the inputs.

    Next, we consider the seed lengths and output lengths for $h_1$.  
    We choose $o_1 = 2 \log(1/\eps)$, and then Theorem~\ref{thm:pseudorand} tells us that we should take $\sd_1 = 2 t_1 o_1$, which is indeed how we define it.   

    Finally we turn to $h_2$, the ``outer'' hash function in the small hash.  We must have $t_2 = o_1 = 2 \log (1/\eps)$ so that the output of $h_1$ fits as the input of $h_2$, and this is indeed the case.  We choose $o_2$ to be a constant, $\Chash$, to be determined later. Finally, Theorem~\ref{thm:hash} tells us that with these choices of $t_2, o_2$, we must take $\sd_2 = \Chashtwo( o_2 + \log t_2 )$, which is indeed what we do.
\end{proof}

The next lemma is about when a point $p$ is removed from $\M$.
\begin{lemma} \label{lem:forget-avmps-z}
Suppose Alice's depth at time $t$ is $\Aiter(t)$, and let $p \in M_A$ be a $j$-stable meeting point.   Let  $t_p < t$ being the last time Alice reached depth $p$ before time $t$. Then $p \notin \M_A(t)$ if and only if there exists a time $t^\prime \in (t_p, t]$ such that $\ell_A(t^\prime) \geq p + 2^{j+1}$.
The analogous statement also holds for Bob.
\end{lemma}
\begin{proof}
    Notice first that $\M_A$ is updated every time Alice updates her current depth $\Aiter$  (either in the Computation Phase after she simulates $\Pi$, or in the Transition Phase after she jumps).  

 To prove the lemma, we observe that $p \not\in \M_A(t)$ if and only if there is some time $t' \in (t_p, t]$ so that $p$ is removed from $\M_A$ at that time.  Indeed, since by definition $t_p$ is the last time that Alice simulates $p$, $p$ is added to $\M_A$ at time $t_p$, and will not be added again until after time $t$ if it is forgotten.  The point $p$ will be removed from $\M_A$ at time $t'$ if and only if $p \not\in M_{\Aiter(t')}$, as this is the condition in \textsc{MaintainAvMPs} (Algorithm~\ref{alg:Maintain-AvMPs}).

 By Lemma~\ref{lem:inMa}, $p \not\in M_{\Aiter(t')}$ if and only if $\Aiter(t') \geq p + 2^{j+1}$. 
 Thus, $p \not\in \M_A(t)$ if and only if there is some time $t' \in (t_p, t]$ so that $\Aiter(t') \geq p + 2^{j+1}$, as desired.
\end{proof}

%% file: analysis_start.tex
In this section we analyze our protocol $\Pi'$, and prove Theorem~\ref{thm:main}.
We begin in Section~\ref{subsec:def-lem} with some useful Lemmas and Definitions.  Then in Section~\ref{sec:maintech}, we introduce sneaky attacks and state our main technical lemma.  In Section~\ref{sec:progressProof} we define and analyze our potential function $\Phi$.  In Section~\ref{sec:hash} we prove that our hash functions behave as expected.  Finally, in Section~\ref{sec:done} we put everything together and prove Theorem~\ref{thm:main}.

%% file: defs_lems.tex
Before we can analyze the potential function and hash functions we need a few more definitions and notations. Initially we introduce some definitions followed by key lemmas which will be helpful for our analysis. 
\subsubsection{A few more definitions}

 We start with some terminology about \emph{hash collisions}.  Informally, we say that a \defn{small hash collision} occurs whenever the small hash function $h^s$ suffers a collision, and a \defn{big hash collision} occurs whenever the big hash function $h^b$ suffers a collision. Formally, we have the following definitions. 
 
\begin{definition}[Small and big hash collisions]\label{def:hashcollisions}
\, 
\begin{itemize} 

\item 
Let $I\in [\Itotal] $, We say that a \defn{small hash collision has occurred in iteration $I$} of  Algorithm~\ref{alg:adaptive} either of the following occur:
    \begin{itemize}
    \item There is a hash collision in line~\ref{line:varhash}. That is, at the time $t$ when line~\ref{line:varhash} is executed in iteration $I$:
    \begin{itemize}
        \item for any of the variables    $\var$ from line~\ref{line:var}  we have that $\var_A(t) \neq \var_B(t)$; 
        \item Alice and Bob have the same randomness, i.e., $(\Rblock^s)_A(t) = (\Rblock^s)_B(t)$ and $(\Riter)_A = (\Riter)_B$; and 
        \item Alice and Bob's outcomes are the same, i.e.,  $(H_{\mathtt{var}})_A(t) = (H_{\mathtt{var}})_B(t)$.
    \end{itemize}
    \item There is a hash collision between information releted to different meeting point candidates in any of lines~\ref{line:ms-hash-v}, \ref{line:ms-hash-depth}, \ref{line:ms-hash-small-big-hash}. That is, at the time $t$ when any of those lines are executed:
    \begin{itemize}
        \item Both parties have the same randomness, meaning that $(\Rblock^s)_A(t) = (\Rblock^s)_B(t)$; and
        \item there is a transition candidate for Alice with megastate $\megastate_A$ and a transition candidate for Bob with megastate $\megastate_B$ such that $\megastate_A \neq \megastate_B$ (meaning there exists $\var\in \{\pv, (\phash,\pseed,\pt , \piter) ,\pdepth\}$ such that $\var_A \neq \var_B$ but $H_{\var_A}  = H_{\var_B}$).
    \end{itemize}
    \end{itemize}

\item We say that a \defn{big hash collision occurs at the end of a block $B$} if there is a hash collision  in the big hash exchange phase of block $B$, with respect to the hash function $h^b$ in Line~\ref{line:big-hash-comp}.  That is, this occurs if Alice and Bob's inputs to the hash are different, so
    \[ ( \phash , \pseed , \pt, \piter )_A \neq ( \phash , \pseed , \pt, \piter )_B;\] but their random seeds are the same ($(\Rblock^{b})_A = (\Rblock^{b})_B$ ) and the outcomes in line~\ref{line:big-hash-comp} are the same.

\end{itemize}
\end{definition}

Note that in the Big Hash Computation phase at the end of each block, Alice is responsible for generating half of the randomness $\Rblock^b$ and Bob is responsible for generating the other half.  The reason for this is so that, no matter what the adversary does, the probability that the randomness $\Rblock^b$ from a block $B$ collides with the randomness $\Rblock^b$ from a \emph{different} block $B'$ is very small.  We will work out this probability and union bound over it at the end of the argument.  Until then, we make the following assumption for the rest of the analysis.
\begin{asm}\label{asm:rand-big-neq}
    Fix any two distinct blocks $B$ and $B'$ during the execution of Algorithm~\ref{alg:adaptive}. Let $(\Rblock^b)_A(B)$ and $(\Rblock^b)_B(B')$ be the randomness used for the big hash computation in block $B$ by Alice and in block $B'$ by Bob, respectively.  We assume for the analysis that 
    $(\Rblock^b)_A(B)\neq (\Rblock^b)_B(B')$.
\end{asm}
Again, we stress that this assumption is not an assumption in Theorem~\ref{thm:main}, it is only for the purpose of analysis.  We will remove this assumption in the proof of Theorem~\ref{thm:main} by showing that it holds with high probability.

\begin{remark}\label{rem:mseq}
For two megastates $\megastate$ and $\megastate'$ that belong to the same party, we say that $\megastate \mseq \megastate'$ if the two are equal in every variable \emph{except} possibly $\pseed, \phash, \pt$.  In this case, we say that ``$\megastate$ is equal to $\megastate'$''. 
The reason for this is that once a new megastate $\mathtt{p}$ is ``initialized'' in Line~\ref{line:maintain-comp} and stored in memory, the variables $\piter, \pv, \pdepth$ will not change.  However, during the big hash computation phase at the end of the block, the variables $\pseed, \phash, \pt$ may be retro-actively updated.  Thus, we say that that $\megastate$ and $\megastate'$ are ``equal'' even if these latter variables have changed.
\end{remark}

We will single out some iterations as being ``dangerous.'' We have the following definition:
\begin{definition}[Dangerous iteration]\label{def:danger}
    We say that an iteration $I \in \{1, \ldots, \Itotal\}$ is a \defn{dangerous iteration} if at the beginning of iteration $I$ (after $k \gets k+1$ in Line~\ref{line:iterstart}), it is the case that $\ell^- > 0$ or $k_A > 1$ or $k_B > 1$.
\end{definition}

\begin{lemma}\label{lem:notdangerous}
Suppose that iteration $I \in [\Itotal]$ is not dangerous, and suppose that the adversary does not introduce any corruptions in iteration $I$ and there are no corrupted randomness.  If Alice and Bob both simulate $\Pi$ in iteration $I$ (that is, Line~\ref{line:simulate} executes for both), then their simulation in iteration $I$ is correct, and  $\ell^-$ will still be equal to zero after the Computation Phase completes.
Notice that this holds whether or not a small hash collision occurs. 
\end{lemma}
\begin{proof}
Since $I$ is not a dangerous iteration, by definition $\ell^- = 0$ at the beginning of the iteration.  Thus, Alice and Bob's simulated paths $\simPath_A$ and $\simPath_B$ agree.  Therefore, if the adversary does not introduce any corruptions during their simulation of the next $r$ rounds, their paths will continue to agree. Let  $\megastate_A $ be the mega-state of Alice at the end of this iteration and define $\megastate_B$ similarly for Bob. Given that there is no corrupted randomness during this iteration, at the end of this iteration we have that, for every $\var \in \{ \pv ,\pdepth , \pt, \piter, \phash , \pseed \}$ , $\var_A = \var_B$ and we will continue to have $\ell^- = 0$.

\end{proof}
\begin{definition} \label{def:dangerous-small-hash-coll}
    We say a small hash collision is \defn{dangerous} if the hash collision happens during a dangerous iteration. 
\end{definition}

\begin{definition}[Corrupted Randomness]\label{def:corrupted_rand}
\ 
\begin{itemize}
\item Let $B \in [\Btotal]$.  We say that \defn{block $B$ has corrupted randomness} if, at the end of block $B$ in Algorithm~\ref{alg:adaptive}, either $(\Rblock^s)_A \neq (\Rblock^s)_B$ or $(\Rblock^b)_A \neq (\Rblock^b)_B$. 
\item Let $I \in [\Itotal].$  We say that \defn{iteration $I$ has corrupted randomness} if, at the end of iteration $I$, either
 $(\Rblock^s)_A \neq (\Rblock^s)_B$ or $(\Riter)_A \neq (\Riter)_B$.
\end{itemize}
\end{definition}

The next notion we need is a ``shadow variable'' in Algorithm~\ref{alg:adaptive}, called the \defn{Bad Vote Count} ($\BVC$).

\begin{definition}[Bad Vote Count ($\BVC$)]\label{def:BVC}
Fix a party (say, Alice) executing Algorithm~\ref{alg:adaptive}.  We define a ``shadow variable'' $\BVC$ as follows. 
\begin{itemize} 
    \item $\BVC$ is initialized to $0$ at the beginning of Algorithm~\ref{alg:adaptive}.
    \item Whenever Alice \emph{resets votes} (Lines \ref{line:reset1}, \ref{line:reset2}, and \ref{line:reset-vi}), $\BVC$ also resets to $0$.
    \item In Line~\ref{line:updatevi}, $\BVC$ increments by one if any of the following occur:
    \begin{itemize} 
        \item $\vi{i}$ is incremented but 
        \[    \ttq_i \notin \{ \ttq_{1,B} , \ttq_{2,B} , \ttq_{3,B} \}\]
        where $\ttq_i = \ms{\MPi}$ is as defined on Line~\ref{line:qi}, and $\ttq_{j,B}$ is Bob's value of $\ttq_j$ during this iteration.
         \item $\vi{i}$ is not incremented but 
    \[    \ttq_i \in \{ \ttq_{1,B} , \ttq_{2,B} , \ttq_{3,B} \}\]

    \end{itemize}

\end{itemize}
\end{definition}
As usual, we use $\BVC_A$ to refer to Alice's ``copy'' of $\BVC$, and $\BVC_B$ to refer to Bob's ``copy.''  (Here, ``copy'' is in quotes, because of course Alice cannot compute $\BVC_A$; it is for analysis only.) 

\begin{lemma} Let $I \in [\Itotal]$. If the $\BVC$ (for either Alice or Bob) increments during iteration $I$, then one of the following occured: 

\begin{itemize}
    \item A small hash collision occured in iteration $I$.
    \item A corruption was introduced during iteration $I$.
    \item Iteration $I$ had corrupted randomness. 
\end{itemize}
\label{lem::BVC-coll-corr}
\end{lemma}
\begin{proof}
    Fix a party, say, Alice.  Suppose that the variable $\BVC$ increments for Alice.
    By definition, this means that for some $i$, either $\vi{i}$ increments when        $   \ttq_i \notin \{ \ttq_{1,B} , \ttq_{2,B} , \ttq_{3,B} \}$, or it doesn't increment when $\ttq_i \in \{ \ttq_{1,B} , \ttq_{2,B} , \ttq_{3,B} \}$.

   In Algorithm~\ref{alg:adaptive}, $\vi{i}$ only increments if there is some index $j \in \{1,2,3\}$ such that $H_{\texttt{q}_i.\texttt{v}}  = H^\prime_{\texttt{q}_j.\texttt{v}}$, $H_{\texttt{q}_i.\texttt{depth}}  = H^\prime_{\texttt{q}_j.\texttt{depth}}$, and  $H_{\texttt{q}_i, b} = H^\prime_{\texttt{q}_j, b} $.

   If iteration $I$ has corrupted randomness, we are done, so suppose that it does not.
     Suppose that $\BVC$ increments, and consider the following two scenarios.
        \begin{itemize}
            \item \sloppy Suppose that $\vi{i}$ increments, so such a $j$ exists, but  $\ttq_i \neq \ttq_{j,B}$. 
              Since $\ttq_i \neq \ttq_{j,B}$, then Alice's $\ttq_i$ and Bob's  $\ttq_{j,B}$ differ on at least one of their three hashed attributes: $\texttt{v}$, 
            $(\texttt{prev-seed}, \texttt{prev-hash}, { \texttt{T} }, \texttt{iter}$), or $\texttt{depth}$. If they differ on $\texttt{v}$, then since $H_{\ttq_i.\texttt{v}} = H'_{\ttq_j.\texttt{v}}$, either we must have  $H_{\ttq_i.\texttt{v}} = (H_{\ttq_j.\texttt{v}})_B$, in which case there was a small hash collision (as we are assuming uncorrupted randomness in this iteration); or else $(H_{\ttq_j.\texttt{v}})_B \neq H_{\ttq_j.\texttt{v}}'$, in which case what Alice received was different than what Bob sent, so the adversary must have introduced corruptions. The same argument holds for the other two attributes.

            \item On the other hand, suppose that $\vi{i}$ does not increment, so $\ttq_i = \ttq_{j,B}$ for some $j$, but either $H_{\ttq_i.\texttt{v}} \neq H'_{\ttq_j.\texttt{v}}$, $H_{\texttt{q}_i.\texttt{depth}}  \neq H^\prime_{\texttt{q}_j.\texttt{depth}}$, or $H_{\ttq_i,b} \neq H'_{\ttq_i,b}$.  As we are assuming that iteration $I$ has uncorrupted randomness, then all hashes of attributes of $\ttq_i$ and $\ttq_{i,B}$ agree, so this implies that the adversary must have introduced a corruption.

        \end{itemize}
 This completes the proof.

\end{proof}
\paragraph{The potential function.}
Now we can define the potential function $\Phi$ that we will use to track Alice and Bob's progress through the protocol.

Recall that for a variable $x$ appearing in Algorithm~\ref{alg:adaptive}, $x_A$ refers to Alice's copy, $x_B$ refers to Bob's copy, and $x_{AB} = x_A + x_B$.  Recall that the definitions of $\ell^+, \ell^-,$ and $L^-$ are given in Section~\ref{sec:prelim}.
Define:
\begin{equation}\label{eq:potential}
\Phi = \begin{cases} \ell^+ - C_3 \ell^- - C_2 L^- + C_1 k_{AB} - C_5 E_{AB} - 2C_6 \BVC_{AB} & \text{if } k_A = k_B \\
\ell^+ - C_3 \ell^- - C_2 L^- - 0.9 C_4 k_{AB} + C_4 E_{AB} - C_6 \BVC_{AB} & \text{if } k_A \neq k_B
\end{cases}
\end{equation}
where $C_1, \ldots, C_6$ are constants that will be chosen later. We note that this is very similar to the potential function that was studied in the paper \cite{H14}, and is the same as that in manuscript \cite{HR18}.

 As per our conventions in Section~\ref{sec:notationABtime},
for a time $t$, we use $\Phi(t)$ denote the value of $\Phi$ at time $t$. For an iteration $I \in [\Itotal]$, we use $\Phi(I)$ to denote the value of $\Phi$ at the \emph{end} of iteration $I$ (in Line~\ref{line:iterend} of Algorithm~\ref{alg:adaptive}).

\subsubsection{Useful lemmas} \label{sec:usefullemmas}

In this section, we prove a few lemmas that will be useful in our analysis of \protosneakyf{s} and the potential function.

Our first lemma shows that our implementation of the ``hash chaining'' strategy is effective.  In particular,  Lemma~\ref{lem:bigh-hash-corruption} below shows that, assuming no big hash collisions ever occur, the values of $\phash$, $\pseed$, $\pt$, and $\piter$ in Algorithm~\ref{alg:adaptive} will pick up on any discrepancies between Alice and Bob's simulated paths. Before we state Lemma~\ref{lem:bigh-hash-corruption}, we state and prove two claims that will help us in the proof of Lemma~\ref{lem:bigh-hash-corruption}.

\begin{claim}\label{claim:neq-in-and-out-of-block}
 Suppose there are no big hash collisions throughout the entire protocol $\Pi'$. Fix any block $B$, let $t_1$ be the time just before the big hash computation phase starts in line~\ref{line:big-hash-comp}, and let $t_2$ be the time at the end of the block $B$. Consider any two megastates $\megastate^*_A$ and $\megastate^*_B$ from Alice and Bob's memory, respectively. Suppose that 
 \begin{align*}
   & (\pahash , \paseed , \pat , \paiter ) (t_1) \\&\qquad\qquad\qquad\neq (\pbhash, \pbseed, \pbt, \pbiter) (t_1) 
 \end{align*}
Then 
  \begin{align*}
   & (\pahash , \paseed , \pat , \paiter ) (t_2) \\&\qquad\qquad\qquad\neq (\pbhash, \pbseed, \pbt, \pbiter) (t_2).
 \end{align*}
\end{claim}
\begin{proof}
First notice that big hash computation during the end  of the block does not modify variable $\piter$ in either mega-state. As a result, if $\paiter(t_1) \neq \pbiter(t_1) $ then it immediately follows that $\paiter(t_2) \neq \pbiter(t_2)$, and we are done.

Now assume that $\paiter (t_1) = \pbiter(t_1) = I$. Then, we have two scenarios. Either $I\in B$ or $I\notin B$; we consider each case below.

If $I\notin B$, then none of the variables within the mega-states gets updated during the big hash computation. Hence $\megastate^*_A (t_1) = \megastate^*_A(t_2)$. Similarly for Bob, $\megastate^*_B (t_1) = \megastate^*_B(t_2)$ and the claim follows. 

If $I\in B$, as the tuples do not match at time $t_1$ then there must be a variable $$\var\in \{\operatorname{\mathtt{prev-hash}, \mathtt{prev-seed}, \mathtt{T} }\}$$ such that $\megastate^*_A.\var(t_1) \neq\megastate^*_B.\var  (t_1)$. If the randomness $\Rblock$ is corrupted then at time $t_2$, then by definition $\paseed(t_2) \neq \pbseed(t_2)$, and we are done. On the other hand, if $\Rblock$ is not corrupted then $\var$ is an input to the big hash function, and assuming there are no hash collisions we can conclude that $\pahash(t_2) \neq \pbhash(t_2)$ which completes the proof.
\end{proof}

\begin{claim}\label{claim:piter-neq} 
   Assume that there are no big hash collisions throughout the protocol.  For any time $t$, and any two mega-states $\megastate^*_A$ and $\megastate^*_B$ from Alice and Bob's memory at time $t$, if $\paiter(t) \neq \pbiter(t)$ then there exists a variable $\var\in \{\pseed, \phash, \pt\}$ such that $\megastate^*_A.\var(t) \neq \megastate^*_B.\var(t)$. 
\end{claim}

\begin{proof}
Throughout the proof, all variables are referenced at time $t$.
 Without loss of generality, assume that $\paiter > \pbiter$.

 Then we have two cases: 

 \paragraph{Case 1: At time $t$, at least one of $\pat$ or $\pbt$ is not equal to $\bot$.} Then we claim that $\pat \neq \pbt$.  Notice that at time $t$, $\pat \neq \bot$.  Indeed, given our assumption that $\paiter > \pbiter$, if $\pat = \bot$ then $\pbt = \bot$, and we are assuming that they are not both $\bot$.   Moreover, $\pat$ includes $\paiter$ in the last piece of the partial simulated path that stored in $\pat$.  On the other hand, $\pbt$ is either $\bot$ or includes $\pbiter$ in the last piece of simulated path stored in $\pbt$.  So we can conclude that $\pat \neq \pbt$, establishing the claim.
 \paragraph{Case 2: At time $t$, we have that $\pat = \pbt = \bot$.}  Then $t$ must be after Alice computes the big hash values at the end of the block corresponding to $\paiter$.
 We claim that $(\pahash, \paseed ) \neq (\pbhash , \pbseed ).$ 
 Indeed, when Alice computes the big hash value for $\megastate^*_A$, her input is $\paiter$, while Bob's input is $\pbiter$ when he computes the big hash value for $\megastate^*_B$. Hence, Alice and Bob have different inputs to the big hash function. If $\paiter$ and $\pbiter$ are in different blocks then,  by Assumption~\ref{asm:rand-big-neq}, $\paseed \neq \pbseed.$

 On the other hand,  if iteration $\pbiter$ is within the same block as iteration $\paiter$, then either $\Rblock^b$ is corrupted in line~\ref{line:alice-send-bigrand} of Algorithm~\ref{alg:adaptive}, or it is not.  If it is corrupted, then $$\paseed \neq \pbseed,$$ establishing the claim.  On the other hand, if $\Rblock^b$ is not corrupted then, as we are assuming there are no big hash collisions, we can conclude that $\pahash \neq \pbhash$ since Alice and Bob have different inputs for the big hash function. This completes the proof.
\end{proof}

Now we can state and prove Lemma~\ref{lem:bigh-hash-corruption}, which informally says that if there is a corruption---that is, if Alice and Bob's simulated paths do not match---then the big hash will eventually catch it.

\begin{lemma}[The big hash catches all corruptions]\label{lem:bigh-hash-corruption}
Suppose there are no big hash collisions throughout the entire protocol $\Pi'$.  
Fix any time $t$, and let $q_A,q_B$ be points so that $1 \leq q_A \leq \Aiter(t)$ and $1 \leq q_B \leq \Biter(t)$.
Define
\[ \megastate_A^* = \ms{q_A}_A(t) \qquad \megastate_B^* = \ms{q_B}_B(t).\]
That is, $\megastate_A^*$ is the Alice's megastate corresponding to depth $q_A$ at time $t$.\footnote{Notice that Alice may not have $\megastate_A^*$ in memory at time $t$; $\megastate_A^*$ is just the megastate at depth $q_A$ that Alice most recently computed before time $t$.}

Suppose that $\simPath_A^{(\leq q_A)} \neq \simPath_B^{(\leq q_B)}.$
Then 
\[ (\pahash, \paseed,\pat, \paiter) \neq (\pbhash, \pbseed, \pbt ,\pbiter). \]
\end{lemma}
\begin{proof}
We begin with some useful notation.
     Recall that the simulated paths $\simPath_A := \simPath_A^{(\leq q_A)}$ and $\simPath_B := \simPath_B^{(\leq q_B)}$ store both the simulated transcripts, as well as the iteration numbers in $\Pi'$ where each chunk of transcript was simulated. Define $\mathcal{I}_\simPath$ as the set of  iteration numbers included in the simulated path $\simPath$. 
 Define $\mathcal{I}_{\simPath_{AB}} = \mathcal{I}_{\simPath_A} \cup \mathcal{I}_{\simPath_B}$.

With this notation, we have the following claim.
\begin{claim}\label{cl:useful}

Let $I \in \mathcal{I}_{\simPath_A}$.
Let \[ I_A = \max\{I' \in \mathcal{I}_{\simPath_A} \,:\, I' < I \}\]
be the last iteration in $\mathcal{I}_{\simPath_A}$ before $I$ where Alice simulated $\Pi$.  Let $\tcmp$ be the time at the beginning of the computation phase of iteration $I$.  Then, using the notation from Remark~\ref{rem:mseq}, we have
\[ \megastate_A(\tcmp)  {\ \mseq\ }   \megastate_A(I_A).\]
That is, Alice's megastate $\megastate_A$ at time $\tcmp$ is the same as  her megastate at the end of iteration $I_A$ ,up to possible differences in the variables $\phash, \pseed, \pt$.

\end{claim}
\begin{proof}
Without loss of generality, we focus on Alice; the case for Bob is identical.

Let $q = \pdepth_A(I)$ be the point that Alice simulated in iteration $I$.
Notice that during iteration $I_A$, Alice simulated the point $q-1$; this is because in a simulated path, the simulated points are consecutive.

Further, we claim that $I_A$ is the \emph{last} iteration (either in $\mathcal{I}_{\simPath_A}$ or not) before time $t$, where Alice simulates the point $q-1$.  If not, suppose that at some iteration $I' > I_A$, Alice also simulated the point $q-1$.  But this would overwrite the simulated path $\simPath_A$ at the depth $q-1$, and $I_A$ would not appear in $\mathcal{I}_{\simPath_A}$, a contradiction.
This implies that $I_A$ is the last iteration (either in $\mathcal{I}_{\simPath_A}$ or not) before iteration $I$ where Alice simulates the point $q-1$.  

Since Alice simulated depth $q$ in iteration $I$, she must have began iteration $I$ at depth $q-1$.  Since $I_A$ was the last iteration that Alice simulated depth $q-1$, the megastate that she begins iteration $I$ with is the same as the megastate that she ended iteration $I_A$ with.  Since the megastate does not update until the computation phase, (except possibly the variables $\phash, \pseed, \pt$ that may be updated during the big hash phase at the end of the block, if $I$ and $I_A$ are in different blocks), the megastate she begins iteration $I$ with is equal to her megastate at time $\tcmp$.  We conclude that $\megastate_A(\tcmp) {\ \mseq\ } \megastate_A(I_A),$
as desired.
\end{proof}

Now we continue with the proof of the lemma.
Notice that if $\paiter \neq \pbiter$, then we are done, so we assume that $\paiter = \pbiter$.  Thus, it suffices to show that in iteration
\[ I^* := \paiter = \pbiter,\]
we have
\[ (H_A, R_A, \Bsim^{\leq \paiter} , I_{H_A})(I^*) \neq (H_B, R_B, \Bsim^{\leq \pbiter} , I_{H_B})(I^*),\]
where $(H_A, R_A, \Bsim^{\leq \paiter} I_{H_A})(I)$ is defined as the tuple $(\phash_A, \pseed_A, \pt_A , \piter_A)(I)$ at the end of iteration $I$.

The proof proceeds by induction.
To set this up, suppose that $\simPath_A \neq \simPath_B$, and 
 let $\Icorr$ be 
 the smallest $I' \in \mathcal{I}_{\simPath_{AB}}$ such that either $I' \not\in \mathcal{I}_{\simPath_A} \cap \mathcal{I}_{\simPath_B}$, or $\sigma_A(\Icorr) \neq \sigma_B(\Icorr)$.  (Here, $\sigma_A(\Icorr)$ denotes the transcript $\sigma$ that Alice simulated in iteration $\Icorr$, and similarly for Bob.) That is, $\Icorr$ is the earliest iteration of $\Pi'$  in which this simulated path encountered a problem, either because one party simulated while the other did not; or both simulated but simulated divergent paths.

 We now prove by induction on $I \in \mathcal{I}_{\simPath_{AB}}$ that
 \begin{equation}\label{eq:IH}
 \forall I \in \mathcal{I}_{\simPath_{AB}} \text{ such that } \Icorr\leq I,\ \ \  (H_A,R_A,\Bsimp_A,I_{H_A})(I) \neq (H_B, R_B,\Bsimp_B, I_{H_B})(I).
 \end{equation}

 \paragraph{Base Case.}
 Let $I = \Icorr$.  There are two cases, either $I \in \mathcal{I}_{\simPath_A} \cap \mathcal{I}_{\simPath_B}$ or it is not.  
 \begin{itemize}
     \item [1.] In the first case, $I \in \mathcal{I}_{\simPath_A} \cap \mathcal{I}_{\simPath_B}$. Then by the definition of $\Icorr$, $\sigma_A(I) \neq \sigma_B(I)$. Then we have that $\Bsimp_A \neq \Bsimp_B$  and hence,
 \[ (R_A, H_A, \Bsimp_A,\IHA)(I) \neq (R_B, H_B, \Bsimp_B, \IHB)(I),\]
establishing \eqref{eq:IH} for the base case of $I = \Icorr$.
\item[2.] In the second case, $I \not\in \mathcal{I}_{\simPath_A} \cap \mathcal{I}_{\simPath_B}$.  Assume without loss of generality that $I \in \mathcal{I}_{\simPath_A}$, but not in $\mathcal{I}_{\simPath_B}.$  This means that Bob did not simulate $\Pi$ during iteration $I$, which implies that $\piter_B(I) \neq I$.  In particular, $\IHA \neq \IHB$, and we again conclude that 
 \[ (R_A, H_A, \Bsimp_A, \IHA)(I) \neq (R_B, H_B, \Bsimp_B, \IHB)(I),\]
establishing \eqref{eq:IH} for the base case of $I = \Icorr$.
\end{itemize}
 \paragraph{Inductive Step.}
Fix $I \in \mathcal{I}_{\simPath_{AB}}$ so that $I > \Icorr$, and assume that \eqref{eq:IH} holds for all $I' < I$ in $\mathcal{I}_{\simPath_{AB}}$.

Again we have two cases, either $I \in \mathcal{I}_{\simPath_A} \cap \mathcal{I}_{\simPath_B}$ or it is not.  

\begin{itemize}
  \item [1.]  Suppose that $I \in \mathcal{I}_{\simPath_A} \cap \mathcal{I}_{\simPath_B}$, let 
 \[ I_A = \max\{ I' \in \mathcal{I}_{\simPath_A} \,:\, I' < I \},\]
 and define $I_B$ similarly.  That is, $I_A$ is the iteration where Alice last simulated $\Pi$ before $I$, and $I_B$ is the iteration where Bob last simulated $\Pi$ before $I$. 

 Let $\tcmp$ be the time at the beginning of the computation phase of iteration $I$.  Then from Claim~\ref{cl:useful} we conclude that
 \[ \megastate_A(\tcmp) \mseq \megastate_A(I_A) \qquad \text{ and } \qquad \megastate_B(\tcmp) \mseq \megastate_B(I_B).\]
Now, either (a) $I_A \neq I_B$ or (b) $I_A = I_B$. 
\paragraph{Case (a) $I_A \neq I_B$.} From Claim~\ref{claim:piter-neq} we know that there exists a variable $\var\in\{\phash$, $\pseed, \pt\}$ such that $\megastate_A.\var(\tcmp) \neq \megastate_B.\var (\tcmp)$. 
During iteration $I$, the variables $\phash$ and $\pseed$ do not change. As a result, if $\var\in \{\phash, \pseed\}$, then the lemma follows. If $\var = \pt$, then during $I$ both Alice and Bob append a piece of transcript $(\sigma_A,I)$ and $(\sigma_B,I) $ to $\pt_A$ and $\pt_B$, respectively. However if $\megastate_A.\mathtt{T} (\tcmp)\neq \megastate_B.\mathtt{T} (\tcmp)$ then we have that $\megastate_A.\mathtt{T}  \circ (  \sigma_A \circ I )  \neq\megastate_B.\mathtt{T}  \circ  ( \sigma_B \circ I)$.  Thus, $\Bsim^{\leq \megastate}_A \neq \Bsim^{\leq \megastate}_B$, therefore 
 \[ (R_A, H_A, \Bsimp_A, \IHA)(I) \neq (R_B, H_B, \Bsimp_B, \IHB)(I),\]
 as desired.

\paragraph{Case (b) $I_A = I_B= I'$.} Then by induction we know that  $ (R_A, H_A, \Bsimp_A, \IHA)(I') \neq (R_B, H_B, \Bsimp_B, \IHB)(I').$ Further, using Claim~\ref{claim:neq-in-and-out-of-block} we have that $ (R_A, H_A, \Bsimp_A, \IHA)(\tcmp) \neq (R_B, H_B, \Bsimp_B, \IHB)(\tcmp)$.  This is because after the end of a block, the values corresponding to mega-states created during that block are fixed. Now during iteration $I$, the parties copy the $\pseed$ and $\phash$ from $\megastate_A(I_A)$ and $\megastate_B(I_B)$, so if those values differ then $(R_A,H_A)(I)\neq (R_B,H_B)(I)$. If the difference is in the variable $\pt$ then both Alice and Bob append another piece of transcript and, similar to Case~(a), $\Bsim_A^{\leq \megastate} \neq \Bsim_B^{ \leq \megastate }$ and the proof is complete.

 \item[2.] Next, suppose that $I \not\in \mathcal{I}_{\simPath_A} \cap \mathcal{I}_{\simPath_B}.$
 In this case, the same argument as in Case 2 of the base case applies.
\end{itemize}

This shows that \eqref{eq:IH} holds for all $I \in \mathcal{I}_{\simPath_{AB}}$. In particular, we have
\[ (H_A, R_A, \Bsimp_A ,  \IHA)(I^*) \neq (H_B, R_B,\Bsimp_B, \IHB)(I^*),\]
which is what we wanted to show.
\end{proof}

\begin{corollary} \label{cor:transcript-ms} At any time $t$ and
    for any $p$ and $p'$, if $\simPath^{(\leq p)}_A \neq \simPath^{(\leq p')}_B$ then $\ms{p}_A \neq \ms{p'}_B$.
\end{corollary} 
\begin{proof}
This is a direct result of Lemma~\ref{lem:bigh-hash-corruption}. Define  $\megastate_A^* = \ms{p}_A(t) $ and $ \megastate_B^* = \ms{p'}_B(t).$
 Since $\simPath_A^{(\leq p)} \neq \simPath_B^{(\leq p')}$, we have that,\[ (\pahash, \paseed,{\pat}, \paiter) \neq (\pbhash, \pbseed, {\pbt}, \pbiter).\]
 The variables in this tuple are variables of the megastates for Alice and Bob. We conclude that $\ms{p}\neq \ms{p'}$, which proves the statement. 
\end{proof} 

\begin{corollary} \label{cor:ms-not-match-then-triplet-not-match}
    If a time $t$ is such that 
    $\ms{\Aiter(t)}_A \neq \ms{\Biter(t)}_B$ then 
    \[
        (H_A,R_A, { \Bsim_A^{ \leq \megastate} }  , \IHA)(t) \neq (H_B,R_B, { \Bsim_B^{ \leq \megastate} } ,\IHB)(t) \ .
    \]
\end{corollary}
\begin{proof}
Let $ \megastate_A  = \ms{\Aiter(t)}_A$  and $\megastate_B = \ms{\Biter(t)}_B$. Notice that if $\pv_A \neq \pv_B$ or $\pdepth_A \neq \pdepth_B$ then $\simPath_A(t) \neq \simPath_B(t)$ and using Lemma~\ref{lem:bigh-hash-corruption} we can conclude that
   \[
        (H_A,R_A,\Bsim_A^{\leq \megastate}  ,\IHA)(t) \neq (H_B,R_B, { \Bsim_B^{\leq \megastate}},\IHB)(t) \ .
    \]
Now if $\pv_A = \pv_B$ and $\pdepth_A = \pdepth_B$, but $\megastate_A \neq \megastate_B$, then there exists a variable $\var$ so that
    $$\var\in \{\piter, \phash,\pt, \pseed\},$$
    such that $\var_A \neq \var_B$. As $\var$ is also one of the variables of $(H,R,I)$ then we can conclude that, 
    
   \[
        (H_A,R_A,\IHA)(t) \neq (H_B,R_B,\IHB)(t) \ .
    \]
    this completes the proof.
\end{proof}

For the remainder of Section~\ref{sec:analysis} and also in Section~\ref{sec:grossproof}, we will analyze Algorithm~\ref{alg:adaptive} under the assumption that \textbf{no big hash collisions ever occur}. In Section~\ref{sec:hash} we will show that the probability that a big hash collision occurs at some point is at most $1/\poly(d)$, which will be an acceptable contribution to the final probability of failure when we prove Theorem~\ref{thm:main} in Section~\ref{sec:done}. 

\begin{lemma}\label{lem:ms-not-matching-sim-ms-not-matching}
    Fix an iteration $I$.  Let $t_0$ be the start of this iteration (Line~\ref{line:iterstart}) and let $t_1$ be the end of this iteration (Line:~\ref{line:iterend}). Further define $\megastate_A(t_0) = \ms{\Aiter(t_0)}_A $ and $\megastate_B(t_0) =  \ms{\Biter(t_0)}_B$. Define $ \megastate_A(t_1)$ and $\megastate_B(t_1)$ similarly for time $t_1$. We claim that, if $\megastate_A(t_0)\neq \megastate_B(t_0)$ and both Alice and Bob simulate during the computation phase of iteration $I$, then $\megastate_A(t_1)\neq \megastate_B(t_1)$.
\end{lemma}
\begin{proof}
Notice that as both parties have simulated during iteration $I$, we have that $k_A = k_B = 1$. As a result, no transition will take place during this iteration. Thus, the only phase affecting the mega-states of  Alice and Bob during this iteration is the computation phase.

    From Corollary~\ref{cor:ms-not-match-then-triplet-not-match}, we have that $ (H_A,R_A,\Bsimp_A, \IHA)(t_0)\neq (H_B,R_B,\Bsimp_B, \IHB)(t_0)$.  During the computation phase, Alice and Bob copy the values of $H$ and $R$ and append this iteration's simulation to the end of $\Bsimp_A$ and $\Bsimp_B$, respectively.  Hence, If any of these variables do not agree with each other at time $t_0$, then they will not agree at time $t_1$ either. Further, from Claim~\ref{claim:piter-neq}, we know that it is impossible for Alice and Bob's megastates to only differ in the variable $\piter$, so they must differ in one of the other variables. Thus the proof is complete.

\end{proof}

We next record the following useful facts about bad spells, formalizing facts that (we hope) are intuitive.
\begin{lemma}[Useful facts about bad spells]  \label{lem:noweirdcoincidences}

Suppose that a point $b$ becomes the divergent point  at time $t_0$, and remains the divergent point until time $t_1$.

Then the following hold.
\begin{enumerate}
 \item If $t_0$ happens within  an iteration (that is, between Lines~\ref{line:iterloop} and Line~\ref{line:iterend} in Algorithm~\ref{alg:adaptive}; in particular, \emph{not} during the Big Hash Computation phase at the end of a block), then one of the following will occur during that iteration:  
    \begin{itemize}
        \item [(1)] $\Aiter=\Biter=b$ and $\ell^- = 0 $ before the computation phase, and afterwards $\max(\Aiter,\Biter)>b$ and $\ell^->0$ and $\ms{\Aiter} \neq \ms{\Biter}$.
        \item [(2)] Exactly one party jumps to $b$. 
    \end{itemize}
\item If $t_0$ is not during an iteration (that is, \emph{not} between Lines~\ref{line:iterloop} through \ref{line:iterend}), then $t_0$ is during the Big Hash Computation phase starting at Line~\ref{line:big-hash-phase} at the end of a block. 
Further, the randomness $\Rblock^b$ exchanged during that block is corrupted during the randomness exchange in Line~\ref{line:alice-send-bigrand}. 
    
    \item Fix a time $t\in [t_0, t_1]$. For any point $p$ such that $b < p \leq \min\{\Aiter(t),\Biter(t)\}$, $\ms{p}_A \neq \ms{p}_B$.

    \item Let $t \in [t_0,t_1]$. Then $\ms{\Aiter(t)} \neq \ms{\Biter(t)}$.
    \item If the bad spell ends at time $t_1$, then $t_1$ must occur during an iteration denoted by $I_1$ and further, one of the following occurred at time $t_1$:
    \begin{enumerate}
        \item [(1)] Both parties jumped to a point $p \leq b$. 
        \item [(2)] One party jumped to $b$, while the other party was waiting at $b$.\footnote{For example, if Alice was ``waiting at $b$,'' this means that $\Aiter(I^{(1)}) = \Aiter(I^{(1)} - 1) = b$, and similarly for Bob.}
    \end{enumerate}
 \item If the bad spell does not end at time $t_1$, then one of the following occurs: 
        \begin{itemize}
            \item[(1)] $t_1$ is during the Big Hash Computation phase at the end of a block, and the randomness $\Rblock^b$ exchanged during that phase is corrupted.
           \item[(2)] $t_1$ is  during the Transition Phase of Algorithm~\ref{alg:adaptive}, and at time $t_1$ the divergent point changes to a point $b'$ with $b'<b$.
        \end{itemize} 
\end{enumerate}
\end{lemma}
\begin{proof}  We prove each item in turn. 
\begin{enumerate}
 
    \item First, observe that the Alice and Bob's megastates do not change during the verification phase, so $t_0$ is either in the computation phase or the transition phase of some iteration $I^{(0)}$.

    Suppose first that $t_0$ is in the computation phase. 
    Since $t_0$ is when $b$ first becomes the divergent point, we claim it must be that $\ell^-=0$ right before time $t_0$. 
    Indeed, suppose not, so there was some divergent point $b'$ prior to time $t_0$.  
    First suppose that $b' < b$. But then $b$ could not have become the divergent point, a contradiction.  On the other hand, suppose that $b' = b$.  But then $b$ would not have \emph{become} the divergent point at time $t_0$, as it was already the divergent point.  Finally, suppose that $b' > b$.  But then $\Aiter(t_0), \Biter(t_0) \geq b' > b$, so since $b'$ is the divergent point, this means that Alice and Bob agree about $b$: that is, just before time $t_0$ $\ms{b}_A = \ms{b}_B$ and $\simPath_A^{(\leq b)} = \simPath_B^{(\leq b)}$.  Since $t_0$ is during the computation phase, neither $\ms{b}$ nor $\simPath^{(\leq b)}$ can change at time $t_0$ for either party, so we conclude that
    \[ \ms{b}_A(t_0) = \ms{b}_B(t_0) \qquad \text{ and } \simPath_A^{(\leq b)}(t_0) = \simPath_B^{(\leq b)}(t_0).\]
    But this contradicts the assumption that $b$ became the divergent point at time $t_0$.  In any of the three cases we have a contradiction, so we conclude that $\ell^- = 0$ right before $t_0$.

    We therefore conclude that $\ell^-=0$ before the computation phase, i.e., $\Aiter=\Biter=:p$ and for all $p' \leq p$ we have $\ms{p'}_A=\ms{p'}_B$. We now claim $p=b$. For let us consider how a bad spell can begin (which necessarily happens if $b$ becomes the divergent point during the computation phase). If both parties just computed dummy rounds, then we would still have $\Aiter=\Biter$ and, as no new mega-states were created, $\ell^-$ still equals $0$. So after the computation phase at least one party increased the length of their simulated path by $1$, implying this party's simulated path now has length $p+1$. As a bad spell began during the computation phase, either the simulated paths or the mega-states at depth $p+1$ disagree. But as these data agreed at all points up to depth $p$, it must be that $p$ becomes the divergent point during the computation phase. This implies $p=b$. 
    
    Finally, we note that if after the computation phase we have $\simPath^{(\leq \Aiter)}_A \neq \simPath^{(\leq \Biter)}_B$ then Lemma~\ref{lem:bigh-hash-corruption} implies that $\ms{\Aiter}_A \neq \ms{\Biter}_B$. Else, if $\simPath^{(\leq \Aiter)}_A = \simPath^{(\leq \Biter)}_B$, since $\ell^- > 0$ after the computation phase by definition it must be that $\ms{p}_A\neq \ms{p}_B$.

    Suppose now that $t_0$ lies in the transition phase. Note that if neither party jumps, then the divergent point cannot change, contradicting the assumption that $b$ becomes the divergent point at time $t_0$. Suppose now for a contradiction that both parties jump to the same point in the transition phase. If they both jump to a point $p \leq b$, then $b$ could not be the divergent point at time $t_0$, a contradiction. If they both jump to a point $p > b$, then for $b$ to be the divergent point after the transition it must have also been the divergent point prior to the transition, again a contradiction. 

    Thus, it cannot be that both parties jump to the same location, and that at least one party jumps. Without loss of generality, assume that Alice jumps to a point that is smaller than the point that Bob jumps to if he also jumps, or that she jumps and Bob does not. If Alice jumps to a point $p<b$, then $b$ could not be the divergent point at time $t_0$, and similarly if $p>b$ then for $b$ to be the divergent point after the transition if must have also been the divergent point prior to the transition. So the only remaining possibility is that she jumps to the point $b$ (and Bob does not). 
\item Notice that the only time \emph{outside} of an iteration that a megastate is updated is during the big hash computation line~\ref{line:big-hash-comp}. As a result, if the divergent point changes at time $t_0$, then $t_0$ must be during the big hash computation. Further, notice that if the randomness $\Rblock^b$  is not corrupted then the agreement between megastates of Alice and Bob do not change after the big hash computation. More precisely, if two megastates $\megastate_A$ and $\megastate_B$ were equal before time $t_0$, then Alice and Bob use the same inputs and same randomness to compute their big hash functions.  Then $\phash$ and $\pseed$ will be updated in the same way in the Big Hash Computation phase, and $\megastate_A$ and $\megastate_B$ we remain the same after the Big Hash Computation phase. On the other hand, if $\megastate_A \neq \megastate_B$ before time $t_0$, then according to Lemma~\ref{claim:neq-in-and-out-of-block}, $\megastate_A \neq \megastate_B$ after the Big Hash Computation phase, meaning $\megastate_A\neq \megastate_B$ at time $t_0$ as well. Thus $\Rblock^b$  must be corrupted if the divergent point changes at time $t_0$. 
  
    \item For each time $t$ we prove this statement by induction on $p$. For the base case, let $p=b+1$. Let $\megastate_A = \ms{b+1}_A $ and $\megastate_B = \ms{b+1}_B$. If $\piter_A \neq \piter_B$ then we can conclude that $\megastate_A \neq \megastate_B$ which proves our claim. On the other hand, if $\piter_A = \piter_B = I$ where $I\in [I^{(0)}, I^{(1)}] $ then Alice and Bob have both simulated during  iteration $I$. Now if $I = I^{(0)}$ then from item 1 we know that $\megastate_A \neq \megastate_B$ after time $t_0$, thus the claim is complete. 
       
       Let $p$ be any point such that $b+1 <  p \leq \min (\Aiter, \Biter)$. Define $\megastate_A = \ms{p}_A$ and $\megastate_B = \ms{p}_B$. Then similar to the base case, if $\piter_A \neq \piter_B$ then we have that $\megastate_A \neq \megastate_B$ which proves our claim. If $\piter_A = \piter_B$, then both Alice and Bob have simulated during the computation phase. Then prior to the computation phase Alice and Bob are at depth $p-1$. By induction, $\ms{p-1}_A \neq \ms{p-1}_B$. Then from Lemma~\ref{lem:ms-not-matching-sim-ms-not-matching} we know that after the simulation phase $\megastate_A \neq \megastate_B$.
    
    \item We distinguish two cases. Suppose first that $\Aiter(t)\neq \Biter(t)$. Then $\simPath_A(t) \neq \simPath_B(t)$. As we are assuming that no big hash collisions occur, Lemma~\ref{lem:bigh-hash-corruption} implies that $( H_A, R_A,\Bsimp_A , \IHA )(t) \neq (H_B,R_B, \Bsimp_B,\IHB )(t)$, which indeed implies $\ms{\Aiter(t)} \neq \ms{\Biter(t)}$.
    
   On the other hand, if $\Aiter(t)=\Biter(t)$, then we can apply Item 3 with $p=\Aiter(t)=\Biter(t)$ to deduce $\ms{\Aiter(t)} \neq \ms{\Biter(t)}$. 
    
    \item First, from Item 3, we observe that for a bad spell to end it must be that the parties are both at a point $p$ with $p \leq b$. Indeed, at any point $p>b$ we have $\ms{p}_A \neq \ms{p}_B$, implying that the bad spell is still ongoing. 

    Suppose it is not the case that both parties jumped at time $t_1$; without loss of generality, suppose Alice does not jump. If Alice simulated during the computation phase of $I^{(1)}$, then she must have been at point $p-1$ prior to this round; since $p-1<p \leq b$, this contradicts the assumption that $b$ is the divergent point. Furthermore, if she waited in this iteration (i.e., she did dummy computations in the computation phase) and $p<b$, then again we contradict the fact that $b$ was the divergent point at the start of $I^{(1)}$. We conclude that she must have been waiting at point $b$. 

    As the bad spell was still ongoing at the start of iteration $I^{(1)}$, Bob must have not been at point $b$, and as $b$ is the divergent point he must have been at a point $p>b$. Since both Alice and Bob must be at the same point at time $t_1$, it follows that Bob must have jumped to the point $b$ at time $t_1$.
    
    \item If the bad spell does not end at time $t_1$, but $b$ is no longer the divergent point after time $t_1$, then it must be that the divergent point changes at time $t_1$; let $b' \neq b$ be the new divergent point. If $b'>b$, then we know from Item 3 that $\ms{b'}_A \neq \ms{b'}_B$, implying that $b'$ cannot be a divergent point. Hence, it must be that $b'<b$. Note that as $b$ was a divergent point, by definition it must be that 
    $$\simPath_A^{(\leq b')} = \simPath_B^{(\leq b')} \text{ and } \forall p' \leq b',~\ms{p'}_A = \ms{p'}_B \ .$$

    First assume that $t_1$ is during an iteration, and denote this iteration by $I_1$.  Then
    as $b'$ is now the divergent point, we either have that (a) at least one party's simulated path is of length most $b'$, or that(b)  both parties have simulated paths of length at least $b''>b'$, but the simulated paths disagree at depths $b'+1,\dots,b''$.   As $b$ was the divergent point before time $t_1$, it must be that the simulated paths had agreed at depths $b'+1,\dots,b$, showing that case (b) case is in fact impossible, and (a) holds: at least one party's simulated path has length at most $b'$.  Suppose without loss of generality that this party was Alice; then Alice's simulated path had length \emph{exactly} $b'$ at time $t_1$ (or else the divergent point would be higher).  Since $\Aiter \geq b > b'$ before time $t_1$, it must be that she jumped back to the point $b'$ at time $t_1$. We also note that both Alice and Bob could not have both jumped back to $b'$, as otherwise---appealing to the proof of the previous part---the bad spell would have ended at time $t_1$. Thus, one party jumped back to $b'$, which must occur during the transition phase of Algorithm~\ref{alg:adaptive}. It follows that $t_1$ is during the transition phase of Algorithm~\ref{alg:adaptive} and that the new divergent point is $b'$ with $b'<b$, as claimed. 
    
    Next, assume that $t_1$ is not during an iteration.  By Item 2, $t_1$ must be during the Big Hash Computation phase at the end of the block.  As no megastates are created during the Big Hash Computation phase, there must be two megastates $\megastate_A$ and $\megastate_B$ such that $\megastate_A = \megastate_B$ prior to Big Hash Computation phase and $\megastate_A \neq \megastate_B$ after the Big Hash Computation phase. If $\Rblock^b$ is not corrupted, then with a similar argument to Item 2, we have that the agreement/disagreement status between megastates of Alice and Bob do not change during the big hash computation. Hence, $\Rblock^b$ must be corrupted. 
\end{enumerate}

\end{proof}

We begin by giving an upper bound for the number of iterations that either have corrupted randomness or corrupted communication. In particular, we show that the they are both on the order of $\eps d$, which is the adversary's corruption budget. 
\begin{lemma} \label{lem:bound-corrupted-iterations} Define $Q$ as  the total number of iterations with either corrupted communication or corrupted randomness. Then, $ Q = O( \eps d)$.
\end{lemma}
\begin{proof}
    We count the number of iterations suffering from each type of corruption separately. 
    For an iteration to have a corrupted randomness either the randomness shared at the start of the block is corrupted or the randomness shared at the start of the iteration is corrupted. The total number of iterations having a corrupted  randomness on iteration level is at most $2 \eps d $ as this is the maximum number of iterations in which the adversary can introduce corruptions. To count the number of iterations with corrupted randomness shared at block level, notice that according to Theorem~\ref{thm:goodECC} the randomness exchange is protected by an ECC with minimum distance $4 \Iblock $. Then the adversary must invest in at least $2 \Iblock$ corruptions to corrupt this randomness. As a result there are at most $ \frac{2  \eps d }{2 \Iblock }$ blocks having corrupted randomness which means that there are at most $ \frac{ \eps d }{\Iblock} \Iblock = \eps d $ iterations having corrupted randomness.
    Finally, similar to randomness exchanged at iteration level, there are at most $ 2 \eps d $ iterations having corrupted communication. Adding all cases together, there are at most $ O( \eps d )$ iterations having either corrupted randomness or corrupted communication, i.e., $ Q = O ( \eps d )$ as claimed.
\end{proof}

Finally, we prove a few useful lemma about meeting points.

\begin{lemma}[jumps go above stable points]\label{lem:MP-jump}

Let $p$ be a $w$-stable point, and suppose that $p$ is removed from Alice's meeting point set $\M_A$ at time $t$.  
Suppose that $\tjmp$ is the next time after $t$ that Alice jumps to a point $q < p + 2^w$.  Then $q \leq p$.
\end{lemma}

\begin{proof}
    Suppose towards a contradiction that $q > p$.  
    Suppose that $q$ is $u$-stable, for some $u < w$; note that we must indeed have $u < w$, because by assumption $q \in (p, p + 2^w)$ and thus lies strictly between two consecutive multiples of $2^w$.
    
    By Lemma~\ref{lem:forget-avmps-z}, the point $p$ will be removed from Alice's memory when she first reaches a depth of $c_p = p + 2^{w+1}$.
    
    Now we claim that $c_p < q + 2^{u+1}$.  To see this, we first observe that Lemma~\ref{lem:inMa} implies that $q \in M_{c_p}$ if and only if $c_p \in [q, q + 2^{u+1} - 1],$ so as long as $q \in M_{c_p}$, we will have $c_p < q + 2^{u+1}$. 
    
    Next we show that $q \in M_{c_p}$.  
    As $q \in \M_A(\tjmp)$ (since Alice jumped to $q$ at that time), if $q \not\in M_{c_p} = M_{\Aiter(t)}$, then Alice had ``forgotten'' $q$ by time $t$, so we must have $\Aiter(t') = q$ for some $t' \in (t, \tjmp)$, so that $q$ can have been added back before Alice jumped there.  But the only way to add $q$ back would have involved jumping to a point $q' \leq q$ at some time $t'' \in (t, t'] \subseteq (t, \tjmp)$, contradicting our assumption that $\tjmp$ was the next time after $t$ that Alice jumped to  a point above $p + 2^w$.
    
     Therefore, $q \in M_{c_p}$, and we conclude by the reasoning above that $c_p < q + 2^{u+1}$.  Since $u < w$ by the above, we see that $c_p < q + 2^w.$  

    Thus, recalling the definition of $c_p$, we have
    \[p + 2^{w+1}  \leq c_p < q + 2^w,\]
    implying that $q > p + 2^w.$  However, this is false, since we have assumed that $q < p+2^w$.
    
\end{proof}

Suppose that Alice and Bob are voting on a shared meeting point, so the parameters $k_A$ and $k_B$ are increasing.  When $k_A = k_B \in \{2^j, \ldots, 2^{j+1} - 1\}$, they are ``voting'' for a shared scale-$j$ transition candidate, if it exists.
Intuitively, if Alice and Bob ``skip'' such a transition candidate (that is, they do not jump when $k_A, k_B$ are in this window), then the adversary had to introduce on the order of $2^j$ corruptions in order to make them do that. We formalize this intuition in the following lemma. 

\begin{lemma}[The adversary must pay to ``skip'' a meeting point] \label{lem:skipsCostAdv} 

    Fix a time $\tsp$ in iteration $\Isp = \I(\tsp)$, and suppose that  $k_A(\tsp) = k_B(\tsp) =: k$.  
    Suppose that there is some $j < \lfloor \log k \rfloor$ and some point $q$ so that 
    \begin{equation} \label{eq:needfromq}
    q \in \MPalljl(j, \Aiter(\tsp)) \cap \MPalljl(j, \Biter(\tsp)),
    \end{equation}
    and that $q$ is an available jumpable point for Alice and Bob at time $\tsp$. 
     Then there are at least $(0.6) 2^{j} $  iterations  $I$  in the window $ W = \{\Isp -k +1, \Isp - k + 2, \ldots, \Isp\}$ such that during iteration $I$, $\BVC_{AB}$ increases by at least one. In particular $\BVC_{AB}( \tsp ) \geq (0.6) 2^{j}$.

\end{lemma}

\begin{proof}
    First, notice that, because $k_A(\tsp) = k_B(\tsp)$, $k_A(I) = k_B(I)$ for each $I \in W$; indeed, in each iteration, $k_A$ and $k_B$ either increase by one or get reset to zero, so the only way $k_A(\tsp) = k_B(\tsp) = k$ is if each have increased by one in each iteration for the past $k$ iterations. Further, this logic implies that both Alice and Bob ran dummy rounds (that is, Line \ref{line:dummy-rounds} in Algorithm~\ref{alg:adaptive} was executed) during the window $W$, because otherwise $k_A$ or $k_B$ would have been reset to zero.  In particular none of the quantities $\M_A, \M_B, \ell_A,$ or $\ell_B$ change during the window $W$. Further, $\BVC_{A}$ and $\BVC_B$ do not decrease during the window $W$, as by Definition~\ref{def:BVC}, $\BVC$ only decreases when either Alice or Bob reset their status, meaning that either $k_A$ or $k_B$ gets reset to zero, which by the above does not happen during $W$.

    This discussion, along with \eqref{eq:needfromq}, implies that
    \begin{equation}\label{eq:wtsq}
    \forall I \in W \text{ s.t. } j_A(I) = j_B(I) = j,  \qquad q \in \MPall_A(I) \cap \MPall_B(I).
    \end{equation}
    That is, for all iterations $I$ in the last $k$ iterations before $\Isp$ such that $j_A = j_B = j$ during iteration $I$, $q$ is a transition candidate for both Alice and Bob.

    Now consider the iterations $I^\ast$ so that $k_A(I^*), k_B(I^*) \in [2^j, 2^{j+1} -1 ]$ (observing that all of these iterations lie in $W$, due to our assumption that $j < \lfloor \log k \rfloor$).
    
    Note that for these iterations, we have $j_A(I^\ast) = j_B(I^\ast) = j$.  We claim that $\BVC_{AB}$ increases by at least one for at least $(0.6) 2^j$ of these iterations.  This will establish the lemma,
    using the fact that $\BVC_{AB}$ cannot decrease in the window $W$. To establish the claim, fix such an $I^\ast$ and let $i\in \{1,2,3\}$ be such that $q = \MPi_A$ during $I^*$.  Then since, by \eqref{eq:wtsq},  $q\in \MPall_B(I^*)$, the only way that $\vi{i}_A$ does not increment in line \ref{line:updatevi} is if $\BVC_A$ increments (see the Definition \ref{def:BVC} of $\BVC$, and recall that since $q$ is an available jumpable point, this means that not only is $q$ available for both Alice and Bob, but also the mega-state it corresponds to is the same for both of them).
    If $\vi{i}_A > (0.4) 2^j$ at iteration $2^{j+1}-1$, Alice would have jumped to $q$ in line \ref{line:back-jump} and she did not, so we conclude $\vi{i}_A$ did not increment, and hence $\BVC_A$ did, for at least $(0.6) 2^j$ iterations.  Using the definition that $\BVC_{AB} = \BVC_A + \BVC_B$, this proves the claim and the lemma. 
\end{proof}

%% file: sneaky.tex
In this section, we formally define a sneaky attack, and state and prove several useful lemmas about sneaky attacks.  
As discussed in Section~\ref{sec:tech}, a sneaky attack is the only exception to the intuition that whenever $k_A = k_B$ gets large (relative to $L^-$), then the adversary has to ``pay'' by increasing $\BVC$.

We also state our main technical lemma, Lemma~\ref{lem:maintech}, which roughly says that \emph{almost} all of the time that $k_A = k_B$ gets large (relative to $L^-$), the adversary had to pay by increasing $\BVC_{AB}$.  

We begin with an auxiliary definition, and then formally define a sneaky attack.

\begin{definition}[The constant $\csneak$]\label{def:csneak}
Fix a constant $\csneak \in \mathbb{Z}$, so that $\csneak \geq 3$.  This constant will appear in the definition of a \protosneaky below, and we will choose $\csneak$ to be sufficiently large later on in the analysis.
\end{definition}

\begin{definition}[\Protosneaky]\label{def:sneaky}
Fix a time $\tsp$ and suppose that $k_A(\tsp) = k_B(\tsp) =: k$.
Suppose that  
$\ell^-(\tsp) > 0$.
Let $j_p$ be the next \specialscale at time $\tsp$, and 
suppose that $p$ is a next available \specialpt.
Let $\Aiter = \Aiter(\tsp)$ and $\Biter = \Biter(\tsp)$.
Let $w$ be such that  $p$ is a scale-$w$ MP for Alice.   

Let $\hat{p} = p + 2^w$, $q = p + 2^{w-1}$, and  $c_q = q + 2^{w} $. 
Suppose that $b$ is the divergent point at time $\tsp$.
We say that a \defn{\protosneaky for Alice heading towards $p$ is in progress at time $\tsp$} 
if $k \leq 2^{w+1}$, and 
if 
there are times $t_{\hat{p}} < t_{c_q} < t_{\hat{q}} < t_b < \tsp$ so that all of the following hold: 
     \begin{itemize}

            \item $\Aiter(t_{c_q}) = c_q$.  Further, $t_{c_q}$ is the last time that this happens before $\tsp$.  That is,
            \[ t_{c_q} = \max\{t \leq \tsp \,:\, \Aiter(t) = c_q\}.\]

             \item $t_{\hat{p}} = \max \{t < t_{c_q} \,:\, \Aiter(t) = \hat{p}\}$ is the last time that Alice passes $\hat{p}$ before $t_{c_q}$.
            \item At time $t_{\hat{q}}$, at least one of Alice and Bob jumps to a meeting point $\hat{q} \geq \hat{p}$ 
            to end a bad spell; in particular, $\ell^-(t_{\hat{q}}) = 0$, $\Aiter(t_{\hat{q}}) =\Biter(t_{\hat{q}}) = \hat{q}$.\footnote{As per Lemma~\ref{lem:noweirdcoincidences}, either \emph{both} Alice and Bob jumped to $\hat{q}$ simultaneously; or only one of them did, but the other was already there and simulated dummy rounds.}  Moreover, $t_{\hat{q}}$ is the first time that $\ell^-$ is reset to zero after time $t_{c_q}$, and 
            $\hat{q} = \hat{p}$.%
            \footnote{The reason we give $\hat{q}$ a different name (even though it is equal to $\hat{p}$) is because the time $t_{\hat{p}}$ is already defined, so we use $t_{\hat{q}}$ to denote the time of the jump, and use the letter $\hat{q}$ for the point itself to not cause confusion.  We will use the fact that $\hat{q} \geq \hat{p}$ often, but the fact that $\hat{q} = \hat{p}$ only occasionally.}

             \item Bob jumps to $b$ at time $t_b$, and further $b\geq \hat{p} - 2 ^{w-c^\ast}$. The divergent point becomes $b$ at time $t_b$, and $b$ remains the divergent point for all times in $[t_b, \tsp]$. 
          
            \item For all $t' \in [t_{\hat{p}}, t_b)$, $\Biter(t') < \hat{p} + 2^{w-\csneak}$; and $\Biter(\tsp) < \hat{p} + 2^{w-2}$.

            \item At time $\tsp$, $q$ is not in Alice's memory.  That is, $q \not\in \M_A$ at time $\tsp$.   
    
        \end{itemize}

A \defn{\protosneaky for Bob} is defined analogously (switching the roles of Alice and Bob in the above); note that in this case $p$ would be a scale-$w$ MP for Bob, rather than for Alice.

\end{definition}

We refer the reader to Figure~\ref{fig:sneaky} for a picture of a \protosneakyf. 

\begin{observation}
We make the following observations about the definition of a sneaky attack.
\begin{enumerate}
    \item[(a)] \label{obs:cqforget} Observe that $t_{c_q}$ is the time at which Alice will ``forget'' $q$.   Indeed, 
    $q$ is $(w-1)$-stable, as it lies between two multiples of $2^w$, and Lemma~\ref{lem:forget-avmps-z} implies that Alice forgets it at depth $q +2^w = c_q$.

    \item[(b)] \label{obs:Alice-Depth-Sneaky-Attack} Observe that for any time $t \in [t_{\hat{p}}, \tsp]$, we have $\Aiter(t) \geq \hat{p}$.  Indeed, by the definition of $t_{\hat{p}}$, this holds for any $t \in [t_{\hat{p}}, t_{c_q}]$; otherwise since $c_q > \hat{p}$, Alice would have had to cross $\hat{p}$ at some time in $[t, t_{c_q}]$, contradicting the choice of $t_{\hat{p}}$ as the last time before $t_{c_q}$ that this occured. 
    
    Now suppose that for some $t \in [t_{c_q}, \tsp]$, we have $\Aiter(t) < \hat{p}$; say $t$ is the first such time.  By Lemma~\ref{lem:inMa}, none of the points between $q$ and $\hat{p} - 1$ are available for Alice after she has passed $c_q$, as $c_q \geq q + 2^w$.
    Since Alice only jumps to available meeting points, 
    this means that the first time after $t_{c_q}$ that Alice jumps above $\hat{p}$, she in fact must jump above $q$.  This implies that $\Aiter(t) < q$.  
    However, at time $\tsp > t$, $p$ is an order-$w$ MP for Alice, which means that $\Aiter(\tsp) > p + 2^w = \hat{p}$.  Thus, there was some time $t^* \in (t, \tsp)$ so that $\Aiter(t^*) = q$.  But then $q$ will be re-added to $\M_A$ at time $t^*$.  Since $t_{c_q} < t^*$ was the last time before $\tsp$ that Alice reached $c_q$, part (a) of the Observation implies that $q \in \M_A$ at time $\tsp$ as well.  But this contradicts the last part of Definition~\ref{def:sneaky}.
   
\end{enumerate}
   
\end{observation}

We next define a few auxiliary notions surrounding \protosneakyf{s}.  The first has to do with the ``end'' of a \protosneakyf, when Alice and Bob are voting about whether to jump to the point $p$.

\begin{definition}[Voting window; \sneakyjmpf]\label{def:votewindow}

If a \protosneaky is in progress at time $\tsp$, then by definition $k_A(\tsp) = k_B(\tsp) = k$, which implies that for iterations in the window $W = \{\Isp - k + 1, \ldots, \Isp\}$, Alice and Bob have been simulating dummy rounds and the parameters $\ell_A, \ell_B$ have not changed.\footnote{Indeed, these dummy rounds are simulated in Line~\ref{line:dummy-rounds} of Algorithm~\ref{alg:adaptive}; notice that the dummy rounds are simulated only if $k > 1$.  Further, if Line~\ref{line:dummy-rounds} is not executed, then $k$ gets reset to $0$ and hence will be $1$ at the beginning of the next iteration, so running dummy rounds is the only way that $k$ can increment over several iterations.} We call this window $W$ the \defn{voting window for the \protosneakyf.}

Suppose that, at some time $\tjmp \geq \tsp$, Alice and Bob have continued voting (that is, neither $k_A$ nor $k_B$ have been reset to $0$ between $\tsp$ and $\tjmp$), and then they successfully jump to $p$ to resolve the bad spell at time $\tjmp$.  Then we say that the \protosneaky is \defn{completed}, and we call the jump at time $\tjmp$ a \defn{\sneakyjmpf}.

We refer to the time interval $[t_{\hat{p}}, \tjmp]$ as the \defn{time window for the \protosneakyf}.

\end{definition}

The next definition has to do with the part of the \protosneaky where Alice is getting ``driven down'' (Step 2 in Figure~\ref{fig:sneaky}).

\begin{definition}\label{def:divewindow}
    Suppose that a \protosneaky towards $p$ is in progress at time $\tsp$ (say, for Alice), and let $w$ be the scale of $p$ for Alice.  Let $t_{\hat{p}}, t_{c_q}$ be as in Definition~\ref{def:sneaky}.  Define the \defn{diving window for the sneaky attack} to be the set $\divewindow$ of iterations $I$ between time $t_{\hat{p}}$ and $t_{c_q}$ where:
    \begin{itemize}
        \item There is some $\ell \in [\hat{p}, c_q]$ so that Alice simulates $\Pi$ at depth $\ell$ for the last time in $[t_{\hat{p}}, t_{c_q}]$; and 
        \item $\ell^- > 0$.
    \end{itemize}
\end{definition} 

That is, intuitively we would like to define $\divewindow$ to be the set of iterations where Alice is getting ``driven down'' below Bob.  However, Alice may jump during this phase, and for technical reasons we want her to be at a different point for every iteration in the diving window.  Thus, we just count the iterations where she is at a particular point for the \emph{last} time during this window.

\begin{remark} \label{rem:div-size}
    If $S$ is a sneaky attack with a jump of scale $w$, then $| \divewindow| = 2 ^{w-1}$. The reason for this is that the diving window contains an iteration for each depth between $[\hat{p},c_q]$. As a result we have that, $ | \divewindow | = c_q - \hat{p} = ( q + 2 ^w ) - ( p + 2 ^w ) = q - p = 2^{w-1}$.
\end{remark}

\begin{lemma}\label{lem:divewindow-progression}  

Suppose that  $\divewindow$ is a diving window, and that
\[ \divewindow = \{I_1, I_2, \ldots, I_T \}\]
for some $I_1 < I_2 < \cdots < I_T$.  Then 
for any $i = 1, \ldots, T-1$, $\Aiter(I_i) < \Aiter(I_{i+1})$, and there is no $I \in (I_i, I_{i+1})$ so that $\Aiter(I) \leq \Aiter(I_i)$.   
\end{lemma}
\begin{proof}

We first prove that $\Aiter(I_i) < \Aiter(I_{i+1})$ by induction on $i$, starting with $i = T-1$.  Suppose that $i = T-1$.  Then $\Aiter(I_{i+1}) = \Aiter(I_T) = c_q$, since by definition $I_T$ is the iteration in which the diving window ends, namely the iteration where Alice reaches $c_q$.  This implies that $\Aiter(I_{T-1}) < \Aiter(I_T)$, since $\Aiter(I_j) \leq c_q$ for all $I_j \in \divewindow$, and $\Aiter(I_j)$ is distinct for each $I_j \in \divewindow$.  
Now fix some $i^* < T-1$ and assume that for all $i \geq i^*$, $\Aiter(I_{i}) < \Aiter(I_{i+1})$. Consider $i = i^* - 1$.  We claim that $\Aiter(I_{i}) <  \Aiter(I_{i+1}) = \Aiter(I_{i^*})$.  Indeed, suppose that $\Aiter(I_{i}) >  \Aiter(I_{i^*})$.  But then there is some $j \geq i^*$ so that $\Aiter(I_i) \in [\Aiter(I_j), \Aiter(I_{j+1})]$, where we have used the inductive hypothesis for $j$. In particular there is some $I \in [I_j, I_{j+1}]$ so that $\Aiter(I) = \Aiter(I_i)$.  But this contradicts the choice of $I_i$ as the \emph{last} iteration when Alice was at $\Aiter(I_i)$.
Thus (using the fact that $\Aiter(I_i) \neq \Aiter(I_{i+1})$, as all the values $\Aiter(I_i)$ are distinct), we conclude that $\Aiter(I_{i})  <\Aiter(I_{i+1})$.
This establishes the inductive hypothesis and the claim.

Now we prove the second statement.  Fix $i$ and
 let $I \in (I_{i}, I_{i+1})$.  If $\Aiter(I) \leq \Aiter(I_{i})$, then there must be some $I' \in [I, I_{i+1}] \subseteq (I_{i}, I_{i+1}]$ so that $\Aiter(I') = \Aiter(I_{i})$, since Alice must cross $\Aiter(I_{i})$ again on her way to $\Aiter(I_{i+1})$.  But this contradicts the choice of $I_{i}$ as the last iteration in which Alice reaches $\Aiter(I_{i})$.  
\end{proof}

\begin{lemma}[There are many corruptions or small hash collisions during a \protosneakyf] \label{lem:divwindow-corr-coll} Suppose there are no big hash collisions throughout Algorithm~\ref{alg:adaptive}.
    Let $\mathcal{S}$ be a sneaky attack in progress for Alice towards a point $p$.  
    Then in at least $  (1-2^{1- \csneak})|\divewindow|$  iterations in $\divewindow$, at least one of the following occurred:
    \begin{itemize}
        \item There was a dangerous small hash collision. 
        \item The adversary inserted a corruption.
        \item There was corrupted randomness for that iteration (Definition~\ref{def:corrupted_rand}). 
    \end{itemize}
    Above, $\csneak$ is as in Definition~\ref{def:csneak}.
\end{lemma}
\begin{proof} 
  
Notice for any time $t \in  [t_{\hat{p}}, t_{c_q})$,  we have from Definition~\ref{def:sneaky} that $\Biter(t) < \hat{p} + 2^{w - \csneak}$.  (Indeed, Definition~\ref{def:sneaky}  says that this holds until $t_b > t_{c_q}$.)
 
Thus for any iteration $I\in \divewindow$ such that $\ell_A(I) \geq \hat{p} + 2^{w-\csneak} $ we have $\ell_A(I) \neq \ell_B(I)$.

Next, let $I^\ast\in \divewindow$ be the iteration where $\ell_A (I^\ast) = \hat{p} + 2 ^{w-\csneak} + 1$. Note that such an iteration must exist because $c_q \geq \hat{p} + 2^{w-1} > \hat{p} + 2^{w -\csneak}$, and there is an iteration in $\divewindow$ for each point between $\hat{p}$ and $c_q$.
Lemma~\ref{lem:divewindow-progression} implies that for any iteration  $I \in \divewindow$ after (and including) $I^\ast$ and before time $t_{c_q}$, we have $\Aiter(I) \geq \Aiter(I^*) > \hat{p} + 2^{w - \csneak}$, and in particular $\ell_A(I) \neq \ell_B(I)$ for all such $I$.

Fix an iteration $I \in \divewindow$ after $I^*$ and before time $t_{c_q}$, By Lemma \ref{lem:bigh-hash-corruption}, the tuple $$(\phash, \pseed,  \pt  , \piter)$$ that Alice and Bob have computed at the beginning of iteration $I$ will not match. Also note that, since that $\ell_A (I) \neq \ell_B(I)$, iteration $I$ is a dangerous iteration.  

In order for the simulation for Alice to proceed during the computation phase of iteration $I$, the check of $(\phash, \pseed,  \pt , \piter)$ on Line~\ref{line:checksmallhash} must pass, which means that the adversary must have corrupted Bob's transmitted small hash value $H_{\pseed, \phash, \piter}$; there was corrupted randomness in iteration $I$; or there was a dangerous small collision.

Finally we count the number of such iterations $I \in \divewindow$ after $I^\ast$. There are at most $2^{w-\csneak}$ iterations in the $\divewindow$ that can occur before $I^\ast$ (using Lemma~\ref{lem:divewindow-progression} to assert that the depths of the iterations in $\divewindow$ are increasing). This leaves at least
\[ |\divewindow| - 2^{w-\csneak} \]
iterations remaining. Since by Remark~\ref{rem:div-size} we have $|\divewindow| = 2^{w-1}$, then there are at least $(1-2^{1-\csneak}) |\divewindow|$ iterations in the diving window containing either corruption or dangerous small hash collision. This concludes the proof.

\end{proof}

\begin{lemma} [\Protosneaky{s} are not too tall]  \label{lem:one-sneaky-counter-bound-coll-corr}
Let $I$ be an iteration during the voting window of a sneaky attack $\mathcal{S}$ for Alice towards a point $p$. 
Let $H_{\mathcal{S}}$ be the number of dangerous small hash collisions that occur during the diving window of $\mathcal{S}$, and similarly $Q_{\mathcal{S}}$ be the number of iterations having corrupted randomness or corrupted communication during the diving window of $ \mathcal{S}$.
Then $ k_A(I) = k_B(I) = k $ where $ k \leq \alpha (Q_{\mathcal{S}}+ H_{\mathcal{S}}) $ for some absolute constant $\alpha$.

\label{lem:sneaky-jump-corr-coll} 
    
\end{lemma}
\begin{proof}
\newcommand{\HS}{H_{\mathcal{S}}}
 To prove the statement, we first show that the maximum number of iterations in a diving window has an upper bound of $c^\prime (  Q_\mathcal{S} + \HS) $. Then, by relating the size of the diving window to the size of the voting window we can show the desired bound on $k$.
     From Lemma \ref{lem:divwindow-corr-coll}, we know that
     in at least $ ( 1 - 2^{1- \csneak})|\divewindow|$ iterations in the diving window, there is either a dangerous small hash collision, corrupted communication or corrupted randomness. 
     This implies that
     \[ ( 1 - 2^{1- \csneak})|\divewindow| \leq  Q_{\mathcal{S}} + \HS. \] 
     Then setting $ c^\prime = \frac{1}{(1-2^{1-\csneak})}$ we immediately get that $ |\divewindow| \leq c^\prime ( Q_{\mathcal{S}} + H_S)$.

Let $w$ be the scale of the meeting point $p$ for Alice, so Remark~\ref{rem:div-size} implies that $|\divewindow|= 2^{w-1}$. 
    Thus, the above implies that
    \[ 2^{w-1} \leq c' (  q_{\mathcal{S}} + \HS ).\]
    Further, we have $k \leq 2^{w+1}$ from the definition of a sneaky attack. Thus, we have
    \[ k \leq 16c' ( Q_{\mathcal{S
    }} + \HS). \]
    Setting $\alpha = 16c'$ proves the statement. 
\end{proof}

\begin{lemma}[Diving windows and voting windows are disjoint]\label{lem:disjoint} 
    Let $\mathcal{S}_1$ and $\mathcal{S}_2$ be two distinct \protosneakyf{s} that resolve with two \sneakyjmpf{s}.
    Then the voting windows and the diving windows for $\mathcal{S}_1$ and $\mathcal{S}_2$ are all pairwise disjoint. 
\end{lemma}

\begin{proof}
\newcommand{\sone}{\mathcal{S}_1}
\newcommand{\stwo}{\mathcal{S}_2}
    Each sneaky attack has a series of points and times as per Definition~\ref{def:sneaky}.  We subscript these with $i \in \{1,2\}$ to indicate which sneaky attack we are talking about.  That is, for $i \in \{1,2\}$, sneaky attack $\cS_i$ is towards the meeting point $p_i$, which is at scale $w_i$ for whichever party the attack is for, and involves points $\hat{p}_i, \hat{q}_i, q_i,$ and $c_{q_i}$; and times $t_{\hat{p}_i} < t_{c_{q_i}} < t_{\hat{q}_i} < t_{b_i} < t_{\mathrm{jump}_i}$.

    \begin{claim}
        The voting windows (c.f. Definition~\ref{def:votewindow}) of $\sone$ and $\stwo$ are disjoint.  
    \end{claim}
        \begin{proof}

            We have two possible scenarios. Either  $t_{\mathrm{jump_1}} \neq t_{\mathrm{jump_2}}$ or $t_{\mathrm{jump_1}} = t_{\mathrm{jump_2}} =: \tjmp$.

            For the first scenario, without loss of generality assume that $t_{\mathrm{jump}_1} < t_{\mathrm{jump}_2}$. 
            At time $t_{\mathrm{jump}_1}$, the counters $k_A, k_B$ are reset to $0$ as part of the sneaky jump.  This implies that the voting window for $\stwo$ begins after time $t_{\mathrm{jump}_1}$, since during the voting window the counters $k_A,k_B$ are increasing. 
Hence the voting windows are disjoint.

            Now we consider the second scenario, where $t_{\mathrm{jump_1}} = t_{\mathrm{jump_2}} = \tjmp$.  This scenario is more involved, and there are two cases: either both sneaky attacks are for the same party, or they are for different parties.
            \begin{enumerate}
                \item \textbf{Both sneaky attacks are for the same party.} Without loss of generality assume both $\sone $ and $\stwo$ are for Alice. As both $\sone$ and $\stwo$ complete at the same time, they must complete at the same point $p_1 = p_2 =: p$; and since they are both for Alice, the scale of $p$ for Alice in both attacks is the same.  Thus, $w_1 = w_2 =: w$.  Now knowing the jump destination $p$ and the scale of the jump  $w$ uniquely defines the sneaky attack.  Indeed, given $p$ and $w$, the points $\hat{p}, q, \hat{q}$ and $c_q$ are determined; $b$ is determined since it is the divergent point at the time of the jump. 
                Further, the times involved in the attack are also determined: We set $\tsp = \tjmp$; the time $t_b$ is determined as it is the time that $b$ became the divergent point; the time $t_{c_q}$ is determined as it is the last time that Alice reached $c_q$ before $\tsp$; the time $t_{\hat{p}}$ is determined as the last time before $t_{c_q}$ that Alice passed $\hat{p}$; and the time $t_{\hat{q}}$ is the first time after $t_{c_q}$ that $\ell^- = 0$.
                
We conclude that $\sone = \stwo$, a contradiction.  Thus this case cannot occur.
            \item \textbf{The two sneaky attacks are for different parties.}  As before, since both attacks resolve with the same jump, we have $p_1 = p_2 =: p.$ Without loss of generality, assume that $\sone$ is for Alice, so $p$ is a scale $w_1$ meeting point for Alice. Similarly, $\stwo$ is for Bob, and $p$ is a scale $w_2$ meeting point for Bob. Furthermore, suppose without loss of generality that $w_1\geq w_2$. Notice that when $w_1\geq w_2$, then $\hat{p}_1 \geq \hat{p}_2$. 
            Indeed, this is true because 
            \[ \hat{p}_1 = p + 2^{w_1}  \geq p + 2^{w_2} = \hat{p_2}. \]

During the attack $\stwo$, Alice jumps to a point $b_2$ at time $t_{b_2}$, where $b_2 < \hat{p}_2 \leq \hat{p}_1$. 
By Observation~\ref{obs:Alice-Depth-Sneaky-Attack}(b), we have $\Aiter(t) \geq \hat{p}_1$ for all times $t \in [t_{\hat{p}_1}, \tjmp]$.  Thus, $t_{b_2}$ cannot lie in this interval, so $t_{b_2} < t_{\hat{p}_1}$.  However, from the definition of a \protosneaky for $\stwo$, this means that $b_2$ is the divergent point for the entire interval $[t_{b_2}, \tjmp)$.  This is a contradiction because from the definition of a \protosneaky for $\sone$, there is a point $t_{\hat{q}_1}  \in (t_{\hat{p}_1}, \tjmp) \subseteq [t_{b_2}, \tjmp)$ where a bad spell ends, and there is \emph{no} divergent point.

            \end{enumerate}   
            This proves the claim.
        \end{proof}
   
    \begin{claim}
        The voting window $\sone$ is disjoint from the diving window of $\stwo$; and vice versa.
        \end{claim}
        \begin{proof}
        Without loss of generality we prove that the voting window of $\sone$ is disjoint from the diving window of $\stwo$.
            During iterations included in the voting window of $\sone$,
            Alice and Bob are both stationary in the robust protocol $\Pi'$ and do not simulate the original protocol $\Pi$. This is in contrast to iterations in the diving window of $\stwo$, which correspond to iterations in which at least one of the parties is simulating the original protocol $\Pi$. Hence voting windows and diving windows cannot overlap, proving the claim.
        \end{proof}
  
    \begin{claim}
        The diving windows of $\mathcal{S}_1$ and $\mathcal{S}_2$ are disjoint from each other.
        \label{claim:div-disjoint}
        \end{claim}

    \begin{proof} 
   
            If the time windows of $\mathcal{S}_1$ and $\mathcal{S}_2$ are disjoint from each other (that is, if $t_{\mathrm{jump}_1} \leq t_{\hat{p}_2}$ or $t_{\mathrm{jump}_2} \leq t_{\hat{p}_1}$), then the diving windows are also disjoint. Thus, suppose that $\sone$ and $\stwo$ overlap in time. Either the time window of one attack is contained in the other, or one attack begins during  the other attack and continues for a longer duration. Each case is treated separately.
  
            \textbf{Case 1:} \textit{The time window of one attack starts during the other attack and continues after the other is completed.}
            
                Without loss of generality assume that $\tphatone \leq \tphattwo <t_{\mathrm{jump}_1} < t_{\mathrm{jump}_2}$ and that $\mathcal{S}_1$ was directed towards Alice.
                
                Recall from Observation~\ref{obs:Alice-Depth-Sneaky-Attack} that for any time $t \in [\tphatone, t_{\mathrm{jump}_1}]$, $\Aiter(t) \geq \hat{p}_1$, and 
                from Definition \ref{def:sneaky} we have that for any such $t$, $\Biter(t) \geq \hat{p}_1 - 2^{w_1-\csneak}$.
                In particular, $\Aiter(t_{\hat{p}_2}), \Biter(t_{\hat{p}_2}) \geq \hat{p}_1 - 2^{w_1-\csneak}$.  
                 Since at least one of Alice and Bob is at $\hat{p}_2$ at time $t_{\hat{p}_2}$ (depending on the target of $\stwo$), we conclude that
                \[ \hat{p}_2 \geq \hat{p}_1 - 2^{w_1 - \csneak}.\]

           Notice that at time $t_{\mathrm{jump}_1}$, both Alice and Bob simultaneously jump to point $p_1$, where $p_1 < \hat{p}_1 - 2^{w_1-\csneak} \leq \hat{p}_2$.
           We claim that this leads to a contradiction.  Indeed, for whichever party $\stwo$ was against, that party must remain \emph{below} $\hat{p}_2$ during the interval $[t_{\hat{p}_2}, t_{\mathrm{jump}_2}]$, by Observation~\ref{obs:Alice-Depth-Sneaky-Attack}(b).  But $t_{\mathrm{jump}_1}$ lies in that interval, and we have just seen that both parties jump \emph{above} $\hat{p}_2$ at that time.
           Thus this case cannot occur.

\textbf{Case 2:} \textit{The time window of one attack is completely contained in the other attack.} Again without loss of generality assume that $t_{\hat{p}_1} \leq t_{\hat{p}_2 } < t_{\mathrm{jump}_2} < t_{\mathrm{jump}_1}$, so the time window for $\stwo$ is contained in that of $\sone$. Furthermore assume without loss of generality that the attack $\sone$ is for Alice.  Then for all $t \in [t_{\hat{p}_1}, t_{\mathrm{jump}_1}]$, we have that $\ell_A(t) \geq \hat{p}_1$ and  $\ell_B(t) \geq \hat{p}_1 - 2^{w_1 - \csneak}$. 
 (Indeed, the first is from Observation~\ref{obs:Alice-Depth-Sneaky-Attack}(b); the second follows because if Bob went to a point less than $\hat{p}_1 - 2^{w_1 - \csneak}$ during that interval, he would have gone higher than $b$, which by Definition~\ref{def:sneaky} satisfies $b \geq \hat{p}_1 - 2^{w_1 - \csneak}$.  But then the divergent point at time $t_{\mathrm{jump}_1}$ would have been less than $b$, since Alice could not have resolved it, a contradiction.) 

Note that if $\hat{p}_2 \leq \hat{p}_1 $ then at the time $t_{\mathrm{jump}_2} < t_{\mathrm{jump}_1}$, Alice and Bob simultaneously jump to a common point  strictly above $\hat{p}_1$, and in particular $\Aiter(t_{\mathrm{jump}_1}) < \hat{p}_1$, contradicting the above.  Therefore $\hat{p}_2 > \hat{p}_1$. 

Now we have two subcases, depending on whether $\tqhattwo$ occurs before or after $\tcqone$.

\textbf{Subcase 1:} $\tqhattwo < \tcqone$, \textit{that is, Alice and Bob jump to $\hat{q}_2$ before Alice reaches $c_{q_1}$.} 
In this scenario we will show that the diving windows are disjoint. Notice that according to Definition~\ref{def:sneaky}, for each $i\in\{1,2\}$, $\hat{q}_i = \hat{p}_i$. Now take any depth $\ell \in [\hat{p}_2, c_{q_2}] $, let $I_1$ be the iteration in $\I_{\mathrm{dive}_1}$ corresponding to this depth and similarly define $I_2$ for $\I_{\mathrm{dive}_2}$. Then $I_1 > I_2$ because at time $t_{\hat{q}_2}$ both Alice and Bob jump to the point $\hat{p}_2$ and as a result in order for Alice to reach the point $c_{q_1} \geq \ell$ she needs to pass the depth $\ell$ at least one more time. Thus $I_2$ is not the last time Alice has visited $\ell$ before $c_{q_1}$ in $\sone$. This holds for all $\ell$ in the defined interval thus the diving windows are disjoint. (We note that they are disjoint as \emph{sets}, although they may be interleaved.)

\textbf{Subcase 2:} $\tqhattwo > \tcqone$, \textit{that is, Alice and Bob jump to $\hat{q}_2$ after Alice reaches the point $c_{q_1}.$} We will show that this scenario is not possible. 
First, we claim that 
\[ \hat{q}_2 = \hat{p}_2 \leq \hat{p}_1 + 2^{w_1 - \csneak}.\]
Indeed, this is because Bob must stay above $\hat{p}_1 + 2^{w_1 - \csneak}$ until $t_{b_1}$. Further, we claim that $t_{\hat{q}_2} < t_{b_1}$, which would imply that $\hat{q}_2 \leq \hat{p}_1 + 2^{w_1 - \csneak}$, as desired. To see why $t_{\hat{q}_2} < t_{b_1}$, for the sake of contradiction, assume otherwise.  Then, according to Definition~\ref{def:sneaky}, for the attack $\cS_1$, the bad spell starting at $t_{b_1}$ ends at ${\tjmp}_1$ with the jump to $p_1$. As a result, we have $\ell^- > 0$ throughout the window $[t_{b_1}, {\tjmp}_1)$.  This means that if $t_{\hat{q}_2} \in [t_{b_1}, {\tjmp}_1 )$, then $\ell^- (t_{\hat{q}_2}) > 0 $ and thus Alice and Bob cannot end a bad spell for $\cS_2$ at time $t_{\hat{q}_2}$, contradicting the definition of a sneaky attack for $\cS_2$. This implies that $t_{\hat{q}_2} < t_{b_1}$ and hence $\hat{q}_2 \leq \hat{p}_q + 2^{w_1 - \csneak}$.

However, at time $t_{c_{q_1}}$, Alice reaches $c_{q_1}$ and has forgotten all points in the range $(\hat{p}_1, \hat{p}_1 + 2^{w_1-2})$, which includes $\hat{q}_2$, as $\csneak > 2$.  But this is a contradiction, as Alice would not be able to jump to $\hat{q}_2$ at time $t_{\hat{q}_2}$.

    \end{proof}    This completes the proof of the claim, which proves the lemma.
\end{proof}

\begin{lemma}[The total rewind due to sneaky attacks is bounded]  \label{lem:total-rewind-upb} Let $ S := \{ \mathcal{S}_1 , \mathcal{S}_2 , ... , \mathcal{S}_m\} $ be defined as the set of all sneaky attacks that were completed  during the execution of Algorithm ~\ref{alg:adaptive}. 
Define $\ell^\ast$ to be the total rewind length caused by all the sneaky attacks, i.e., $\ell^\ast = \sum_{i=1}^m \Delta \ell^+ (t_{\mathrm{jump_i}})$, 
where $\Delta\ell^+(t_{\mathrm{jump_i}})$ denotes the change in the value of $\ell^+$ caused by the jump occurring at time $t_{\mathrm{jump_i}}$. Then $\ell^\ast \leq c (  Q+ \HC_s^d ) $ for some constant $c$, where $Q$ is the total number of iterations having corrupted communication or randomness during Algorithm~\ref{alg:adaptive}, and $\HC_s^d$ is the total number of dangerous small hash collisions during dangerous iterations throughout Algorithm~\ref{alg:adaptive}. 
\end{lemma}
\begin{proof}
Recall from Lemma  \ref{lem:divwindow-corr-coll} that, for every sneaky attack $\mathcal{S}_i$,  
there are at least $(1 - 2^{1 - \csneak})|\divewindow|$ iterations where either a small hash collision, a corruption, or an iteration with corrupted randomness occurred.
In particular, this implies that 
\[ ( 1 - 2^{1- \csneak})|\divewindow| \leq  Q_{\mathcal{S}_i}+ H_{\cS_i}, \]
where $Q_{\cS_i}$ is the number of iterations having corrupted randomness or communication during $\cS_i$, and $H_{\cS_i}$ is the number of dangerous small hash collisions during $\cS_i$.

By Claim~\ref{claim:div-disjoint}, we know that the diving windows of sneaky attacks are disjoint, as a result each small hash collision corresponds to diving window of at most one sneaky attack. Then, $\sum_{i=0}^{k} H_{\cS_i} \leq \HC_s^d$. 

Similarly, each corruption also corresponds to at most a single sneaky attack, which implies that  $\sum_{i=1}^m Q_{\mathcal{S}_i} \leq Q$ . As a result, 
\begin{equation*}
    \sum_{i=0}^k |\divewindows{i}| \leq c' \left( \sum_{i=0}^k Q_{\mathcal{S}_i} + \sum_{i=0}^k H_{S_i} \right)  \leq  c' ( Q + \HC_s^d ),
\end{equation*}
where $c' = 1/(1 - 2^{1-\csneak})$ is a constant, and where $\divewindows{i}$ is the diving window for $\cS_i$.

 Finally, let $\cS$ be a sneaky attack that completes at time $\tjmp$, and suppose that $\cS$ has scale $w$ at time $\tjmp$.  Since in the definition of a sneaky attack, we have $k < 2^{w+1}$, we conclude that $ j \leq w $. Then, from Definition~\ref{def:transition-candidates}, the scale of the jump is at most $ w+1$ for either party. As a result,
 \[ \Delta \ell^+(\tjmp) \leq 2^{w+2}. \]

 On the other hand, by Remark~\ref{rem:div-size}, the size of the diving window for such an attack satisfies
 \[ |\divewindow| = 2^{w-1}. \]
 Therefore we can conclude that 
 \[\Delta \ell^+(\tjmp) \leq 16 | \divewindow |. \]
Summing up over all of the completed sneaky attacks, we conclude that
 \begin{equation*}
     \ell^* \leq 16  \sum_{i=0}^k |\divewindows{i}| \leq 16 c' ( Q+ \HC_s^d ),
 \end{equation*}
 as desired.
\end{proof}

Finally, we state our main technical lemma, which roughly says that if Alice and Bob rewind a lot relative to $L^-$, then either $\BVC$ increases a lot, or a sneaky attack is in progress.

\begin{restatable}{lemma}{mainTechLem}[Main Technical Lemma] \label{lem:maintech}
Fix a time $\tsp$, and suppose  that $k_A(\tsp) = k_B(\tsp) =: k$.  
Let $\jsp = \lfloor \log k \rfloor.$
Suppose that $p$ is the next available \specialptf, and that the next jumpable scale is $j_p$.
Let $w$ be such that $p$ is a scale-$w$ MP for whichever party is deeper at time $\tsp$. Further let $b$ be the divergent point at time $\tsp$ and let $t_b$ be the time when $b$ became the divergent point. 
Suppose that 
\begin{equation}
L^-(\tsp) < 2^{j_p-\csneak} \leq 2^{j_p-3}, 
\end{equation}
recalling from Definition~\ref{def:csneak} that $\csneak \geq 3$ is the constant from Definition~\ref{def:sneaky}.

 Then at least one of the following must occur:

\begin{enumerate}
 \renewcommand{\labelenumi}{(\arabic{enumi})}
    \item $\BVC_{AB}(\tsp) \geq \frac{0.6}{4}k$; or
    \item $t_b$ is during the Big Hash Computation phase at the end of a block. 
    \item A \protosneaky heading towards $p$ is in progress for either Alice or Bob at time $\tsp$.
\end{enumerate}

\end{restatable}

As the proof of Lemma~\ref{lem:maintech} is quite long, we defer it to its own section at the end of the paper, Section~\ref{sec:grossproof}.

%% file: potential.tex
In this section we will analyze how the potential function evolves throughout Algorithm~\ref{alg:adaptive}. We analyze this in two steps. First we study the potential change per iteration and we show that the potential makes progress during most iterations. Next we analyze the potential per block, showing that the big hash computation  that happens at the end of the block minimally effects the potential function.  Finally, we conclude by presenting upper and lower bounds on the final values of the potential function. 

\subsubsection{Potential Change Per Iteration}
\begin{lemma}[Potential function makes progress]
\label{lem:progress}
There exist constants $C_1 < C_2 < \cdots < C_6$ (from the definition of $\Phi$ in \eqref{eq:potential}), and a constant $C^-$, so that the following holds.

Suppose that no big hash collision ever occurs throughout the execution of Algorithm~\ref{alg:adaptive}. 
Let $I \in [\Itotal]$, and let $t_0$ be the time at the beginning of iteration $I$ (Line~\ref{line:iterstart}) and let $t_1$ be the time at the end of iteration $I$ (Line~\ref{line:iterend}).
Then either a \sneakyjmpf\ occurs in iteration $I$, or the following holds.
\begin{itemize}
    \item [1.] Suppose that:
    \begin{itemize}
        \item The adversary does not introduce a corruption during iteration $I$;
        \item Iteration $I$ does not have corrupted randomness;
        \item There is no small hash collision in iteration $I$. 

    \end{itemize}

    Then  $\Phi(t_1) - \Phi(t_0) \geq 1$. 
   \item  [2.] Otherwise, $\Phi(t_1) - \Phi(t_0) \geq -C^{-}$.
    \end{itemize}
\end{lemma}

\begin{proof}
    We will prove this statement by analyzing the different phases of an iteration separately. 
    In Claim~\ref{claim:potential-progress-verf-comp}, we focus on the Verification and Computation Phases.  We show that during these phases, the change in $\Phi$ is bounded in magnitude by some constant $C^-$. We also show that, if the conditions of Item 1 above are satisfied, then in fact $\Phi$ increases by at least one during these phases. 
    
    In Claim~\ref{claim:potential-progress-transition} we focus on the Transition Phase. 
 We show that, unless a \sneakyjmp occurred in iteration $I$, $\Phi$ cannot decrease during the Transition Phase.

    The two claims together will prove the lemma. 
    
    Before we begin, we re-state the definition of the potential function $\Phi$ for the reader's convenience:
    \[ \Phi = \begin{cases} \ell^+ - C_3 \ell^- - C_2 L^- + C_1 k_{AB} - C_5 E_{AB} - 2C_6 \BVC_{AB} & \text{if } k_A = k_B \\
\ell^+ - C_3 \ell^- - C_2 L^- - 0.9 C_4 k_{AB} + C_4 E_{AB} - C_6 \BVC_{AB} & \text{if } k_A \neq k_B
\end{cases} \]

    We begin with the Verification and Computation Phases.
 
\begin{claim}[Potential progress in Verification and Computation Phases] \label{claim:potential-progress-verf-comp}
Assume the hypotheses of Lemma~\ref{lem:progress}.  Let $t_2$ be the time that the computation phase ends.  Then:

\begin{itemize}
   \item [1.] If there is neither an adversarial corruption to randomness or communication nor a  small hash collision in iteration $I$, then   $ \Phi(t_2) - \Phi(t_0) \geq 1 $. 

    \item [2.] Otherwise, $\Phi(t_2) - \Phi(t_0) \geq -C^{-}$ for some fixed constant $C^-$.  
  
\end{itemize}

\end{claim}
\begin{proof} 
Throughout the proof for each variable $\texttt{var}$, we will use $\Delta \texttt{var}$ to denote the change in $\texttt{var}$ during the Verification and Computation Phases.  (Note that we will overload this notation in the proof of the next claim, when we focus on the transition phase.) 
 That is, for this claim, $\Delta(\texttt{var}) := \texttt{var}(t_2)- \texttt{var}(t_0)$.

First we observe that each term in $\Phi$ changes by at most a constant during the Verification and Computation phases. Indeed, $\ell^+$ can at most increase by one and cannot decrease at all; $k$ could get reset to zero at the end of the Computation Phase, but this line only executes if $k=1$, so $k$ can change by at most $1$; and $E$ may increase by $1$ during the Verification Phase, and may get set to $0$ at the end of the Computation Phase, but that only executes if $E=0$. As for the negatively appearing terms, $\ell^-$ and $L^-$ can both increase by at most $1$, and $\BVC$ does not change. This implies that $\Delta \Phi \geq -C^-$, for some constant $C^-$, no matter what.

Now suppose that the hypotheses of Item 1 are satisfied; we will show that $\Delta \Phi \geq 1$.
First we notice that, as there are no hash collisions, corruptions, or corrupted randomness in iteration $I$, by 
Lemma~\ref{lem::BVC-coll-corr},
$\BVC_{AB}$ cannot increase.  
We break down the analysis into three cases.
\\

\begin{itemize}

    \item \textbf{Case 1: Simulation proceeds.}

      In this case, we consider the setting where at least one of Alice and Bob simulated $\Pi$ during the computation phase. Recall that we are assuming there are no small hash collisions, corruptions, or corrupted randomness.  In this case, we claim that $\ell^-(t_0) = \ell^-(t_2) = 0$  and further, \emph{both} Alice and Bob simulate $\Pi$ in the computation phase.

      First, we show that $\ell^-(t_0) = 0$. Let $\megastate_A = \ms{\Aiter(t_0)}$ and $\megastate_B = \ms{\Biter(t_0)}$. Suppose toward a contradiction that  $\ell^-(t_0) > 0 $,  then by the definition of $\ell^-$ we know that either $\simPath_A(t_0) \neq \simPath_B(t_0)$ or $ \megastate_A \neq \megastate_B$.  From Corollary \ref{cor:transcript-ms} we know that in both cases, $\megastate_A \neq \megastate_B$. Then Corollary~\ref{cor:ms-not-match-then-triplet-not-match} implies that there exists a variable $\var \in  \{\pv , (, \pseed, \phash, \pt, \piter)  \}$ such that $ \var_A \neq \var_B$. 
  
      This variable is hashed with a small hash during the verification phase. Since the simulation proceeded for at least one of Alice or Bob, for that party, after the verification phase, we had that 
      \begin{equation}
          H_{\var} = H'_{\var}
      \end{equation}
  in order to pass the check in line~\ref{line:checksmallhash}. As we are assuming that there are no corruptions or corrupted randomness, this means that there must have been a small hash collision. However this contradicts our initial assumption that there are no small hash collisions during this iteration. Thus we can conclude that $ \megastate_A = \megastate_B$ and as a result $\ell^- = 0$.

      Given that $\ell^-(t_0)=0$, now we show that both parties will simulate during the computation phase. 
      In order to simulate, per Algorithm~\ref{alg:adaptive} the following conditions must be satisfied at Line~\ref{line:checksmallhash}:
      \begin{itemize}
          \item The small-hashed value of $\pv$ and $(\pt,\phash, \pseed, \piter)$ needs to line up with the hash sent by the other party
          \item $k=1$
          \item $E=0$
          \item \texttt{Rew = False}
      \end{itemize}
   Without loss of generality, assume that Alice simulated during this iteration. Hence, $E_A=0$, indicating that $H_{k_A} = H^\prime_{k_B}$. Since there were no corrupted communications or randomness, and no small hash collisions, it must be that $k_A = k_B = 1$ during the verification phase.  Indeed, otherwise Alice's copy of $E_A$ would have incremented and would not be zero. Given that $k_B = 1$ during the verification phase, we can infer that Bob's parameters were reset to zero during the previous iteration in line~\ref{line:reset1} or line~\ref{line:reset2} of Algorithm~\ref{alg:adaptive}.  As a result, $E_B=0$ for Bob as well because, as there are no corruptions in this iteration, $H_{k_B} = H^\prime_{k_A}$. Given that \[H_{(\phash, \pseed,\pt, \piter).A} = H_{(\phash, \pseed, \pt,\piter).B}\] and that there were no corruptions during iteration $I$, we have that \[H_{(\phash, \pseed,\pt, \piter).A}^\prime = H_{(\phash, \pseed, \pt , \piter).B},\] so Bob also meets the requirements for simulating an iteration. Hence, both Alice and Bob simulate during this iteration. As there are no corruptions during this iteration we have $\simPath_A(t_2) = \simPath_B(t_2)$. Further, since there is no corrupted randomness during this iteration, we can conclude that $\megastate_A (t_2) = \megastate_B(t_2)$.

    Thus, we conclude that $\ell^-(t_2) = 0$ (and $L^-(t_2) = 0$) and both Alice and Bob simulated $\Pi$ in the Computation Phase of iteration $I$. As noted above, $\Delta \BVC_{AB} \leq 0$. 
    As $k$ is reset to zero at the end of the computation phase, we also have $k_A(t_2) = k_B(t_2) = 0$; and as noted above $k_A(t_0) = k_B(t_0) = 0$.  Thus, $\Delta k = 0$.

   Similarly, $\Delta E = 0$. Finally, $\Delta \ell^+ = 1$, as both Alice and Bob correctly simulated the next round. Altogether, we see that $\Delta \Phi \geq 1$ in this case.

    \item \textbf{Case 2: Simulation doesn't proceed and $k_A = k_B$.} Now, we focus on the case when neither Alice nor Bob simulated $\Pi$ in the Computation Phase, but $k_A = k_B$ at the beginning of iteration $I$.  
    As Alice and Bob do not change their location in $\Pi$ during the Verification and Simulation Phase of this iteration, then parameters depending on the location of Alice and Bob, such as $\ell^-,\ell^+, L^-$, do not change. Since there are no hash collisions or corruptions during this iteration, $\Delta \BVC_{AB} = 0$, and, as $k_A = k_B$ then $\Delta E_{AB}= 0$. 
    Thus, the only parameter changing throughout $[t_0, t_2]$ is $k_{AB} $, which increases by 2. As a result, for $C_1 \geq 0.5 $ we get that $\Delta \Phi  =  C_1 \Delta k_{AB} \geq 1 $. 
    
    \item \textbf{Simulation doesn't proceed and $k_A \neq k_B$.} As above, parameters depending on Alice and Bob's location in the protocol do not change. Also, $\Delta \BVC_{AB}= 0$.
    Then the only parameters changing throughout $[t_0,t_2]$  are $k_{AB} $ and $E_{AB}$. While $k_{AB}$ increases by two, $E_{AB}$ also increases by two, as a result the change in the potential equals, $\Delta \Phi  = 2 ( 1 - 0.9 ) C_4 = 0.2 C_4$. Thus assuming $C_4 \geq 5$, the change in the potential will be at least $1$.
    
\end{itemize}
\end{proof}

Next, we analyze the transition phase.  
 During the transition phase, if a transition happens due to the parameter $E$ exceeding the limit $k/2$, then we refer to this transition as an \textbf{\textit{error transition}}; other cases will be referred to as \textbf{\textit{meeting point transitions}}. 
\begin{claim}[Potential progress in the transition phase] \label{claim:potential-progress-transition}

Suppose no big hash collision ever occurs throughout the execution of algorithm~\ref{alg:adaptive}. 

Let $I\in [\Itotal]$, Let $t_2$ be the end of the computation phase of iteration $I$, as in Claim~\ref{claim:potential-progress-verf-comp}, and recall that $t_1$ is the end of the iteration $I$ (aka, the end of the transition phase). Then $\Delta \Phi = \Phi(t_1) - \Phi(t_2)$ has one of the following behaviors: 
\begin{enumerate}
    \item $\Delta \Phi \geq 0$.
    \item  One error transition occurred, $k_A(I) \neq k_B(I)$, and further either $k_A(I)= 1 $ or $k_B(I) = 1 $. Then $\Delta \Phi \geq -C^-$.  Further, assuming that there are no small hash collisions or corruptions in iteration $I$, we have $\Phi (t_1) - \Phi(t_0) \geq 1$. 
    \item A sneaky attack has occurred. (In this case, $\Delta \Phi$ may not be bounded by a constant. This will be analyzed separately.)
\end{enumerate}

\end{claim}

\begin{proof}
Overloading notation from the previous claim, $\Delta\texttt{var}$ now denotes the change in the value of a variable $\texttt{var}$ before and after the \emph{transition} phase.
If no transitions occur during the transition phase, then no parameter within the potential function is changed, resulting in $\Delta \Phi = 0$, which proves our claim. However, if there is at least one transition, several scenarios are possible for the transition phase, depending on whether $k_A$ and $k_B$ are equal or how many and what type of transition has occurred. We will examine each scenario individually, demonstrating that if there are no sneaky attack jumps during this transition phase, then $\Delta \Phi$ will always be non-negative. There is a special sub-case (corresponding to Item 2 in the Claim) where the potential will drop during the transition phase, but overall, during the full iteration, the potential will still increase by one.

    \begin{itemize}
        \item [(a)] \textit{$k_A \neq k_B$ and one transition occurred. }
        Without loss of generality assume that Alice makes the transition. We treat two cases.
        
        \textbf{Case 1.} \textit{Alice makes an error transition.} Then Alice's and Bob's locations in the original protocol don't change, so consequently parameters $\ell^+, \ell^-, L^-$ -- which depend on the location of the parties -- remain unchanged during this time window. Hence the only parameters that could change are $k_{AB}$, $E_{AB}$, and $ \BVC_{AB}$.  Notice that during the iteration $I-1$, $E_{A}(I-1) \leq 0.5( k_{A} (I-1) = 0.5 ( k_{A}(I)- 1 ) $, and during each iteration $E_A$ can only increase by one, thus during iteration $I$, $E_{A} (I) \leq 0.5 (k_A (I) - 1 ) + 1  = 0.5 k_A(I) + 0.5 $. Now summing up the difference among the parameters according to the potential formula, 
        \begin{align*}
           \Delta \Phi &= C_4( \Delta E_{AB} (I) - 0.9 \Delta k_{AB}(I) ) - C_6 \Delta \BVC_{AB}(I)\\
           & \geq C_4( \Delta E_{AB} (I) - 0.9 \Delta k_{AB} (I) ) \\
           & \geq  C_4 ( - 0.5 k_A(I) - 0.5  + 0.9 k_A(I) )\\
           & = C_4(0.4 k_A(I) - 0.5 ) .
        \end{align*}
        
            \textbf{Subcase 1.1: $k_A(I) > 1$.}  In this case from the above we clearly have $\Delta \Phi \geq C_4(0.4 k_A(I) - 0.5) \geq 0 $.

            \textbf{Subcase 1.2: $k_A(I) = 1$.}
            If $k_A(I) = 1 $, then according to the above calculation, the potential difference may be negative during the transition phase. As a result, we will analyze this case for the full iteration and directly prove the statement of Lemma~\ref{lem:progress}.  Consequently, for this subcase, we use $\Delta \var$ to refer to the change in a variable over the \emph{entire} iteration, not just the transition phase. 
            If $k_A(I) = 1 $, then Alice has reset her parameters in the previous iteration; hence, at the beginning of iteration $I$, $k_A = 0$, and at the end of iteration $I$, $k_A = 0$ again.  Thus, $\Delta k_A(I) =0$.
            Similarly, we have $\Delta E_A = 0 $.   Moreover, $ \Delta \ell^-_{AB} \leq 1 $  and $\Delta L^-_{AB} \leq 1.$  This is because we are in Situation (a), so only Bob is continuing simulation; hence $\ell^-_A$ and $L^-_A$ do not change.

            Now we prove the conclusion of Lemma~\ref{lem:progress} in the case that there may be hash collisions or corruptions or corrupted randomness.
            During the computation and verification phase, according to Claim~\ref{claim:potential-progress-verf-comp}, $\Delta \Phi \geq -C$ for some constant $C$.  As there are no parameter changes for Bob during the transition phase, parameter changes during the transition phase are caused by Alice. Consequently, with  $k_A = 1 $ then the inequalities $E_{A} \leq 0.5 k_A + 0.5 $ and $\BVC_{A} \leq k_A$ show that the negative overall potential change is still bounded by a constant. Thus, the conclusion of Lemma~\ref{lem:progress} still holds in this case.
            
            Next, we prove the conclusion of Lemma~\ref{lem:progress} when there are no hash collisions, corruptions, or corrupted randomness during iteration $I$.  In this case, notice that $\Delta E_{AB} = \Delta E_{A} + \Delta E_{B} = 1 $. Indeed, we have $\Delta E_{A} = 0$ as noted above. 
 Moreover, $E_B$ will increase by one: in case (a), $k_A \neq k_B$, and if there are no hash collisions, corruptions, or corrupted randomness, this means that Bob's copies of $H_k$ and $H_k'$ will disagree on Line~\ref{line:checkk}, and $E_B$ will be incremented. 
 
 Finally, adding up the terms in the potential function we get that
              \begin{align*}
            \Delta \Phi &= - C_3 \Delta \ell^- - C_2 L^- - 0.9 C_4\Delta k_{AB} + C_4 \Delta E_{AB} - C_6 \Delta \BVC_{AB}\\
           & \geq - C_3 \Delta \ell^- - C_2 L^- - 0.9 C_4\Delta k_{AB} +C_4 \Delta E_{AB}\\
            & = - C_3 \Delta \ell^- - C_2 L^- - 0.9 C_4 (\Delta k_A + \Delta k_B ) + C_4 \ ( \Delta E_A + \Delta E_B)\\
            & \geq - C_3 - C_2 - 0.9 C_4 +   C_4 \\
            & = -C_3 - C_2 + 0.1 C_4 \ .
        \end{align*}
        This guarantees that $\Phi$ will increase by at least one over the course of the whole iteration $I$, for a large enough choice of $C_4$.

        \textbf{Case 2.} \emph{Alice made a meeting point transition.}  Abbreviate $j_A = j$, so $k_A = 2^{j+1} - 1$. Observe that the length of Alice's jump is at most $4k$. This limitation exists because in each jump, the parties consider meeting points at scales $j$ and $j+1$, and $\MPthree$ is at a higher depth than $\MPtwo$. This is because, according to Definition~\ref{def:transition-candidates}, $\MPthree$ is the deepest point $p$ such that $2^j | p$. Since $2^j$ divides $\MPtwo$, then $\MPthree$ must be at a higher depth. Therefore, the largest jump occurs when Alice decides to jump to $\MPone$, where $\Delta \ell_A \leq 2^{j+2} - 1 \leq 4k_A$. As a result, all parameters dependent on Alice's location will change by a maximum of $4k_A$. Specifically, we have $\Delta \ell^+ \leq 4k_A$, $\Delta \ell^- \leq 4k_A$, and $\Delta L^- \leq 4k_A$. Additionally, as Alice made a meeting point transition, $\Delta E_{AB} = -E_A$ and $E_A \leq 0.5k_A$. Finally, since Alice is the only party transitioning, $\Delta k_{AB} = -k_A$. Summing up the terms according to the potential function, we find that
        \begin{align*}
        \Delta \Phi &= \Delta \ell^+ - C_3 \Delta \ell^- - C_2 \Delta L^- - 0.9 C_4 \Delta k_A + C_4 \Delta E_A - C_6 \Delta \BVC_{AB}\\
        &\geq \Delta \ell^+ - C_3 \Delta \ell^- - C_2 \Delta L^- - 0.9 C_4 \Delta k_A + C_4 \Delta E_A \\
        &\geq -4k_A - 4(C_3 + C_2) k_A + 0.9 C_4 k_A - C_4 E_A \\
        &\geq -4k_A - 4(C_3 + C_2) k_A + 0.9 C_4 k_A - 0.5 C_4 k_A \\
        &= k_A ( 0.4 C_4 - 4(1 + C_3 + C_2) ) \ .
    \end{align*}
        For large enough $C_4$, this ensures a non-negative change in the potential.
        
        \item [(b)] \textit{$k_A \neq k_B$ and two transitions happen.}
        For all types of transitions, from Item (a), Case 2, we know that $|\Delta \ell^+|  \leq 4 k_{AB}$, $ |\Delta \ell^- | \leq 4 k_{AB}$, $  | \Delta L^- | \leq 4 k_{AB} $. Additionally, from Item (a), Case 1,  we know that $ E_{A} \leq 0.5 k_A + 0.5 $ and $ E_B \leq 0.5 k_B + 0.5  $ . Adding the terms up we see that  
        \begin{align*}
            \Delta \Phi  &= \Delta \ell^ + - C_3 \Delta \ell^- - C_2 \Delta L^- - 0.9 C_4 \Delta k_{AB} + C_4 \Delta E_{AB} - C_6 \Delta \BVC_{AB}\\ 
            &  \geq \Delta \ell^+ -  C_3 \Delta \ell^- - C_2 \delta L^- - 0.9 C_4 \Delta k_{AB} + C_4 \Delta E_{AB} \\ 
            & \geq -4 k_{AB}  - 4C_3 k_{AB} - 4C_2 k_{AB} +  0.9 C_4 k_{AB} - C_4 ( 0.5 k_{AB} +  1 ) \\ 
            &  \geq k_{AB}  ( -4 ( 1 + C_3 + C_2 )   + 0.9 C_4 - 0.5 C_4 - \frac{C_4}{k_{AB}} ) \\
            & \geq k_{AB}  \left( C_4 ( 0.9 - 0.5 - \frac{1}{k_{AB}} \right) - 4 (1 + C_2 + C_3)
        \end{align*}
       The term $\BVC_{AB}$ was removed after the first inequality because $\BVC$ is reset to zero after any transition. This reset has a positive impact on the potential.

     Considering that $k_A \neq k_B$, it follows that $k_{AB} \geq 3$, which implies $0.9 - 0.5 - \frac{1}{k_{AB}} \geq 0.4 - \frac{1}{3} > 0.06$.

        Then assuming $C_4 > \frac{4(1+ C_2 + C_3)}{0.06}$ the change in the potential is non-negative.

        \item [(c)]  \textit{$k_A = k_B = k $ and at least one error transition happened.} As there are either one or two error transitions, the positions of Alice and Bob remain unchanged during the transition phase. Therefore, $\Delta \ell^+ = 0, \Delta \ell^- = 0, \Delta L^- = 0$. Furthermore, since there is at least one error transition, we can conclude that $\Delta E_{AB} \geq 0.5k$ and $\Delta k_{AB} \leq 2k$. For the potential change, we observe that
        \begin{align*}
            \Delta \Phi &= C_1 \Delta k_{AB} - C_5 \Delta E_{AB} -2C_6 \Delta \BVC_{AB}\\
            &\geq  C_1 \Delta k_{AB} - C_5 \Delta E_{AB}  \\
            & \geq - C_1 2k + 0.5 C_5 k \\
            &  = k ( - 2 C_1 + 0.5 C_5) \ .
        \end{align*}
        For large enough $C_5$, this results in a non-negative change of potential.
        \item [(d)]  \textit{$k_A = k_B = k $ and one meeting point transition happened.} Similar to previous cases we have that $\Delta k_{AB} \leq 2k$ , $|\Delta \ell^+| \leq 4k_{AB}$, $ \Delta \ell^- \leq 4k_{AB}$ and $\Delta L^- \leq 4 k_{AB} $. If only one meeting point transition happens there are two possible explanations: Either the meeting point was available to both parties and one party didn't count enough votes, which implies that there were many bad votes so we will have $\BVC_{AB} >  0.6 k $. Otherwise the meeting point was not available to one of the parties and the other party jumped by over-counting the votes. This causes an increase in $\BVC$ of at least $0.4 k$, implying $\BVC_{AB} \geq 0.4 k $. As a result the change in the potential will be dominated by the change in the $\BVC_{AB}$. Formally:  
        \begin{align*}
            \Delta \Phi &= \Delta \ell^+ - C_3 \Delta \ell^- - C_2 \Delta L^- + C_1 \Delta k_{AB}  - C_5 \Delta E_{AB} - 2 C_6 \Delta \BVC_{AB}  \\ 
            & \geq \Delta \ell^+ -  C_3 \Delta \ell^- - C_2 \Delta L^-+ C_1 \Delta k_{AB} - C_5 -  2 C_6 \Delta \BVC_{AB} \\ 
            & \geq -4k_{AB}  - C_3 4k_{AB} - C_2 4k_{AB}  - C_1 2k_{AB}  +  2 C_6 0.4 k_{AB} - C_5\\
            & \geq k_{AB} ( -4( 1 + C _3 + C_2 ) - 2 C_2 + 0.8 C_6 ) - C_5 \\
                & \geq 2k ( -4( 1 + C _3 + C_2 ) - 2 C_2 + 0.8 C_6 ) - C_5.
        \end{align*}
        Thus, for large enough $C_6$, the potential difference will be non-negative. The $C_5$ carrying over throughout the computation is due to the fact that the $E$ value for the party not making the transition may increase by $1$ during this iteration. The party making a transition will reset the corresponding $E$ value to zero, which only has positive effect in the potential change.
        
        \item [(e)]  \textit{$k_A = k_B = k $ and two meeting point transitions happened.}
        Similar to previous items, $|\Delta \ell^+| \leq 4k_{AB}$, $ \Delta \ell^- \leq 4k_{AB}$ and $\Delta L^- \leq 4 k_{AB} $. Also $E_{AB}$ getting reset to zero will only have a positive impact on the potential change.
        
        If $\ell^- (t_1) \neq  0  $ and Alice and Bob have jumped to points $p_A$ and $p_B$, the mega states corresponding to $p_A$ and $p_B$ do not agree. If one of the meeting points is not an available transition candidate  for the other party ($p_A$ for Bob or $p_B$ for Alice), then $\BVC_{AB} > 0.4 k $ because votes counted towards this jump were bad votes. Otherwise, $p_A$ and $p_B$ are available to both parties as transition candidates. Without loss of generality, assume that $p_A > p_B$. During the transition, priority is given to the candidate which is the deeper  and has acquired enough votes. Since Bob has not transitioned to $p_A$ even though it is an available candidate for both parties, the votes for $p_A$ were under-counted on Bob's side, implying $\BVC_{AB} \geq 0.6k $. Therefore, in both cases, the potential difference will be dominated by $\Delta \BVC_{AB} = \BVC_{AB}$. Formally:

        \begin{align*}
         \Delta \Phi &= \Delta \ell^+ - C_3 \Delta \ell^- - C_2 \Delta L^- + C_1 \Delta k_{AB}  - C_5 \Delta E_{AB} - 2 C_6 \Delta \BVC_{AB}  \\ 
         &\geq \Delta \ell^+ - C_3 \Delta \ell^- - C_2 \Delta L^- + C_1 \Delta k_{AB} - 2C_6 \Delta \BVC_{AB} \\ 
            & \geq - 4 k_{AB}  - 4 C_3 k_{AB} - 4 C_2 -  2 C_1 k_{AB}  + 2C_6 0.6 k \\
                        & \geq - 8 k  - 8 C_3 k - 8 C_2 -  4 C_1 k  + 2C_6 0.6 k \\
            & = k ( 1.2 C_6 - 8( 1 + C_2 + C_3) - 4C_1  ) 
        \end{align*}
        Hence, for sufficiently large $C_6$, the potential change is non-negative.
        
        Now assume that $\ell^-(t_1) = 0 $. This means that both parties jumped to a common point $p$. Let $j_p$ be the scale of this jump; note that $j_p$ was the next jumpable scale before the jump occurred. If $ L^- \geq 2^{j_p-\csneak}$, we claim that $k_{AB} \leq 2k  \leq 2^{\csneak + 2 } L^-$. This is because $k\leq 2^{j_p+1}$ as $j_p$ was the scale of the jump. Then the potential change is dominated by the change in $L^-$. Formally, letting $\gamma = 2^{\csneak + 2}$, we have:
        \begin{align*}
            \Delta \Phi &= \Delta \ell^+ - C_3 \Delta \ell^- - C_2 \Delta L^- + C_1 \Delta k_{AB}  - C_5 \Delta E_{AB} - 2 C_6 \Delta \BVC_{AB}\\ 
            &\geq \Delta \ell^+ - C_2 \Delta L^- + C_1 \Delta k_{AB}  \\
             & \geq 2 \gamma L^- + C_2 L^- - 2 \gamma C_1 L^- \\
             & = L^- (   C_2 - \gamma(  2 + C_1 ) )  
        \end{align*}
        For large enough $C_2$, the change in the potential is indeed non-negative. 
        
        Otherwise, if $L^- < 2^{j_p - \csneak} $, then according to Lemma~\ref{lem:maintech}, two cases are possible: either  $\Delta \BVC_{AB}$ is large or a sneaky attack jump occurred. The case of sneaky attack jump is treated separately later in our analysis. For now, assume that $\BVC_{AB} \geq \frac{0.6}{4} k$. Then, the potential change is dominated by $\BVC_{AB}$. Formally:  
        \begin{align*}
          \Delta \Phi &= \Delta \ell^+ - C_3 \Delta \ell^- - C_2 \Delta L^- + C_1 \Delta k_{AB}  - C_5 \Delta E_{AB} - 2 C_6 \Delta \BVC_{AB} \\ 
          & \geq \Delta \ell^+ + C_1 \Delta k_{AB} - 2C_6 \Delta \BVC_{AB} \\
          & \geq -  8k  - 8 C_1 k + C_6 \frac{0.6}{2}k \\ 
          & = k \left( \frac{0.6}{2}C_6 -8 - 8 C_1 \right).
        \end{align*}
        As a result for large enough $C_6$, the potential difference will be non-negative.

    \end{itemize}

This proves the claim about the behavior of $\Phi$ during the transition phase.
\end{proof}
Together, Claims~\ref{claim:potential-progress-verf-comp} and \ref{claim:potential-progress-transition} prove Lemma~\ref{lem:progress}.
\end{proof}

\begin{remark}[Potential change when a sneaky attack occurs]\label{rem:sneaky-jump-potential}
Lemma~\ref{lem:progress} shows that the potential function $\Phi$ is well-behaved if a sneaky attack does not occur.  
Suppose that a sneaky attack jump \emph{does} occur during a transition phrase, completing a sneaky attack $\cS$. Define $q_{\cS}$ and $H_{\cS}$ as in Lemma \ref{lem:one-sneaky-counter-bound-coll-corr}. Then the proof of Lemma~\ref{lem:progress} shows that $\Phi$ decreases by at most $\Delta |\ell^+| + | C_1 \Delta k_{AB}| \leq \Delta |\ell^+| + C_1  k_{AB}$, as the reset of values in the rest of the terms defining $\Phi$ impact the potential change positively. 

In Lemma~\ref{lem:totaldive}, we will bound the total contribution of this term over all sneaky jumps.

\end{remark}

While Lemma~\ref{lem:progress} says that we may not make progress if small hash collisions occur, in fact we will still make progress in non-dangerous rounds, even if there are hash collisions.  We record this fact in the following lemma.

\begin{lemma}\label{lem:non-dangerous-fine}
    Let $I$ be a non-dangerous iteration, so that $I$ does not have any corruptions or corrupted randomness.  Then $\Phi$ increases by $1$ over the course of iteration $I$.
\end{lemma}
\begin{proof}
    We begin with a claim that says that, under the conditions of the lemma, we only have certain types of small hash collisions.
    \begin{claim}\label{cl:only_MP_collisions}
    Given the assumptions of the lemma, small hash collisions only occur between variables related to meeting point candidates.  That is, small hash collisions can only occur on lines~\ref{line:ms-hash-v},\ref{line:ms-hash-small-big-hash},\ref{line:ms-hash-depth} corresponding to the second point in the definition of a small hash collision (Definition~\ref{def:hashcollisions}).
    \end{claim}
    \begin{proof}
        Since $I$ is non-dangerous, we have that $k_A = k_B = 1$ at the beginning of iteration $I$, and also that $\ell^- = 0$, which implies that $\megastate_A = \megastate_B$ at the beginning of iteration $I$.  As a result, all of the variables \texttt{var} in line~\ref{line:var} are the same for Alice and Bob: $\mathtt{var}_A = \mathtt{var}_B$.  Thus, it is not possible to have a hash collision on line~\ref{line:varhash}. This proves the claim.
    \end{proof}
    Now, notice that transitions do not occur in non-dangerous iterations.   However, the transition phase is the only time that small hash collisions related to meeting point candidates might matter.  Thus, Claim~\ref{cl:only_MP_collisions} implies that hash collisions do not affect the execution of $\Pi'$ during iteration $I$.  As a result, we may assume that there are no hash collisions at all, in which case the analysis of Lemma~\ref{lem:progress} applies, and we conclude that $\Phi$ increases by one over the course of iteration $I$, as desired.
\end{proof}

\subsubsection{Potential Change Per Block}

 One challenge that we have to deal with in the analysis is that there are two cases in which the potential function could decrease by more than a constant. The more challenging case is that of a sneaky attack, which is analyzed in Lemma~\ref{lem:total-rewind-upb}.  Another case that we must deal with is the case that the adversary introduced enough corruptions to interfere with the randomness exchange performed in Line~\ref{line:alice-send-bigrand}. This second case is much easier to account for. Intuitively, we first observe that as the randomness is protected by an error-correcting code, if the adversary causes a decoding error in the randomness exchange then it must have invested a large number of corruptions. This implies that there are not too many blocks with corrupted randomness, as we make formal in the following lemma.
\begin{lemma}\label{lem:num-block-corr-rand}
    Let $\Bcorr$ be the number of blocks in which Alice or Bob have different seeds for the big hash computation and let $Q$ be the total number of corruptions so far. There are at most $Q/(2\Iblock)$ such blocks. 
\end{lemma}

\begin{proof}
    By Theorem~\ref{thm:goodECC}, the binary error-correcting code encoding the bit string $R$ has minimum distance at least $4\Iblock$. Thus, for decoding to fail, at least $2\Iblock$ corruptions must be introduced by the adversary. As the adversary has introduced $Q$ corruptions so far, we have 
    \[
       \Bcorr \leq \frac{Q}{2\Iblock} \ .
    \]
\end{proof}
Next, we show that the damage caused by each such block is not too large.
\begin{lemma} \label{lem:block-phi-dec}
    Fix a block $B$ and suppose the randomness exchange of Line~\ref{line:alice-send-bigrand} $\Rblock^b$ is corrupted. Then the potential function $\Phi$ decreases by at most $3 \Iblock$ during the Big Hash Computation phase at the end of Block $B$.
\end{lemma}

\begin{proof}
    Clearly, $\BVC_{AB}$, $k_{AB}$ and $E_{AB}$ do not change during the Big Hash Computation phase. However, if $\Rblock^b$ is corrupted, then  Alice and Bob will in general not have matching \emph{mega-states} for the depths $p$ that were simulated in this block -- parameters such as $\phash$, $\pseed$, $\piter$ and $\pt$ can now differ between the parties, as they were retro-actively updated at the end of the block based on $\Rblock^b$. Thus, $\ell^+$ can decrease by at most $\Iblock$, while $\ell^-$ and $L^-$ can increase by at most $\Iblock$ as well. It therefore follows that the total potential decrease is at most $3\Iblock$, as claimed. 
\end{proof}

Finally we argue that we do not have to worry about the potential function changing during the Big Hash Computation phase in blocks where $\Rblock^b$ is not corrupted.

\begin{lemma}\label{lem:block-phi-no-dec}
    Let $B$ be a block such that $\Rblock^b$ is not corrupted. Then during the Big Hash Computation phase, the potential function $\Phi$ does not change. 
\end{lemma}
\begin{proof}
    During the big hash computation variables $\BVC_{AB}, E_{AB}$ and $k_{AB}$ do not change. The following claim will allow us to establish that the variables $\ell^+, \ell^-$ and $L^-$ do not change. 
    \begin{claim}
        Let $t_0$ be right before the start of big hash computation in line~\ref{line:big-hash-phase} and let $t_1$ be the end of the block. Further, let $\megastate_A$ and $\megastate_B$ be any two megastates from Alice and Bob's memory. Then $\megastate_A(t_1) = \megastate_B(t_1)$ if and only if $\megastate_A (t_0) = \megastate_B(t_0) $. 
    \end{claim}
    \begin{proof}
        If $\megastate_A (t_0) = \megastate_B(t_0)$, then $ \piter_A = \piter_B $ at time $t_0$. As a result, either the both megastates are simulated during the block $B$---in which case both Alice and Bob use the same variables as inputs to the big hash function, implying their outputs will be the same---or both megastates were simulated in a different block, in which case none of the variables within the megastates are updated during the big hash computation. Hence, we have that $\megastate_A(t_1) = \megastate_B(t_1)$.
        If $\megastate_A(t_0) \neq \megastate_B(t_0)$ then Claim~\ref{claim:neq-in-and-out-of-block} implies $\megastate_A(t_1) \neq \megastate(t_1)$. 
    \end{proof}
   The above claim implies $\ell^+$ and $\ell^-$ do not change, which also implies $L^-$ does not change. Thus $\Phi (t_1) - \Phi(t_0) = 0$.
\end{proof}

\subsubsection{Upper Bound on $\Phi$}\label{sec:upperBd} 
In this section, we provide an upper bound on the potential function $\Phi$. This is presented in Lemma~\ref{lem:phiub} below.

\begin{lemma}[Upper bound on $\Phi$]\label{lem:phiub}
Suppose that the constants $C_2, C_6$ from \eqref{eq:potential} are sufficiently large, in terms of $C_1$ from \eqref{eq:potential}, and $\csneak$ from Lemma~\ref{lem:maintech}. 
     Suppose that there are no big hash collisions.  Fix $I \in [\Itotal].$  Let $D = D(I)$ be the number of dangerous iterations between iteration $1$ and iteration $I$.  Then 
    \[ \Phi(I) \leq \ell^+(I) +  \cupper  ( Q + \HC_s^d )  \leq I - D + \cupper  ( Q + \HC_s^d ) ,\]
    for some constant $\cupper$, where $\HC_s^d$ is the number of dangerous iterations with small hash collisions throughout the execution of Algorithm~\ref{alg:adaptive}. 

\end{lemma}
\begin{proof}
    Recall from \eqref{eq:potential} that
   \[
\Phi = \begin{cases} \ell^+ - C_3 \ell^- - C_2 L^- + C_1 k_{AB} - C_5 E_{AB} - 2C_6 \BVC_{AB} & \text{if } k_A = k_B \\
\ell^+ - C_3 \ell^- - C_2 L^- - 0.9 C_4 k_{AB} + C_4 E_{AB} - C_6 \BVC_{AB} & \text{if } k_A \neq k_B
\end{cases}
\] 
Here, we consider $\Phi(I)$, which we recall is the value of $\Phi$ calculated at the \emph{end} of iteration $I$.  Below, we drop the ``$I$'' from the notation, and just note that all variables are calculated at the end of iteration $I$.
Let $D = D(I)$ denote the total number of dangerous rounds up to iteration $I$. 

We establish the lemma in each of two cases.

\paragraph{Case A. $k_A \neq k_B$.}
In this case, we have
\[ \Phi \leq \ell^+ + C_4( E_{AB} - 0.9 k_{AB}),\]
by removing terms that are always negative.
Note that $\ell^+ \leq I - D$.  Indeed, the only way that $\ell^+$ can increase is if $\ell^- = 0$ and $k_A = k_B = 1$, because Alice and Bob must both simulate rounds of $\Pi$ (executing Line~\ref{line:simulate}), which requires $k_A = k_B = 1$; and because they must have $\simPath_A = \simPath_B$, which requires $\ell^- = 0$.  

Thus, it suffices to show that
\begin{equation}\label{eq:wtsphi}
E_{AB} \leq 0.9 k_{AB}.
\end{equation}
Indeed, then we would conclude that in this case
\[ \Phi \leq \ell^+ \leq I - D, \]
which would prove the lemma. To see \eqref{eq:wtsphi}, notice that in $I$, either Line~\ref{line:reset2} was executed by a party (say, Alice) or it was not.  If it was executed by Alice, then $E_A = k_A = 0$ at the end of iteration $I$.  If not, then $2E_A < k_A$.  In either case, we have $E_A \leq k_A/2$, which is enough to establish \eqref{eq:wtsphi}.

\paragraph{Case B. $k_A = k_B$.}
 In this case, we can bound
\[ \Phi \leq \ell^+ - C_2L^- + C_1 k_{AB} - 2C_6 \BVC_{AB},\]
where again the inequality holds as we have dropped only negative terms from $\Phi$.
As before, we have $\ell^+ \leq I - D$, so we wish to bound
\begin{equation}\label{eq:wtscase2}
C_1 k_{AB} - C_2 L^- - 2C_6 \BVC_{AB} .
\end{equation}

To establish~\eqref{eq:wtscase2}, 
 we first recall the result of  our main technical Lemma,
Lemma~\ref{lem:maintech}. 

Let $p$  be the next available jumpable point and $j_p$ be next jumpable scale.  Lemma~\ref{lem:maintech} says that three cases are possible.  Either:
\begin{itemize}
    \item The assumption of Lemma~\ref{lem:maintech} does not hold: That is, $L^- \geq  2^{j_p-\csneak} $.
    \item The assumption of Lemma~\ref{lem:maintech} holds, and we have outcome $(1)$: $\BVC_{AB} \geq \frac{0.6}{4} k$, which implies $k_{AB} \leq 16\cdot \BVC_{AB}$.
    \item The assumption of the Lemma~\ref{lem:maintech} holds, and we have outcome $(2)$: a sneaky attack is in progress.
\end{itemize}
We will analyze each case separately.

\textbf{Case B.1:}  
 If $L^- \geq  2^{j_p-\csneak} $, then $k_{AB} < 2k <  \leq 2^{\csneak + 2 } L^-$.  This is because $k \leq 2^{j_p+}$ given that $j_p$ is the next jumpable scale.  Selecting $C_2\geq 2^{\csneak +2 } C_1$ we get that, 
 
\begin{equation}
    C_1 k_{AB} - C_2 L^- - 2C_6 \BVC_{AB} \leq 0. 
\end{equation}

 \textbf{Case B.2:} 

 If $k_{AB} \leq 16 \cdot \BVC_{AB}$, by choosing $C_6 > 16 C_1$ large enough, we have that $2C_6 \BVC_{AB} \geq C_1 k_{AB}$ and hence
\begin{equation}
    C_1 k_{AB} - C_2 L^- - 2C_6 \BVC_{AB} \leq 0 
\end{equation}

\textbf{Case B.3:} A sneaky attack $\cS$ is in progress. 
 In this scenario $k_{AB}$ may not be on the same scale as other parameters in the potential function. However, by Lemma~\ref{lem:one-sneaky-counter-bound-coll-corr}, 
we have
 \[ k_{AB} \leq \alpha(Q_{\cS} + H_{\mathcal{S}}),\]
 for some constant $\alpha$, where $Q_{\cS}$ is the number iterations with corrupted randomness or communication during the diving window of sneaky attack $\cS$, and $H_{\cS}$ is the number of dangerous iterations with small hash collisions that occur during the diving window of $\cS$.

As a result, we can bound
\begin{equation}
    C_1 k_{AB} - C_2 L^- - 2C_6 \BVC_{AB} \leq C_1 \alpha  ( Q_{\mathcal{S}} + H_{\cS}) \leq C_1 \alpha (Q + \HC_s^d ). 
\end{equation}
Setting $\cupper = C_1 \alpha$, altogether we have that

   \[ \Phi(I) \leq \ell^+ +  \cupper  ( Q + \HC_s^d )  \leq I - D + \cupper  ( Q + \HC_s^d ) \]

\end{proof}

\subsubsection{Lower Bound on $\Phi$}\label{sec:lowerBound}

Before we prove our lower bound on $\Phi$, we establish a bound on the total potential decrease caused by sneaky jumps and the Big Hash Computation phase at the end of each block.
Recall from Remark~\ref{rem:sneaky-jump-potential} that the potential decrease during an iteration which includes a sneaky jump is at most $|\Delta \ell^+ (I) | + C_1 k_{AB}( I)$.  In Lemma~\ref{lem:totaldive} below, we bound this quantity.

\begin{lemma}\label{lem:totaldive} 
    Let $\mathcal{I}_{SJ}$ be the set of iterations where a \sneakyjmpf \  occurs.  Then the total decrease in $\Phi$ over all the iterations in $\I_{SJ}$ is bounded by
    \[ \sum_{I \in \I_{SJ}} \max\{ -\Delta\Phi(I), 0\}\leq \sum_{I \in \mathcal{I}_{SJ}} \bigl(|\Delta\ell^+(I)| + C_1 k_{AB}(I)\bigr) \leq  C^{\prime\prime} ( Q + HC^d_s) ,\]
    where $C_1$ is the constant from \eqref{eq:potential} and $C^{\prime\prime}$ is an absolute constant.

\end{lemma} 

\begin{proof} 
The first inequality follows from Remark~\ref{rem:sneaky-jump-potential}.  To establish the second inequality,
   recall from Lemma~\ref{lem:total-rewind-upb}  that the total rewind length $ \ell^\ast =  \sum_{S\in \mathcal{I}_{SJ}}|\Delta\ell^+(I)|$ satisfies $\ell^*   \leq C (  Q+ HC^d_s  )$ and further $\sum_{i \in \mathcal{I}_{SJ}} | \divewindows{i}  | \leq C^\prime( Q +HC^d_s )$.  
   Observe that if a sneaky jump is at scale $w$, then  $ k_{AB} \leq 2^{w+4}$; indeed, the definition of a sneaky attack implies that $k_A = k_B \leq 2^{w+3}$.
   Further, according to Remark~\ref{rem:div-size}, $|\divewindow| \geq 2^{w-1}$. Then we have that $ k_{AB} \leq 32 |\divewindow | $. This implies
   \begin{align*}
         \sum_{i \in \mathcal{I}_{SJ}} (|\Delta\ell^+(I)| + C_1 k_{AB}(I))  &= \ell^\ast +       C_1 \Sigma_{i \in \mathcal{I}_{SJ}}k_{AB}(I) \\
         &\leq \ell^* + 32C_1 \Sigma_{i \in \mathcal{I}_{SJ}} |\divewindow| \\
          & \leq  C ( Q + \HC^d_s )  + 32 C^\prime  (Q + \HC^d_s ) \\ &= C^{\prime\prime} ( Q  + HC^d_s),
   \end{align*}
   defining $C''$ appropriately.
   This proves the lemma.
\end{proof}
 Now we bound the total potential decrease caused by Big Hash Computation phases at the end of each block. 

\begin{lemma}\label{lem:total-block-phi}
    The total potential decrease accrued during the Big Hash Computation phases at the end of all blocks is at most $3/2 Q$.
\end{lemma}
\begin{proof}
    According to Lemma~\ref{lem:block-phi-no-dec}, there are no potential decreases during the Big Hash Computation phase if $\Rblock^b$ is not corrupted. As a result, we only need to count the number of blocks with corrupted $\Rblock^b$, and sum of the potential losses during the big hash computations. Lemma~\ref{lem:num-block-corr-rand} states that there are at most $Q/(2\Iblock)$ blocks with corrupted randomness and further Lemma~\ref{lem:block-phi-dec} shows that the potential decrease during each one of such blocks is at most $3\Iblock$. Hence, the total potential decrease caused by big hash computation is $\frac{Q}{2\Iblock} \times 3 \Iblock = \frac{3}{2} Q$, which proves the claim.
\end{proof}

\begin{lemma}[Lower bound on $\Phi$]\label{lem:philb} 
Let $\HC^d_s$ be the number of iterations with small hash collisions during dangerous iterations.
Let
$S$ be the total number sneaky attack jumps and let $Q$ be the total number of iterations having either corrupted randomness or corrupted communication. Furthermore
suppose that there are no big hash collisions.
Then for any iteration $I \in [\Itotal]$, we have
\begin{equation*}
            \Phi(I) \geq I - ( C^- + 2C^\prime + 1  ) \HC_s^d   - O( \epsilon d )
        \label{eq:phi-upb}
\end{equation*}
for some absolute constant $C'$, where $C^-$ is the constant from Lemma~\ref{lem:progress}.

\end{lemma}
\begin{proof}
   
    In order to compute this lower bound we will look at the amount of potential accumulated or lost during each iteration prior to our current iteration. 
     First we look at dangerous iterations. Let $D$ denote the number of dangerous iterations. 
     If no hash collision, corruption, or sneaky jump occurs, (which is the case in at least $D - \HC_s^d - Q- S$ iterations), then by Lemma~\ref{lem:progress}, the potential function $\Phi$ will increase by at least one.
     If a hash collision or corruption occurs, then again by Lemma~\ref{lem:progress}, the potential will decrease by at most a constant $C^-$.
     Finally, over all of the dangerous iterations,  sneaky jumps overall can decrease the potential by at most $C'(  Q + \HC_s^d )  $ for some absolute constant $C'$, by Lemma~\ref{lem:totaldive}. Further, notice that the number of sneaky attacks $S$ is also upper bounded by $C'(Q + \HC_s^d )$ for some constant $C'$ (which without loss of generality is the same constant, by taking the smaller of the two to be larger), since the number of sneaky attacks cannot exceed the total rewind length caused by all sneaky attacks, which is bounded by a constant times $(Q + \HC_s^d)$ by Lemma~\ref{lem:total-rewind-upb}.
    
     Then if we restrict the potential to only the potential accumulated during dangerous iterations we get
     \begin{align*}
         \Phi_{|\text{Dangerous}} &\geq  (D -\HC_s^d - Q) - S - C^-( Q + \HC_s^d ) - C^\prime ( Q + \HC_s^d )  \\
         &\geq D - (\HC_s^d + Q)  - C^- ( Q + \HC_s^d ) - 2 C^\prime ( Q  + \HC_s^d ) \\
         &= D  -( C^- + 2C^\prime + 1  ) ( Q + \HC_s^d )
     \end{align*}
     
Next we focus on the $I-D$ non-dangerous iterations.  Sneaky attack jumps only occur during dangerous iterations as making a jumping transition is not possible during  non-dangerous iterations. Thus, we only experience corruptions and small hash collisions  in non-dangerous iterations. 

By Lemma~\ref{lem:non-dangerous-fine}, if there are no corruptions, then $\Phi$ increases by $1$ over the course of iteration $I$, regardless of small hash collisions.

On the other hand, if there are corruptions, then the potential function $\Phi$ may decrease by at most $C^-$, as per Lemma~\ref{lem:progress}.

 Consequently, using the same reasoning as for the dangerous iterations above, the potential accumulated in these non-dangerous iterations is lower bounded by
   \begin{align*}
       \Phi_{|\text{Not Dangerous}} &\geq I  - D - Q   - C^-  Q   \\ 
      & \geq I - D  - ( C^- + 1 )  Q   
   \end{align*}
 Finally we need to account for the potential lost during the Big Hash Computation phase of each block. According Lemma~\ref{lem:total-block-phi}, the total potential lost during these phases is at most $3Q/2$.
 Adding the three parts together we get that, 
    \begin{align*}
        \Phi &= \Phi_{|\text{Dangerous}} +   \Phi_{|\text{Not Dangerous}} - \frac{3}{2}Q  \\
        & \geq I  -( C^- + 2C^\prime + 1  ) ( Q + \HC_s^d ) -  \left( C^- + 1  + \tfrac{3}{2}\right)  Q     \\
    \end{align*}
    Taking into account Lemma~\ref{lem:bound-corrupted-iterations}, which shows that $Q$ is at most $O(\epsilon d)$, 
    we can conclude that,
    \begin{equation*}
        \Phi \geq I - ( C^- + 2C^\prime + 1  ) \HC_s^d   - O( \epsilon d ),
    \end{equation*}
as desired.
\end{proof}

%% file: hash.tex
\begin{lemma}[Small hash collisions] \label{lemma:num-small-hash}

Fix a constant $C'$, and suppose that the constant $\Chash$ (chosen in \textsc{Initialization}, Algorithm~\ref{alg:init}) is sufficiently large compared to $C'$.
There is a constant $C$ so that the following holds.
Suppose that the total number of dangerous iterations during the execution of Algorithm~\ref{alg:adaptive} is $D$, and suppose that $D \geq C\cdot d \cdot \eps$.  Let $\HC_s^d$ be the number of dangerous iterations with small hash collisions.  Then 
\[ \Pr[ \HC^d_s \geq D/C' ] \leq \frac{2}{d^{C_\delta}},\]
where $C_\delta$ is the constant chosen in Line~\ref{line:randex}. 
\end{lemma}
\begin{proof}
    Let $\hat{D}$ be the number of dangerous iterations without corrupted randomness.  Notice that $\hat{D} \geq D - \eps d \geq (1-1/C)D$, as at most $\eps d$ iterations can have corrupted randomness and we are assuming $D \geq C\eps d$.  As dangerous hash collisions can only occur in dangerous iterations where Alice and Bob are using the same randomness, we restrict our attention to those.

    We recall from Algorithm~\ref{alg:init} (\textsc{Initialize}) that the small hash function is given by 
    \[ h_s(x, R_1, R_2) = h_2( h_1( x, R_1), R_2),\]
    where:
    \begin{itemize}
        \item $h_1$ has input length $t_1 = \log s + O(r \log d +{  \log ^ 2 d })$, output length $o_1 = 2\log(1/\eps)$, and seed length $\seed_1 = 2 \cdot t_1 \cdot o_1$ (as required for Theorem~\ref{thm:pseudorand}).  
        \item $h_2$ has input length $t_2 = 2 \log(1/\eps)$, output length $o_2 = \Chash$, for a constant $\Chash$ that we will specify later, and seed length $\seed_2 = O( \log\log(1/\eps))$ (as required for Theorem~\ref{thm:hash}).
    \end{itemize}
    We first analyze the number of hash collisions caused by $h_1$.

    \begin{claim}\label{cl:h1}
    The number of hash collisions in $h_1$ is at most $D/(2C')$ with probability at least $1 - 2d^{-C_\delta}$, where $C_\delta$ is the constant chosen in line~\ref{line:randex}. 
    \end{claim}
    \begin{proof}
        Since the seed for $h_1$ is shared at the beginning of a block, the adversary is able to tailor their corruptions to the seed for that block. Thus, our approach will be to argue that for any \emph{fixed} pattern of at most $\eps d$ corruptions over the $r \cdot \Itotal = \Theta(d)$ rounds of $\Pi'$, the conclusion holds with very high probability, and then union bound over all possible patterns of corruptions.

        With that in mind, fix a pattern of corruptions.  Fix a particular variable $\var$ that is hashed by the small hash function.  For an iteration $I$, let 
        \[ Z(I) = \mathbf{1}[h_1(\texttt{var}_A(I)) = h_1(\texttt{var}_B(I)] \ . \] 
        By Theorem~\ref{thm:pseudorand}, within a given block $B$, the random variables $\{Z(I)\}$ appearing in the dangerous iterations in that block (with uncorrupted randomness) are $\delta$-close to a collection of fully independent Bernoulli-$\eps^2$ random variables $W(I)$, using the fact that $2^{-o_1} = 2^{-2\log(1/\eps)} = \eps^2$.  Between the blocks, these random variables are independent. 

        If instead of the $Z(I)$, we consider the $W(I)$, we see that:
        \begin{align}\label{eq:bigC}
            \PR{\sum_I W(I) \geq D/C' } &\leq {D \choose D/(2C')} \eps^{2 D/(2C')} \notag \\
            &\leq (eC')^{D/(2C')} \eps^{2D/(2C')}\notag \\
            &\leq (eC'\eps)^{2Cd\eps / (2C')}\notag \\
            &= (eC'\eps)^{Cd\eps / C'}, 
        \end{align}
        where above we used the fact that for any $b \leq a$, ${a \choose b} \leq (ea/b)^b$.  

        Now we can union bound over all of the ways for the adversary to distribute the corruptions, as well as over the $O(1)$ variables that are hashed with the small hash function. 
        There are at most
        \begin{equation}\label{eq:countadv}
         {r \cdot \Itotal \choose \eps d} \leq {2d \choose \eps d} \leq \left(\frac{2e}{\eps}\right)^{d\eps}
         \end{equation} ways for the adversary to distribute corruptions.  Thus, by choosing the constant $C$ large enough relative to $C'$ in \eqref{eq:bigC}, the probability that there exists \emph{any} corruption pattern that causes $\sum_I W(I)$ to be larger than $D/(2C')$ is at most 
        \[6 \cdot \left(\frac{2e}{\eps}\right)^{d\eps} \cdot  (eC'\eps)^{Cd\eps / C'} \leq \eps^{\Omega(d\eps)}.\]

        Finally, we recall that the $W(I)$ are $\delta$-close to the $Z(I)$, for $\delta = 2^{-C_{\delta} \Iblock} = d^{-C_{\delta}}$.  This implies that 
        \[ \PR{\sum_I Z(I) \geq D/C' }\leq \eps^{\Omega(d\eps)} + d^{-C_{\delta}}.\] 
        Since $C_{\delta}$ is a constant, the second term is much larger than the first and dominates the expression.  This proves the claim. 
    \end{proof}

    Next, we analyze the collisions arising from $h_2$.  
    
    \begin{claim}\label{cl:h2}
    Provided that the constant $\Chash$ is sufficiently large (in terms of $C'$), the number of hash collisions in $h_2$ is at most $D/(2C')$ with probability at least $2^{-\Omega(d\eps)}$. 
    \end{claim}
    \begin{proof}
    Fortunately, as this randomness is shared right before it is used, the adversary cannot adapt their corruptions to influence the objects being hashed.  Thus, we may treat the values that $h_2$ is hashing as fixed in our analysis.

    As in the proof of Claim~\ref{cl:h1}, fix a variable $\texttt{var}$ to consider, and let 
       \[ Z(I) = \mathbf{1}[h_2(\texttt{var}_A(I)) = h_2(\texttt{var}_B(I)]. \] 
       Now the $Z(I)$ are fully independent Bernoulli-$p$ random variables, where $p = 2^{-\Chash/\Chashtwo}$, using Theorem~\ref{thm:hash} and the choice of $o_2 = \Chash$.  (Recall that $\Chashtwo$ is the constant in Theorem~\ref{thm:hash}.)  Thus,
       \[ \mathbb{E} \sum_I Z(I) = \hat{D} p \geq (1 - 1/C) D p. \] 
       Thus, if we choose $\Chash$ large enough (in terms of $C, C'$, and $\Chashtwo$) so that
       \[ (1 - 1/C) 2^{-\Chash/\Chashtwo} \leq \frac{1}{4C'},\]
       a Chernoff bound implies that 
       \begin{align*}
           \PR{\sum_I Z(I) \geq D/(2C')} &\leq \PR{ \sum_I Z(I) \geq 2 \EE\left( \sum_I Z(I) \right) } \\
           &\leq \exp(  - \Omega(D) )  \\
           &= \exp(-\Omega(d\eps)). 
       \end{align*}
       Finally, a union bound over  the $O(1)$ possibilities for $\texttt{var}$ establishes the claim.
    \end{proof}

Finally, if a small hash collision occurs, then a hash collision occurs in either $h_1$ or $h_2$, and furthermore, the number of dangerous iterations containing a small hash collision is at most the total number of small hash collisions during dangerous iterations. The lemma follows. 
\end{proof}

\begin{lemma}
    [Number of Dangerous iterations] \label{lem:number-dangerous-iteration} For any protocol $\Pi$, let $D$ be the total number of dangerous iterations. Then $D = O(\eps d )$ with probability at least $1 - \frac{2}{d^{C_\delta}}$. 
\end{lemma}
\begin{proof}

   Let $\alpha$ be a sufficiently large constant.  To prove the lemma, we will show that if $D \geq \alpha \epsilon d$, then the number of dangerous iterations with small hash collisions $\HC_s^d$ is larger than $\frac{D}{C'}$, where $C'$ is the constant from Lemma~\ref{lemma:num-small-hash}.  As Lemma~\ref{lemma:num-small-hash} shows that the probability of having so many hash collisions is at most $\frac{2}{d^{C_\delta}}$, the desired claim would follow. 

    So assume $D \geq \alpha \epsilon d$. By Lemma \ref{lem:phiub} and the fact that $q\leq 2 \eps d $ we get, $$\Phi \leq I - D + \cupper ( \eps d + \HC_s^d ).$$
    Next, by Lemma~\ref{lem:philb}, $$\Phi \geq I - (C^- + 2 C^\prime + 1 ) \HC_s^d - \Theta( \epsilon d ).$$ Combining these two inequalities gives:
    \begin{align*}
        I - (C^- + 2 C^\prime + 1 ) \HC_s^d - \Theta( \epsilon d ) &\leq  I - D + \cupper ( \eps d + \HC_s^d ) \\
      -(C^- + 2 C^\prime + 1  + \cupper ) \HC_s^d - \Theta (\eps d ) &\leq -D \\
       \frac{D - \Theta(\eps d) }{ C^- - 2C' + 1 + \cupper} &\leq \HC_s^d
    \end{align*} 
    By choosing $\alpha$ large enough, assuming $D \geq \alpha \eps d$ will imply that the numerator above is at least $D/2$, which implies that
    \begin{equation}\label{eq:conclude}
        \frac{D}{2 (C^- + 2 C^\prime + 1  + c ) } \leq \HC_s^d.
    \end{equation}
    Recall that Lemma~\ref{lemma:num-small-hash} says that for any $\alpha'$ (which controls the constant $\Chash$), there is some $C$ so that if $D \geq C d \eps$, then the probability that $\HC_s \geq D/\alpha'$ is small.\footnote{We note that in that lemma, $\alpha'$ is called $C'$; we avoid that notation here as we already have $C'$ defined in this scope.}
    Thus, we choose $\alpha' \geq 2(C^- + 2C' + 1 + c)$, and set $\Chash$ appropriately large; then we set $\alpha \geq C$ large enough that the lemma applies, and conclude that
    \[ \Pr[ \HC_s^d \geq D/\alpha' ] \leq \frac{2}{d^{C_\delta}}.\]
    But this implies that the probability that \eqref{eq:conclude} occurs is small.  Since \eqref{eq:conclude} follows from the assumption the $D \geq \alpha \eps d$, this in turn implies that
    \[ \Pr[ D \geq \alpha \eps d ] \leq \frac{2}{d^{C_\delta}},\]
    completing the proof.

\end{proof} 
\begin{corollary}\label{corr:order-small-hash-coll}
    The number of dangerous small hash collisions, $\HC^d_s$, satisfies $\HC^d_s = O(\eps d)$ with probability at least $1 - \frac{2}{d^{C_\delta}}.$
    \end{corollary}
    \begin{proof}
        This follows immediately from Lemma~\ref{lem:number-dangerous-iteration} and the fact that the number of dangerous small hash collisions---which by definition (Definition~\ref{def:dangerous-small-hash-coll}) can only occur during dangerous iterations---is bounded above by a constant (namely, the number of variables hashed with the small hash function) times the number of dangerous iterations.
    \end{proof}

\begin{corollary} [$k$ is not too large during a sneaky attack] Let $I$ be an iteration during the voting window of a sneaky attack $\mathcal{S}$. Then $k_A(I) = k_B(I) = k $ with $k = O(\eps d)$ with probability at least $1-\frac{2}{d^{C_\delta}}$.
    
\end{corollary}

\begin{proof}
    From Lemma~\ref{lem:divwindow-corr-coll} we have $ k \leq \alpha (Q_{\mathcal{S}}  + H_{\mathcal{S}} ) $, where we recall that $H_{\cS}$ is the number dangerous iterations having small hash collisions during the sneaky attack $\cS$.  Since $H_{\mathcal{S}} \leq \HC_s^d$, $Q \leq 2\eps  d $ and further from Lemma~\ref{lem:number-dangerous-iteration} we get that $\HC_s^d$ is at most $O ( \eps d ) $ with probability at least $1-\frac{2}{d^{C_\delta}}$, from which it follows that $k \leq \alpha (\eps d + O (\eps d ) ) = O ( \eps d )$, proving the statement.    

\end{proof}
\begin{corollary}
    [The total rewind during all sneaky attacks is bounded.] 
    Following the notation of Lemma~\ref{lem:total-rewind-upb}, define $\ell^\ast$ as the total rewind length of all sneaky attacks during the execution of Algorithm~\ref{alg:adaptive}. Then $\ell^\ast \leq O(\eps d)$ with probability at least $1-\frac{2}{d^{C_\delta}}$.
\end{corollary}
\begin{proof}
    By Lemma~\ref{lem:total-rewind-upb}, $\ell^\ast \leq 4c(Q+\HC_s^d)$, where $q$ is the total number of corruptions during the protocol execution, $\HC_s^d$ is the total number of small hash collisions, and $c$ is some constant. We now plug in the upper bound $ Q\leq \eps d$ (which is the adversary's corruption budget) and the upper bound $\HC_s^d\leq O(\eps d)$ (which holds with probability $\geq 1-\frac{2}{d^{C_\delta}}$ by Corollary~\ref{corr:order-small-hash-coll}) to derive the corollary. 
\end{proof}

Having now shown that the number of dangerous small hash collisions will be small (i.e., $O(\eps d)$) with high probability, we now turn to bounding the number of big hash collisions. The following lemma argues that we will in fact have \emph{zero} big hash collisions with high probability.

\begin{lemma}[Big hash collisions]\label{lem:bighash}
With probability at least $1 - d^{-10}$, no big hash collisions occur during the entire protocol.
\end{lemma}
\begin{proof}

The seed for the hash function $h^b$ is shared immediately before the big hash computation phase; thus, similarly to the proof of Claim~\ref{cl:h2}, the adversary's corruptions that influence the objects being hashed do not depend on this randomness. As a result, we may treat the values that are hashed by $h^b$ as being fixed in our analysis. 

Consider a block $B$, and let $x_{\megastate}$ be the inputs to the big hash function $h^b$ in line~\ref{line:big-hash-comp} for the $\megastate$ simulated within this block. Further, define 
\begin{align*}
Z(B) = \mathbf{1}[ \exists \megastate_A, \megastate_B \text{ s.t. } &h^b(x_{\megastate_A} , R_A(B) ) = h^b ( x_{\megastate_B}, R_B(B))  \\&\text{ and } \megastate_A \neq \megastate_B, \\&\text{ and } \piter_A, \piter_B \in B ],
\end{align*}
where $R_A(B), R_B(B)$ is the randomness used by Alice and Bob in the hash. (Note that here, $B$ in the subscript refers to Bob, while $B$ elsewhere is the block). 
That is, $Z(B)$ is one if there are two megastates $\megastate_A$ and $\megastate_B$ that Alice and Bob simulated in the block $B$ that are not the same but so that the corresponding big hashes match.

Notice that as we are counting the number of big hash collisions we may restrict our attention to those iterations where Alice and Bob have the same randomness; thus we assume that $R_A(B) = R_B(B) = R(B)$. Then, 
\begin{equation}\label{eq:ZB}
\Pr[Z(B) = 1 ] \leq \log^2 d \times \frac{1}{d^{C_b/\Chashtwo}},
\end{equation}
where $\Chashtwo$ is the constant in Theorem~\ref{thm:hash}.
Indeed, by Theorem~\ref{thm:hash}, and our choice of $o_3 = \obighash$, the probability of a hash collision in $h^b$ is at most $2^{-o_3/\Chashtwo} = d^{-C_b/\Chashtwo}.$  Then union bounding over the at most $\log(d)$ mega-states per block for each party yields \eqref{eq:ZB}.

Finally, we union bound over all $\Btotal$ blocks.  We have
\[ \Btotal = \left\lceil \Itotal / \Iblock \right\rceil
= \left\lceil \frac{ \lceil d/r \rceil + \Theta(d\eps) \rceil }{ \lceil \log d \rceil } \right \rceil \leq \frac{2d}{\log d} 
\]
for sufficiently small $\eps$.  Thus, the probability that a big hash collision occurs at any point in the simulation is at most
\[ \frac{2d}{\log d} \times \log^2 d \times \frac{1}{d^{C_b/\Chashtwo}} \leq d^{-10}\]
by choosing $C_b$ sufficiently large in terms of $\Chashtwo,$ and for sufficiently large $d$.

\end{proof}
Recall that Assumption~\ref{asm:rand-big-neq} asserted that for any two distinct blocks $B$ and $B'$, the randomness $\Rblock^b(B)$ used in $B$ and $\Rblock^b(B')$ used in $B'$ were distinct.  Now we remove this assumption by showing that it holds with high probability.
\begin{lemma}\label{lem:big-rand-uniq}
    Let $B$ and $B'$ be two distinct blocks during the execution of Algorithm~\ref{alg:adaptive}. Let $(\Rblock^b)_A(B)$ and $(\Rblock^b)_B(B')$ denote the randomness used by Alice during block $B$ and by Bob during block $B'$ respectively. Then the probability that $(\Rblock^b)_A(B) =(\Rblock^b)_B(B') $ is at most $d^{-13}$.
\end{lemma}
\begin{proof}
    Each party creates their copy of the seed $\Rblock^b$ by concatenating two smaller seeds, $\Rblock^{b,1}$ generated by Alice in Line~\ref{line:alice-send-bigrand}, and $\Rblock^{b,2}$ generated by Bob in Line~\ref{line:bob-send-bigrand}.  As a result, for Alice the first half of $(\Rblock^b)_A$ is uncorrupted randomness; similarly for Bob the second half of $(\Rblock^b)_B$ is uncorrupted randomness.
    
     Assume without loss of generality that $B' > B$. As the randomness for each block is generated independently, we may treat $(\Rblock^b)_A(B)$ as an arbitrary fixed string of length $2\ell$ relative to the randomness generated during block $B'$.  As noted above, the second half of $(\Rblock^b)_B$ is uncorrupted randomness of length $\ell$.
    \begin{align*}
          \Pr[ (\Rblock^b)_A(B) = (\Rblock^b)_B(B')] &\leq \Pr[(\Rblock^{b,2})_A(B) = (\Rblock^{b,2})_B(B')] \\
          &= 2^{-\ell}.     
    \end{align*}
    Recall that we set $\ell = \sd_3 = \Chashtwo(C_b \log d + \log 1/\eps )$ (in line~\ref{line:big-h-par} of Algorithm~\ref{alg:init}).  Thus, by choosing $C_b$ sufficiently large in terms of $\Chashtwo$, we have that \[\Pr[ (\Rblock^b)_A(B) = (\Rblock^b)_B(B')] \leq d^{-13},\] proving our claim.  
\end{proof}

\begin{lemma}[Unique Big Hash Randomness]\label{lem:all-big-rand-uniq}
    During the running time of Algorithm~\ref{alg:adaptive}, the probability that there exist two distinct blocks $B$ and $B'$ such that $(\Rblock^b)_A(B) =(\Rblock^b)_B(B') $ is at most $\frac{1}{d^{10}}$.
    
\end{lemma}
\begin{proof}
 We will prove this statement by taking a union bound over all possible pairs of $B$ and $B'$.  According to Lemma~\ref{lem:big-rand-uniq}, for any two distinct blocks $B$ and $B'$, 
    $$ P[ (\Rblock^b)_A(B) = (\Rblock^b)_B(B')] \leq d^{-13}.$$
    Note that $\Btotal = \lceil \Itotal/\Iblock \rceil \leq  d + \theta (\eps d) \leq 2d $ for large enough $d$, where the first equality is how we chose $\Btotal$ in Line~\ref{line:Btot} of Algorithm~\ref{alg:init}.  Hence, the total number of possible pairs $(B,B')$ of distinct blocks is at most $\Btotal^2\leq 4 d^2$.  Thus the probability of finding such pair is at most $ 4d^2 \times d^{-13} \leq d^{-10}$ for large enough $d$, proving our claim.
\end{proof}

%% file: rest_of_analysis.tex
In this section, we prove Theorem~\ref{thm:main}.  
We restate the theorem for the reader's convenience.  

\mainThm*

    \begin{proof}
        We first recall a few parameters from the protocol, for the reader's convenience:
                    \begin{equation}
                \Iblock  = \lceil \log d \rceil \ , \ \Itotal = \lceil d/r \rceil + \Theta ( d\epsilon ) \ , \ \Btotal  = \lceil \Itotal/ \Iblock \rceil .
            \end{equation}
        We prove the theorem in the following four steps.  We first outline all four, and then go through each of them.
        \begin{enumerate}
       
            \item First, we establish that the simulation is correct, assuming no big hash collisions occur, and that the number of dangerous iterations is $O(\eps d)$.  To establish this, we use the lower bound on the potential function (Lemma~\ref{lem:philb}), we show that if no big hash collision occurs during the simulation and the number of dangerous iterations throughout the simulation is small then the potential function grows.  By tying this potential increase to the length of the correct simulation using the upper bound on the potential function (Lemma~\ref{lem:phiub}), we will show that $\ell^+ \geq d/r$, which guarantees that by the end of $\Itotal$ iterations the algorithm has correctly simulated $\Pi$.
            
            \item We next bound the probability of failure, which by the above is bounded by the probability that there were any big hash collisions, or that there were too many dangerous iterations. 

            \item Next, we work out the parameters to analyze the memory usage of $\Pi'$.

            \item Next, we work out the parameters to analyze the rate. 

            \item Finally, we bound the running time of $\Pi'$. 
        \end{enumerate}
        We go through each of these below.
        \begin{enumerate}

        \item  \textbf{Correctness of simulation.}

        The upper bound on potential function in Lemma \ref{lem:phiub} tells us that, in any iteration $I$, and in particular iteration $I = \Itotal$, 
    \[ \Phi(\Itotal) \leq \ell^+(\Itotal) +  \cupper  ( Q + \HC_s^d ),  \]
         assuming that no big hash collision occurs during the entire protocol.  Next, we plug in bounds for $Q$ and $\HC_s^d$.
         Assuming that the number of dangerous iterations is $O(\eps d)$, Corollary~\ref{corr:order-small-hash-coll} states that  $\HC_s^d$, the number of dangerous small hash collisions, is at most $O( \eps d )$. Since $Q$, the number of corruptions introduced throughout $\Pi$', is at most $2 \eps d$, we conclude that 
        \[ \Phi(\Itotal) \leq \ell^+(\Itotal) +  \cupper  ( \eps d + \Theta(\eps d  ))    = \ell^+(I) + \tilde{\cupper} \eps d , \] 
        for some constant $\tilde{\cupper}.$
        Now we use the lower bound on the potential  from Lemma~\ref{lem:philb}.
        Lemma~\ref{lem:philb} implies that, at the end of the algorithm, 
    \[ \Phi(\Itotal) \geq \Itotal - (C^- + 2 C^\prime + 1 )  \HC_s^d - \Theta( \epsilon d ),\] 
    and using 
    $ q \leq 2 \eps d $ and $\HC_s^d = \Theta ( \eps d ) $, we obtain
    \begin{align*}
            \Phi(\Itotal) &\geq \Itotal - (C^- + 2 C^\prime + 1 )  \Theta(\eps d ) - \Theta( \epsilon d ) \geq \Itotal - \clower \eps d 
    \end{align*}
for some constant $\clower$.

    Comparing the two, we see that
    \begin{align*}
        \Itotal -  ( \clower + \tilde{\cupper} ) \epsilon d \leq \ell^+(\Itotal). 
    \end{align*}
    Thus, choosing the total number of iterations large enough, so that $$\Itotal \geq \lceil d/r \rceil + 2(\clower + \tilde{\cupper}) \eps d,$$ we have
    \[ \ell^+(\Itotal) \geq \lceil d/r \rceil.\]
      By definition, this means that when Algorithm~\ref{alg:adaptive} terminates at iteration $\Itotal$, the length of the correct simulated path is at least $\lceil d/r \rceil$, which is in fact the entire (iteration) length of $\Pi$ (recall Remark~\ref{rem:padding-length}).  We conclude that Alice and Bob have correctly simulated all of $\Pi$, as desired.

        \item \textbf{Probability of failure.} 
Item 1 shows that if there are no big hash collisions in the entire protocol and there are $O(\eps d)$ dangerous iterations, Alice and Bob can correctly simulate the protocol $\Pi$.

For the protocol to fail, there are three possibilities.  First, it could be that the Assumption~\ref{asm:rand-big-neq} does not hold, which, according to Lemma~\ref{lem:all-big-rand-uniq}, happens with probability at most $d^{-10}$. Second, there could be a big hash collision during the protocol, which, according to Lemma~\ref{lem:bighash}, happens with probability at most $1/d^{10}$. Third, the protocol must have more than $\Theta(\eps d)$ dangerous iterations, which, according to Lemma~\ref{lem:number-dangerous-iteration}, occurs with probability at most $2/d^{C_\delta}$. Here, $C_\delta > 1$ is an arbitrary constant from Line~\ref{line:randex} of Algorithm~\ref{alg:adaptive}. If we choose it to be, say, $C_\delta = 2$,  the total probability of failure is $2/d^{10} + 2/d^{C_\delta}\leq 1/\poly(d)$.
     
        \item \textbf{Memory}. Next we compute the memory that each party uses during Algorithm~\ref{alg:adaptive}.
        Each party stores: 
            \begin{itemize}
                \item A set $\M$ with size at most $O(\log d)$ containing available meeting points. Each point is an integer in the range between $[0,d]$, thus requires at most $\log(d)$ bits for representation. In total, $\M$ will need $\Theta(\log^2 d)$ bits in memory.
                \item A memory $\Mem$, containing mega-states corresponding to each $p\in \M$. Each mega state contains:
                \begin{itemize}
                    \item $\pv$: A state with $\log(s)$ bits
                    \item $\phash$: A hash function with $o_3 = O(C_b \log d)$ bits
                    \item $\pseed$: A seed for $h^b$ with $ \sd_3 = O( \log d + \log 1/\eps)$ bits.  
                    \item $\piter$ and $\pdepth$: integers of size at most $d$, and hence which take at most $\log d$ bits.
                    \item $\pt$: We recall that $\pt$ is a placeholder variable, referencing the partial simulated path $\Bsim$, and not actually stored.  Thus, we do not need to account for any space it takes up.
                \end{itemize}
                Thus, each megastate uses $\log(s) +  O( \log d)$ bits, and all $O(\log d)$ of them together require
                $O( \log(s) \log(d) + \log^2 d)$ bits. 

\item 
                The memory to store the single copy of $\Bsim$ that each $\pt$ points to; this is $O(r \log d + \log^2 d)$ bits.
      
                \item The parameters and state variables for the current iteration/block: $k , E , j ,v , \Icount , \Icnt$.  All of these except $v$ require $O(\log d)$ bits, and $v$ requires $\log s$ bits.
                \item The hash values output and received from both the big and small hash function.  Outputs of the small hash function have size $o_2 = O(1)$, and Alice and Bob must also compute the intermediate values of $h_1$, which have size $O(\log (1/\eps))$. 
                The big hash function has outputs of size $o_3 = O(C_b \log d)$.
                \item The randomness generated and exchanged.  For $h_2$, the randomness $\Riter$ is generated in each iteration, exchanged, and then forgotten; this is $\sd_2 = O(\log\log(1/\eps))$ bits to store.

                The randomness $\Rblock^s$ is a bit trickier, and for this we need to deviate a bit from the stated pseudocode, as per the comment after Line~\ref{line:randex}.    

                Instead of storing all of $\Rblock^s$, Alice and Bob will instead store only $\Rtblock^s$, and generate the blocks $\Rblock^s[I]$ that they need on the fly in each iteration $I$.  To verify that we can do this, we need to open up the pseudorandom \texttt{extend} function of Theorem~\ref{thm:pseudorand}, from~\cite{NN93}.  The way that this works is as follows.  Alice and Bob agree on a suitable matrix $G \in \mathbb{F}_2^{2^{\ell'} \times \ell}$, which is the generator matrix for a $\delta$-balanced error correcting code of rate $\poly(\delta)$.  Then the \texttt{extend} function uses the seed of length $\ell'$ to select a row of this matrix, and then outputs it.  In particular, if the generator matrix is explicit, Alice and Bob can generate any entry of it---and in particular each bit in the short chunk needed for $\Rblock[I]$---on the fly. There exist such explicit matrices: for example, an AG code concatenated with a Hadamard code will work (see page 2 of~\cite{ben2009constructing}). Thus, Alice and Bob need only store $\Rtblock^s$.
                
                  For $\Rtblock^s$, the length of this is the value of $\ell' = O(\log \ell + \log(1/\delta))$ in Theorem~\ref{thm:pseudorand} used in the construction of the small hash.
                  We recall that  $\ell \gets \Iblock \sd_1$ and $\delta \gets 1/\poly(d)$, where
                       \begin{equation} \sd_1 = 2 o_1 t_1 = O( \log(1/\eps) \log (s) + r \log(1/\eps) \log(d) +\log (1/\eps ) \log^2 d ). \notag\end{equation}
                Thus, $\ell'$ is given by
     \begin{align*}
                    \ell' = O( \log(\ell ) + \log (1/ \delta ) )& =  O( \log d + \log \log s + \log r + \log \log d ) \\&= O( \log d + \log \log s + \log r  )\\ &= O(\log d + \log\log s) = O(\log d),
         \end{align*}
 where we have used that $r = \poly(1/\eps)$ and $\eps > 1/d$ (or else the adversary cannot introduce any corruptions); and that $s \leq 2^d$.
                
    Alice and Bob also need the seed chunks $\Rblock^s[I]$  during each iteration; however, as discussed above, they can generate these one bit at a time on the fly, and so we do not need to account for any extra memory usage.
                
                Finally, during the Big Hash Computation phase at the end of each block, Alice and Bob exchange and temporarily store $\Rblock^b$, which contains 
                    \[ \sd_3 = O(\log(d) + \log (1/\eps )) = O(\log d)\]
                    bits by Claim~\ref{claim:bighashokay}.

            \end{itemize}
            Adding up all the requirements, and recalling that $r = O\left( \sqrt{ \frac{\log\log(1/\eps)}{\eps}}\right)$, the dominating terms are 
\begin{align*}
    O\inparen{ \log(s) \log(d) + \log^2(d) + r \log(d)} = O\inparen{\log(s) \log(d) + \sqrt{ \frac{\log\log(1/\eps)}{\eps}}  \log d }
\end{align*}

            Above, we used that $d \leq s$. Thus, when $\eps$ is constant this is $O(\log(d) \log(s))$, as desired.

        \item \textbf{Rate.}
        In the simulation of the Algorithm \ref{alg:adaptive} we have $\Btotal$ blocks such that $$\Itotal = \Btotal \Iblock = \lceil d/r \rceil + \Theta (\eps d ).$$ During each block, $O( \Iblock + \log(1/\eps) ) $ bits of communication are used in order to exchange randomness for the small hash at the start of the block.
Indeed, the length of seed transmitted is the parameter $\ell'$ in Algorithm~\ref{alg:PRand-exchange}.  We have 
\[ \ell' = O( \log(1/\delta) + \log(\ell) ),\]
where $\delta = 2^{-C_\delta \Iblock}$ and 
\[ \ell = \mathrm{sd}_1 \cdot \Iblock = 2(\log s + O(r \log d)) \cdot 2\log(1/\eps) \cdot \Iblock,\]
plugging in the parameters for $\sd_1$ from Algorithm~\ref{alg:init}.
Thus, the dominant terms in $\ell'$ are $\log(1/\delta) = C_\delta \Iblock = O(\Iblock)$ and the $O(\log\log(s))$ and $O(\log r) = O(\log(1/\eps))$ terms that appear in $\log(\ell)$.  Since $O(\log\log(s)) = O(\log d) = O(\Iblock)$, we see that $\ell' = O(\Iblock + \log(1/\eps))$ as claimed.  Now, the number of bits transmitted is $O(\ell')$, as this seed is encoded with a constant-rate error correcting code.

Next, we account for the randomness exchange for the big hash at the end of the block. From the proof of Claim~\ref{claim:bighashokay}, we know that the seed length for $h^b$ is $\sd_3 = O(C_b \log d + \log(1/\eps)) = O(\Iblock + \log(1/\eps))$.  This is communicated after being encoded with a constant rate error correcting code, so the total number of bits communicated is still $O(\Iblock+ \log(1/\eps))$.

        Other than that, in each iteration, we have $r$ bits of communication to simulate the $r$ rounds of $\Pi$; and we have the overhead from exchanging randomness for the small hash, and to exchange the hashes themselves.  The output of each small hash has length $\Chash$, which is constant, so this is dominated by $\sd_2$, the length of the seed for $h_2$, which is $O(\log\log(1/\eps))$.  Let $r_c$ be the amount of extra communication in each iteration, so $r_c = O(\log\log(1/\eps))$, and the total amount of communication in each iteration (not including the randomness exchange at the beginning of each block) is at most $r + r_c$ bits.
        
         Thus, for each block, there are a total of $O(\Iblock + \log(1/\eps)) + \Iblock (r + r_c) $ bits of communication. As a result, in the whole algorithm, the total amount of communication, in bits, is at most
        \begin{align*}
            &\Btotal\cdot ( O(\Iblock + \log(1/\eps)) + \Iblock (r + r_c)  ) \\
            &\qquad= O(\Itotal) + O(\Btotal \log(1/\eps)) + \Itotal (r+r_c).
            \end{align*}
            Now, we have (ignoring floors and ceilings for clarity)
            \begin{align*}
            \Btotal &= \frac{ d/r  + O(d\eps) }{\log(d)} \\
            &= \frac{ d/r }{\log d} + O(d \eps / \log d) \\
            &= \frac{ d \sqrt{\eps/\log\log(1/\eps)} }{\log d} + O\left( \frac{d \eps}{\log d}\right).
            \end{align*}
            Returning to the above, the term $O(\Btotal \log(1/\eps))$ is thus $O(d\sqrt{\eps})$, using the fact that $\eps \geq 1/d$ (or else there would be no errors) and hence $\log(d) \geq \log(1/\eps)$.
            
            Thus, the total amount of communication is
            \begin{align*}
             O(\Itotal) + O(d\sqrt{\eps}) +\Itotal(r + r_c) 
            &= O(d\sqrt{\eps}) + \left( \lceil d/r \rceil + O(d\eps) \right)( O(1) + r + r_c) \\
            &= O(d\sqrt{\eps}) + \left(  d + O(d\eps r) \right) \cdot \inparen{ O(1/r) + 1 + \frac{r_c}{r}} \\
            &= d \inbrak{ O(\sqrt{\eps}) + (1 + \eps r) \cdot\inparen{1 + \frac{r_c}{r} + O(1/r)}}.
            \end{align*}
            By plugging in  $r = \lceil\sqrt{\frac{r_c}{\eps}} \rceil$, and  $r_c = \Theta(\log \log \frac{1}{\eps})$, we see that the total amount of communication is
            \[d \left( 1 + O\left(\sqrt{\eps \log \log \frac{1}{\eps }}\right)\right).\]

        \item \textbf{Time.} We compute the time overheads of $\Pi'$ relative to $\Pi$. The block-level pseudorandomness exchanges have $\poly(\log d, 1/\eps)$ output length, and therefore yield a $\poly(\log d, 1/\eps)$ overhead. For the hash function computations, again the logarithmic lengths on the inputs and outputs imply a $\poly(\log d,\log s, 1/\eps)$ bound on the overhead. Maintaining the available meeting points sets also incur a $\poly(\log d\cdot\log s)$ overhead.
        The steps that just require reading and writing from memory take time at most $\poly(\log s, \log d)$.  Thus, per iteration, the total amount of overhead is $\poly(\log s, \log d, 1/\eps)$.  The number of iterations is $\Itotal = O(d)$, so the total contribution of the above are $d \cdot \poly(\log d, \log s, 1/\eps)$. 
        
        Finally, there is simulating the original protocol $\Pi$ itself.  The number of extraneous iterations of $\Pi$ that are simulated is $\Theta(\eps d)$, with $r$ rounds of communication in each. If we conservatively imagine that all of the time it takes to execute $\Pi$ is contained in these extraneous simulations, it incurs a multiplicative overhead of $\poly(d).$  Thus, the total running time of $\Pi'$ is at most the running time $T$ of $\Pi$, times a $\poly(d/\eps)$ factor.

        Altogether, the total running time is $d \cdot \poly( \log d, \log s, 1/\eps) + T\cdot \poly(d/\eps) = T\cdot \poly(d/\eps)$, as claimed. Above, we have used the fact that $T \geq d$, because the running time of $\Pi$ is at least its depth.
        
        \end{enumerate}

        This completes the proof of all five items, and thus of our main theorem.
    \end{proof}

%% file: maintechrestate.tex
We restate Lemma~\ref{lem:maintech} for the reader's convenience:

\mainTechLem*

\begin{proof}

\input{v2-proof-tech.tex}
\end{proof}

%% file: v2-proof-tech.tex
Let $\Aiter = \Aiter(\tsp)$ and $\Biter = \Biter(\tsp)$. {Throughout the proof we assume that $t_b$ is not during the Big Hash Computation phase at the end of a block.  (That is, we assume that Item 2 does not hold).  Thus, our goal is to show that, under this assumption, either Item 1 or Item 3 holds.}  All claims are considered for time $\tsp$ unless stated otherwise.
\begin{claim}\label{claim:bvc-or-mp12} 
  At time $\tsp$,  either  $p \in  \MPonetwo_A (j_p, \Aiter) \cup \MPonetwo_B(j_p, \Biter)$  or $\BVC_{AB} \geq \constbvc k $.
\end{claim}
\begin{proof}
Throughout, let
\[ \constbvc = \frac{0.6}{4}.\]
    Assume that $p\notin \MPonetwo_A (j_p, \Aiter) \cup \MPonetwo_B(j_p, \Biter)$. Then the fact that $p$ is the next available jumpable point implies that 
    \begin{equation}\label{eq:assumeTWOC}
    p = \MPthree_A ( j_p, \Aiter) = \MPthree_B (j_p,\Biter).
    \end{equation}
    We first claim that 
    \begin{equation}\label{eq:claim1}
    p \in \MPall(\Aiter, j_p -1) \cap \MPall(\Biter, j_p -1).
    \end{equation}
    To see \eqref{eq:claim1}, notice that for 
    \eqref{eq:assumeTWOC} to hold,
    it must be that $p \leq \Aiter, \Biter < p + 2^{j_p}$. 
    Now consider $\Aiter$; there are two options, either $\Aiter < p + 2^{j_p - 1}$ or $\Aiter \geq p + 2^{j_p - 1}$.  In the first case, we have $\MPthree(j_p-1, \Aiter)= p$. In the second case, as $p+2^{j_p-1} \leq \Aiter \leq p + 2^{j_p}$, then $\MPtwo(j_p-1,\Aiter) = \MPJ{\Aiter}{j_p-1} = p$. 
    In either case, $p$ is a transition candidate for Alice at scale $j_p - 1$, that is $p \in \MPall(\Aiter, j_p-1)$.
    The same logic also applies to Bob, so $p \in \MPall(\Biter, j_p - 1)$.  This establishes \eqref{eq:claim1}. 

    Now, we finish the proof of Claim~\ref{claim:bvc-or-mp12}.
  If $ j_p-1 \geq \jsp $, then \eqref{eq:claim1} implies that $j_p -1 $ is also a  jumpable scale; but this contradicts the fact that $j_p$ is the next jumpable scale.   Thus we have that $ j_p = \jsp$. This means that during iterations of scale  $\jsp-1= j_p -1 $, $p$ was an available transition candidate for both parties. Given that Alice and Bob have skipped this scale, then according to  Lemma~\ref{lem:skipsCostAdv} the adversary must have introduced corruptions to cause $\BVC_{AB}$ to increase at least $(0.6) 2^{\jsp-1}$ times during the iterations corresponding to scale $2^{\jsp-1}$.  As $k$ is at most $2^{\jsp+1}$, we have that $ \BVC_{AB} \geq (0.6) 2^{-2} 2^{\jsp+1}\geq \constbvc k $, where $\constbvc= (0.6) 2^{-2}$. This proves the claim.
\end{proof}
For the remainder this proof, assume without loss of generality that $\Aiter \geq \Biter$.
Then we have the following claim.

\begin{claim}\label{claim:only-MP1}
    Either $\MPone_A(j_p, \Aiter) =p$ or $\BVC_{AB} \geq \constbvc k$.
\end{claim}

\begin{proof}
Assume that $\BVC_{AB}< \constbvc k$, so we want to show that $p = \MPone_A(j_p,\Aiter)$.  From Claim \ref{claim:bvc-or-mp12} we know that 
    \begin{equation*}
    p\in \MPonetwo_A (j_p, \Aiter) \cup \MPonetwo_B(j_p, \Biter).
    \end{equation*}
    Further, because we are assuming without loss of generality that $\Aiter \geq \Biter$, in fact we have that $p\in \MPonetwo(j_p,\Aiter)$: Indeed, if $\MPthree(j_p, \Aiter) =p$ and $\Biter \leq \Aiter$ then $\MPthree(j_p, \Biter) = p$, which results in the same contradiction as in proof of Claim \ref{claim:bvc-or-mp12}. 
    Thus, to show that $p = \MPone(j_p, \Aiter)$, it suffices to show that $p \neq \MPtwo_A(j_p, \Aiter)$.
    
    Assume towards a contradiction that $\MPtwo_A(j_p, \Aiter) = p$.
    Now we will show that under this assumption, then $BVC_{AB}\geq \constbvc k$, a contradiction; to see this we will show that in this case, $p$ is an available transition candidate on scale $j_p -1$, which means that the adversary had to invest enough corruptions for Alice and Bob to have skipped it. 
    
    To that end, first observe that 
    \[ p = \MPtwo_A(j_p , \Aiter) = \MPone_A ( j_p -1,\Aiter), \]
    where the first equality is by our assumption and the second always holds by definition.
Further, the assumption that $ \MPtwo(j_p,\Aiter) = p$, along with the fact that $\Aiter \geq \Biter$, implies that $p\in \MP^{\{2,3\}}(j_p, \Biter)$. In more detail, 
because $p$ is a transition candidate, it must be one of $\MPone, \MPtwo$ or $\MPthree$ for Bob, so we just need to show that $p \neq \MPone(j_p, \Biter)$.  To see this, suppose towards a contradiction $p = \MPone(j_p, \Biter)$; this implies that $|\Biter - p| > 2^{j_p +1}$.  However, since $p = \MPtwo(j_p, \Aiter)$, we have $|\Aiter - p| < 2^{j_p + 1}$.  Given that $p \leq \Biter \leq \Aiter$, this is a contradiction.

Thus, either $p$ is $\MPtwo$ or $\MPthree$ for Bob at scale $j_p$.  We treat each scenario separately.
    \begin{enumerate}
        \item $ \MPtwo_B(j_p, \Biter) = p$. In this case, $p$ is a scale $j_p$ meeting point for both Alice and Bob. As a result, by the definition of $\MPone$, $\MPone(j_p -1 ,\Aiter) = \MPone( j_p -1,\Biter) = p$, meaning that $p$ is an available transition candidate for scale $j_p -1$. 
        \item $\MPthree_B (j_p, \Biter)= p.$ In this case, $ \Biter \in [p , p + 2^{j_p})$. Now if $ \Biter \in [ p + 2^{j_p -1 } , p + 2^{j_p}) $ then $ \MPtwo_B (j_p-1, \Biter)= p$; otherwise, $ \MPthree_B(j_p-1, \Biter) = p$.  In either case, again $p$ is an available transition candidate in scale $j_p -1$.
    \end{enumerate}
    This shows that, if $\MPtwo_A(j_p, \Aiter) = p$, then $p$ is an available transition candidate at scale $j_p -1$.
 Now, if $ j_p > \jsp$, $j_p -1$ should be the next available jumpable scale, but this contradicts the lemma statement. As a result we conclude that $ j_p  = \jsp$. In this case, the adversary has invested in many corruptions during scale $\jsp -1 = j_p -1$, in order to skip the point $p$ in this scale. Thus, according to Lemma~\ref{lem:skipsCostAdv}, $\BVC_{AB} \geq (0.6) 2^{\jsp -1}  \geq \frac{0.6}{4} k$. This is a contradiction with the assumption that $\BVC < \constbvc k$. As a result, $ \MPtwo(j_p,\Aiter)\neq p$. Knowing that $p\in \MPonetwo(j_p, \Aiter)$ we can conclude that $\MPone(j_p,\Aiter) = p$, as desired.
\end{proof}

Given Claim \ref{claim:bvc-or-mp12} and Claim \ref{claim:only-MP1}, from now on we assume that $\MPone_A = p$, and in particular that 
\begin{equation}\label{eq:jp_is_wminus1}
j_p = w-1. 
\end{equation} 
As a result, from now on we can also assume that
\begin{equation}\label{eq:starw}
L^-(\tsp) < 2^{w-\csneak} \leq 2^{w-3}.
\end{equation}

We will show that either $\BVC_{AB} \geq \constbvc k$ or a \protosneaky is in progress.

 Let $L^- = L^-(\tsp)$.   Now, as $p$ is the scale $w$ meeting point for Alice, we may write $p = (\alpha -1 ) 2^{w}$ for some integer $\alpha$. Let $\phat = p + 2^w$. Let $b$ be the divergent point at time $\tsp$. Notice that this  exists as we assume that $\ell^-(\tsp) > 0$.
 Let $D = b - p$.  Notice that $D > 0$, because
 \begin{equation}\label{eq:Dgeq0}
 2^{w-3} > L^- \geq \Aiter - b \geq \phat - b, 
 \end{equation}
 where above we have first used \eqref{eq:starw}; then the definition of $L^-$; then the fact that 
 \begin{equation}\label{eq:ellA_small}
 \Aiter < \phat = p + 2^w,
 \end{equation} 
 which follows since $p$ is a scale-$w$ MP for Alice.
 This implies that
 \begin{equation}\label{eq:b_big}
  b > \hat{p} - 2^{w-3} > p.
  \end{equation}

Define $q=p+2^{w-1} = \phat - 2^{w-1}$. Notice that $q < b$, since otherwise we have that
\[ 2^{w-3} > L^- \geq \Aiter - b \geq \hat{p} - q = 2^{w-1}, \]
which is not true.

\paragraph{Organization.}  We break the rest of the proof up into \textbf{Item 1} and \textbf{Item 2}, depending on whether $b < \phat$ or $b \geq \phat$. as follows:
\begin{enumerate}
    \item[\textbf{Item 1.}]
If $b < \phat$, first we will show that $q$ is  a jumpable point. Then we will show that if $q\in \M_A \cap \M_B$ or $q\notin \M_A \cup \M_B$, then $\BVC_{AB}\geq \constbvc k$. Otherwise, we will show that a sneaky attack is in progress. 
\item[\textbf{Item 2.}] if $ b \geq \phat$, or if $\ell^-(\tsp) = 0$, then we will show that $\BVC_{AB}\geq \constbvc k$, by showing that $\phat$ is an available jumpable point in form of $\MPthree$ for both parties in scale $\jsp-1$.
\end{enumerate}

We begin with the proofs regarding \textbf{Item 1}, and then we will do \textbf{Item 2} (which is much shorter) after that.

\paragraph{Item 1.}  We begin with a claim about the point $q$.
\begin{claim}\label{claim:q-status-bvc}
Suppose that $b < \phat$.
    With the set-up above, at time $\tsp$, all of the following hold:
    \begin{itemize}
        \item $q$ is a scale-($w-1$) MP for Alice.  
        \item $q$ is either a scale-($w-1$) or a scale-$(w-2)$ MP for Bob.
        \item $q$ is a jumpable point at scale $w-2$.
    \end{itemize}
\end{claim}
\begin{proof}
    Notice that, since $p$ is a scale-$w$ MP for Alice at time $\tsp$, we have $\Aiter \geq \hat{p}$.  We also have
    \[ 2^{w-3} > L^- \geq \Aiter - b > \Aiter - \hat{p},\]
    where the first inequality is by \eqref{eq:starw}; the second is the definition of $L^-$; and the last is the assumption that $b < \phat$ that we are making in Case 1.  Thus, $\hat{p} \leq \Aiter < \hat{p} + 2^{w-3}$, which implies that $q = \lfloor \Aiter \rfloor_{2^{w-1}} - 2^{w-1}$ is a scale-$(w-1)$ MP for Alice, and in particular that $q \in M_{\Aiter}$.
    Next, the same logic shows that $\Biter < \hat{p} + 2^{w-3}$, and from the definition of $L^-$ we also have $\Biter \geq \Aiter - L^- \geq \hat{p} - 2^{w-3}$, so 
    \begin{equation}\label{eq:bobconstrained}
    \hat{p} - 2^{w-3} \leq \Biter < \hat{p} + 2^{w-3}.
    \end{equation}
    If $\Biter \in [\hat{p} - 2^{w-2}, \hat{p})$, then $q$ is a scale-$(w-2)$ MP for Bob; if $\Biter \in [\hat{p}, \hat{p} + 2^{w-2}),$ the $q$ is a scale-$(w-1)$ MP for Bob.  In either case, $q \in M_{\Biter}$. Now at scale $w-2$, $ q = \MPone_A (w-2, \Aiter)$  and $q \in \MPonetwo_B(w-2, \Biter)$  thus $q$ is a jumpable point at scale $w-2$.
\end{proof}

\begin{claim}\label{cl:if_q_available_BVC_big}
    At time $\tsp$, if $q\in \M_A \cap \M_B$  then $\BVC_{AB} \geq \constbvc k$.
\end{claim}
\begin{proof}
    First we have the following sub-Claim.

            \begin{subclaim} If $q\in \M_A \cap \M_B$, then
                $q$ is an available jumpable point at scale $j_p-1$.
            \end{subclaim}
            \begin{proof}
                     By Claim \ref{claim:q-status-bvc},  $q$ is a jumpable point at scale $w-2$.  As we are assuming in the statement of the Sub-Claim that $q$ is available for both parties, $q$ is an available jumpable point at scale $w-2$. We have also established earlier in \eqref{eq:jp_is_wminus1} that $j_p = w-1$, so $j_p -1  = w-2$, and the proof is complete.
            \end{proof}
         If $j_p -1 \geq \jsp$, it contradicts the assumption in Lemma \ref{lem:maintech} that $j_p$ is the next available jumpable scale.  So $j_p -1 < \jsp$. 
 Alice and Bob have skipped $q$ as a meeting point, which according to Lemma~\ref{lem:skipsCostAdv} implies that $\BVC_{AB} \geq (0.6) 2^{j_p-1}$. 
 As $j_p\geq \jsp$, it means $j_p = \jsp$ and $\jsp-1 = j_p -1 $, so $BVC_{AB} \geq (0.6) 2^{\jsp-1 } \geq \constbvc k$ where $ \constbvc =(0.6) 2^{-2}$, where in the  last inequality we have used the fact that $k < 2^{\jsp + 1}$  by the definition of $\jsp := \lfloor \log k \rfloor$.
 This proves the claim.
\end{proof}

\begin{claim}\label{cl:q_av}
Assume that $b < \phat$.  Then, given the above set-up, we have that $q\in \M_A \cup \M_B$.
\end{claim}
\begin{proof}
    For the sake of contradiction assume otherwise. This means that Alice and Bob have forgotten the state at the point $q$ by reaching a state at depth $c_q = q + 2^{w} = \phat + 2^{w-1}$. Let  $\tcq^A$  be the time that Alice reaches $c_q$ for the last time before $\tsp$. Let $\tcq^B$ be the time that Bob reaches $c_q$ for the last time before $\tsp$.  Then we have the following sub-claim.
            \begin{subclaim}
                With the set-up above, $\tcq^A , \tcq^B < t_b $. 
            \end{subclaim}
            \begin{proof}
              
               For the sake of contradiction assume otherwise, and without loss of generality assume $\tcq^A \geq t_b$. Then, $L^-(\tcq^A) = c_q - b\geq c_q - \phat = 2^{w-1}$. This contradicts the assumption that $L^- \leq 2^{w-3}$. As a result we have that $\tcq^A, \tcq^B < t_b$, as desired.
            \end{proof}

            Let  time $\min\{\tcq^A, \tcq^B\} \leq t \leq t_b$  be the first time either party jumps above $\phat$ after that same party has reached $c_q$.  Note that such a time $t$ must exist because $b < \phat$. This is because, given our assumption that $t_b$ is not during a Big Hash Computation phase, and applying  Lemma~\ref{lem:noweirdcoincidences}, Item 1, we see that some party must jump to $b$ to start the bad spell at time $t_b$. 
            Let us first assume that this is Bob.
            Then by Lemma \ref{lem:MP-jump}, Bob jumps to a point above $q$ at time $t$.  Since $\Biter \geq b$, Bob passes the point  $q$ at least one more time before $\tsp$. As a result,  either $q\in \M_B$, which contradicts our assumption in the start of the proof; or Bob forgets $q$ again at some time after $\tcq^B$ and before $\tsp$, which contradicts the assumption that $\tcq^B$ was the last time Bob forgot $q$ before $\tsp$.  In either case, we have a contradiction of our assumption that $q \not\in \M_A \cap \M_B$.

            The same logic applies if Alice is the party who first jumps above $\phat$.    Indeed, the only part of the logic above that was specific to Bob is that $\Biter \geq b$; since $\Aiter\geq\Biter \geq b$, the same holds for Alice too.

            Thus, in either case we have a contradiction of our assumption that $q \not\in \M_A \cup \M_B$, so we conclude that indeed $q \in \M_A \cup \M_B$, which proves the claim.
\end{proof}

So far we shown (in Claim~\ref{cl:if_q_available_BVC_big}) that if $q\in \M_A \cap \M_B$, then $ \BVC_{AB} \geq \constbvc k$. We've also shown (in Claim~\ref{cl:q_av}) that if $b < \phat$, then $q \in \M_A \cup \M_B$.  

Next, in Claim~\ref{cl:if_q_available_once_sneaky_attack} we will show that, assuming $b < \phat$, if only one party has $q$ as an available meeting point (that is, if $q \not\in \M_A \cap \M_B$ but $q \in \M_A \cup \M_B$) then a sneaky attack must be in progress.  This will complete the proofs for \textbf{Item 1.}

\begin{claim}\label{cl:if_q_available_once_sneaky_attack}
 
    Suppose that $b < \phat$.
    Suppose that $q$ is an available meeting point for exactly one party at time $\tsp$; that is, $q \in \M_A \cup \M_B$ but $q \not\in \M_A \cap \M_B$.  Then a sneaky attack is in progress.
\end{claim}
\begin{proof}
To prove this statement we will show that all properties of a sneaky attack are satisfied. Consider the case where $ q\notin \M_A$ but $q\in \M_B$. The case where $q\notin \M_B$ and $q\in \M_A$ follows in a similar manner.  

Assuming $q\notin \M_A$ but $q\in \M_B$, then at some point before $\tsp$, Alice has reached a point of depth $c_q = q + 2^{w} = \phat + 2^{w-1}$ and forgot $q$. 

\begin{definition}\label{def:tcq}
Let $\tcq \leq \tsp$ be the last time Alice reaches $c_q$. Further define $\tphat$ to be the last time Alice reaches $\phat$ before $\tcq$. 
\end{definition}
Then we have the following claim. 
\begin{subclaim}\label{claim:case1-alice-phat}
For any time $t\in [\tphat, \tsp]$, $\Aiter(t)\geq \phat$. 
\end{subclaim}
\begin{proof}
    For the sake of contradiction, suppose that there exists a time $t \in (\tphat, \tsp]$ such that $\Aiter(t) < \phat$.
    
    If $t \in (\tphat, \tcq]$, then there also exists a time $t^\prime\in (t,\tcq]$ such that $\Aiter (t')= \phat$. This is because Alice needs to pass depth $\phat$ to reach $c_q > \phat$, but this contradicts the definition of $\tphat$.

    On the other hand, if $t\in (\tcq , \tsp]$, suppose without loss of generality that $t$ is the first time after $\tcq$ so that $\Aiter(t) < \phat$. 
    Then according to Lemma \ref{lem:MP-jump}, $\Aiter(t) \leq q $. Given that $\Aiter (\tsp) \geq b$ and $b > q $, Alice must pass the depth $q$ again to reach $\phat$. This reintroduces $q$ into $\M_A$. If Alice never reaches $c_q$ again before time $\tsp$, then $q \in \M_A$ at time $\tsp$, but this contradicts the assumption that $q\notin \M_A$. On the other hand, if Alice \emph{does} reach $c_q$ again, this would contradict the definition of $\tcq$.  In either case, by contradiction we reach the conclusion that $\Aiter (t) \geq \tphat$ for all $t\in [\tphat, \tsp]$, as desired.
\end{proof}

We will show that $\tcq$ and $t_b$ are in two different bad spells using the claims below. 
\begin{subclaim}\label{subclaim:t_b-after-tcq-item1}
    $t_b >  \tcq$.  
\end{subclaim}

\begin{proof}
For the sake of contradiction assume otherwise, so $t_b \leq \tcq$. Then the time $\tcq \in [t_b, \tsp)$ is included in the bad spell containing $\tsp$. 
But then we have
\begin{equation*}
  L^-(\tsp) \geq   L^- (\tcq) \geq c_q - b \geq c_q - \phat = 2^{w-1} > 2^{w-3}.
\end{equation*}

Above, the first inequality is because $L^-$ is non-decreasing during a bad spell; the second is because at time $\tcq$, Alice is at depth $c_q$ while the divergent point is $b$, so $L^-$ is at least $c_q - b$; and the third is because $b < \phat$ by assumption.
However, 
this contradicts the fact that $L^- \leq 2^{w-3}$ at time $\tsp$. Thus $t_b > \tcq$ and the proof is complete. 
\end{proof}

\begin{subclaim}\label{Claim:bob-jumps}
    Bob is at the point $b$ at time $t_b$.  That is, $\Biter(t_b) = b$.
\end{subclaim}
\begin{proof} 

As we have assumed that $t_b$ is not during the Big Hash Computation phase, Lemma~\ref{lem:noweirdcoincidences} Item 1 implies that, for a point $b$ to become a divergent point, one of two things must occur. First, it could be that at time $t_b -1$, $\Aiter = \Biter = b $ and  $\ms{\Aiter} = \ms{\Biter}$; and then at time $t_b$,  $\ms{\Aiter} \neq \ms{\Biter}$ and $b$ becomes the last point Alice and Bob agree on and hence the divergent point.  Alternatively, it could be that at time $t_b$, a party has jumped to point $b$ making it the divergent point.

In the current set-up, the first case is not possible. To see this, notice that $t_b > \tphat$ and as a result $t_b-1\geq \tphat$. From Claim~\ref{claim:case1-alice-phat} we know that $ \Aiter(t) \geq \phat$ for any time $t\in [\tphat, \tsp]$ thus, $\Aiter(t_b-1) \neq b$. 
Thus, we conclude that the second has occurred, and some party has jumped to $b$ at time $t_b$.

Again from Claim~\ref{claim:case1-alice-phat}, $\Aiter(t) \geq \phat > b$ for any $t \in  [\tphat, \tsp]$. Further from Sub-Claim~\ref{subclaim:t_b-after-tcq-item1} and the definition of $\tphat$, we know that $\tphat \leq \tcq < t_b \leq \tsp$, so we cannot have $\Aiter(t_b) = b$.   Thus, we reach the conclusion that Bob has made a jump to point $b$ at time $t_b$, which completes the proof.
\end{proof}

\begin{subclaim}\label{claim:bob-limited-depth}
    For any time $t\in [\tphat, t_b]$, $\Biter(t)  < \phat + 2 ^{w-\csneak}$.
\end{subclaim}

\begin{proof}
    For contradiction assume there exists a time $t\in [\tphat, t_b]$ such that $ \Biter(t) \geq \phat +  2^{w-\csneak}$. We will show that this contradicts the fact that Bob jumps to point $b$ at time $t_b$ (Claim \ref{Claim:bob-jumps}).

    First notice that, if $b$ is $j_b$ stable, then $j_b\leq w-\csneak$. The reason is that $b\in [\phat - 2^{w-\csneak}, \phat)$, 
    and $\phat$ is at least $w$-stable.  
    Thus, $\phat-2^{w-\csneak}$ can be at most $(w-\csneak)$-stable; 
    and anything in the open interval $(\phat - 2^{w - \csneak}, \phat)$ is at most $(w-\csneak)$-stable, as it lies strictly between consecutive multiples of $2^{w-\csneak}$. 
    Thus, if Bob reaches depth $\phat + 2^{w-\csneak}$ (which by our assumption he does at some time in $(\tphat, t]$), he will forget the meeting point $b$, because Bob forgets $b$ at depth $c_b = b + 2^{j_b+1} = \phat + 2^{j_b} \leq \phat + 2^{w-\csneak}$. 
    
    However, by Claim~\ref{Claim:bob-jumps}, $\Biter(t_b) = b$. Then there must be a time $t'\in [t,t_b]$ such that $\Biter(t') < b$ and $b$ is reintroduced as a meeting point afterward. However, from Claim \ref{claim:case1-alice-phat} we know that $\Aiter(t') \geq \phat$, so Alice and Bob cannot be at the same depth at time $t'$; this implies that $t'$ is in a bad spell, and further that the divergent point at time $t'$ is some $b' \leq \Biter(t') < b$. 
    
    This bad spell is only resolved if Alice and Bob simultaneously reach a point above $b'$. Thus this bad spell is not resolved by $\tsp$. But this contradicts the fact that at $\tsp$, $b$ is the divergent point, because $b' < b$ and the new divergent can only decrease in depth, by to Lemma  \ref{lem:noweirdcoincidences} (Item 6). Hence, such a time $t$ does not exist and we can conclude that $\Biter(t) < \phat + 2^{w-\csneak}$ for all $t\in [\tphat, t_b]$.
\end{proof}

\begin{subclaim}\label{cl:different_bad_spells}
    The times $\tcq$ and $t_b$ are in two different bad spells. That is, there exists a time $t^\ast$ and a point $\qhat$ such that $\tcq < t^\ast < t_b$ and $\Aiter(t^\ast) = \Biter(t^\ast)= \hat{q}$  and $\ell^- =0$. Additionally, we have that $ \hat{q} \geq \phat$.
\end{subclaim}

\begin{proof}
    For the sake of contradiction assume otherwise. Then the bad spell containing $\tsp$ starts at some time $t$ and contains both $\tcq$ and $t_b$; so $ \tcq, t_b \in [t, \tsp]$.

Consider the value of $L^-(\tcq)$ at time $\tcq$. 
Using the fact from Sub-Claim~\ref{claim:bob-limited-depth} that $\Biter(\tcq) \leq \phat + 2^{w - \csneak}$, we have: 
 \begin{align*}
    L^-(\tcq) &=  \Aiter(\tcq) - \Biter(\tcq) \\
    &= c_q - \Biter(\tcq) \\
    &\geq \phat + 2^{w-1} - (\phat + 2^{w-\csneak}) \\
    &= 2^{w-1} - 2^{w-\csneak} \\
    &\geq 2^{w-2}.
 \end{align*}
 This contradicts the assumption of Lemma~\ref{lem:maintech} that $L^-(\tsp) \leq 2^{j_p-\csneak} \leq 2^{w-\csneak}$, as $L^- (\tsp ) \geq L^- (\tcq) \geq 2^{w-2}$, and we have $\csneak \geq 3$. Thus it is not possible for $\tcq$ and $t_b$ to be in the same bad spell. We conclude that a $t^\ast$ exists such that $\Aiter(t^\ast) = \Biter(t^\ast) = \qhat$ and $\ell^- (t^\ast) = 0$. 
 
 Finally, we claim that $\qhat \geq \phat$.  This follows because by Sub-Claim~\ref{claim:case1-alice-phat},  $\Aiter(t) \geq \phat$ for all $t \in [\tphat, \tsp]$.
 In particular, this holds at time $t^\ast$, when $\Aiter(t^\ast) = \qhat$.
\end{proof}

\begin{subclaim}\label{cl:qhat_is_phat}
    If $t^\ast$ is the first time Alice and Bob are ending a bad spell since $\tcq$, then $\hat{q} = \hat{p}$,where $\hat{q}$ is as in the statement of Sub-Claim~\ref{cl:different_bad_spells}. 
\end{subclaim}

\begin{proof}

We begin with some intuition.
    For Alice and Bob to end a bad spell they must simultaneously reach a point $\qhat$ such that $\simPath^{(\leq \qhat)}_A =\simPath^{(\leq \qhat)}_B $ and $\ms{\qhat}_A = \ms{\qhat}_B$.
In particular, it is important to note that Alice and Bob do \emph{not} get to re-simulate the point $\qhat$ that meet at: rather $\qhat$ was supposed to be a point that was simulated correctly, and Alice and Bob meet up at the \emph{end} of that previous simulation.  Informally, we will show that if $\qhat \neq \phat$, then the last time that Alice and Bob simulated $\qhat$ (before they first end a bad spell after $\tcq$) was in fact \emph{during} a bad spell, and we show that this implies that they did not simulate $\qhat$ correctly.  Thus meeting up at the end of $\qhat$ cannot have resolved the bad spell.  We make this formal below.

Suppose towards a contradiction that $\qhat$ is not equal to $\phat$.
    Notice that when Alice reaches the point $c_q$ at time $\tcq$, she forgets all meeting points in the interval $(\phat, \phat + 2^{w-2})$. From Claims \ref{claim:bob-limited-depth} and \ref{claim:case1-alice-phat} we know that $\qhat \in [\phat, \phat + 2^{w-\csneak})$. As a result, if $\qhat \neq \phat$, then Alice must have simulated $\qhat$ at some time $t\in [\tcq ,t^\ast)$. 
    Moreover, this time $t$ was during the first bad spell.  To see this, first note that $t \leq t^{\ast}$, so it was before the first time a bad spell ends after $\tcq$.  Moreover, we know that by time $\tcq$, this first bad spell had already begun, since Alice was at depth $\Aiter(\tcq) = c_q$ and Bob was not, since by Sub-Claim~\ref{claim:bob-limited-depth}, $\Biter(\tcq) < \phat + 2^{w - \csneak} \leq c_q$. 

    By definition, since the bad spell ends at the point $\qhat$,  it must be that $\simPath^{(\leq \qhat)}_A = \simPath^{(\leq \qhat)}_B$ and $\ms{\qhat}_A  = \ms{\qhat}_B$ at the time $t^{\ast}$ when the bad spell ends.
    In particular, Alice and Bob's simulated transcripts must have agreed on  $\qhat$ at time $t^{\ast}$.  Since simulated paths include iteration numbers, $\simPath_A^{(\leq \qhat)}$ and $\simPath_B^{(\leq \qhat)}$ can only have agreed if Alice and Bob simulated $\hat{q}$ during the \emph{same} iteration $I$.  Note that the iteration $I$ contains the time $t$ defined above, the time when Alice simulated $\qhat$ during the first bad spell; in particular, $I$ was during a bad spell. 
    
However, Lemma~\ref{lem:noweirdcoincidences} (Item 4), says that if $I$ is an iteration during a bad spell, and if both Alice and Bob simulate a point $\hat{q}$ during iteration $I$, then after the simulation, 
    $\ms{\qhat}_A \neq \ms{\qhat}_B$.
    But this contradicts requirement above that mega-states agree if a bad spell is resolved.

    Overall, we get a contradiction, and conclude that $\qhat = \phat$. The proof is complete.
\end{proof}

Finally, we put everything together to show that a sneaky attack is in progress, which proves the claim.  We go through each element of a sneaky attack.  In the proof above, we have defined the times $\tphat < \tcq < t_{\qhat} < t_b < \tsp$, and we check that the requirements of Definition~\ref{def:sneaky} hold.
\begin{itemize}
    \item By Definition~\ref{def:tcq}, $\Aiter(\tcq) = c_q$ and $\tcq$ is indeed the last time that this happens before $\tsp$.
    \item By Definition~\ref{def:tcq} again, $\tphat$ is the last time that Alice passes $\phat$ before $\tcq$.
    \item The fact that Alice and Bob jump to a point $\qhat \geq \phat$ to end a bad spell at time $t_{\qhat}$; and that $\qhat = \phat$ if this is the first time that a bad spell has ended after $\tcq$, are established by Sub-Claims~\ref{cl:different_bad_spells} and \ref{cl:qhat_is_phat}.
    \item From the statement of the lemma, the divergent point becomes $b$ at time $t_b$ and $b$ remains the divergent point until $\tsp$.  Further, Sub-Claim~\ref{Claim:bob-jumps} implies that Bob jumps to $b$ at time $t_b$.  The fact that $b \geq \phat - 2^{w - \csneak}$ follows from \eqref{eq:b_big}. 
    \item Sub-Claim~\ref{claim:bob-limited-depth} implies that $\Biter(t') < \phat + 2^{w - \csneak}$ for all $t \in [\tphat, t_b)$.  The fact that $\Biter(\tsp) < \phat + 2^{w-3}$ follows from the fact that 
    \[ 2^{w-3} \geq 2^{w - \csneak} \geq L^{-} \geq \Biter(\tsp) - b > \Biter(\tsp) - \phat\] 
    \item The fact that $q \not\in \M_A(\tsp)$ follows from our assumption at the beginning of the proof.  (If instead we had assumed that $q \not\in \M_B(\tsp)$, then the same proof goes through switching the roles of Alice and Bob, and we will conclude that a sneaky attack is in progress, but with Bob rather than Alice as the target.)
\end{itemize}
\end{proof}

This completes the elements of \textbf{Item 1}, where we address the case the $b < \phat$.  

\paragraph{Item 2.} Next, we consider the elements of 
\textbf{Item 2}, where we consider the case that either $b \geq \phat$ or that $\ell^-(\tsp) = 0$.

We begin with the case that $b \geq \phat$; as we will see afterwards, the proof follows almost identically if $\ell^-(\tsp) = 0$.

We already knew from~\eqref{eq:ellA_small} and the fact that $\Biter \leq \Aiter$ that 
\[
\Aiter(\tsp),\Biter(\tsp) \leq \phat + 2^w.
\]
Now with the additional assumption that that $b \geq \phat$ and the fact that $\Aiter(\tsp), \Biter(\tsp) \geq b$, we have that
\begin{equation}\label{eq:aiterbiter}
\phat \leq \Aiter(\tsp), \Biter(\tsp) < \phat + 2^w.
\end{equation}
As before, define $t_b$ as  the time when $b$ becomes the divergent point in the bad spell including $\tsp$. Define $\tphat^A$ and $\tphat^B$ as the last time that the corresponding party has reached $\phat$ before time $t_b$. Then we have the following claim. 

\begin{claim} \label{claim:phat-avaialble}
If $ b \geq \phat$, then $\phat \in \M_A \cap \M_B$. 
\end{claim}

\begin{proof}
Assume towards a contradiction that $\phat \not\in \M_A \cap \M_B$, and suppose first that $\phat \not\in \M_A$.  (The case where $\phat \not\in \M_B$ follows similarly). Let $\tphat \leq \tsp$ be the last time Alice reaches $\phat$ before $\tsp$. (Note that this is \emph{different} than the definition of $\tphat$ in Item 1, as it is the last time before $\tsp$, not before $t_b$.)  
Note that $2^w | \phat$, which implies that if $\phat$ is $j$-stable, then $j\geq w$.  If $p \notin \M_A$, then there exists a time such that Alice has forgotten $\phat$ by reaching a point at depth  $c_{\phat} = \phat + 2^{j+1} \geq \phat + 2^{w+1}$. Let $t$ be the last time Alice reaches $c_{\phat}$. Notice that upon reaching $c_{\phat}$, Alice forgets all meeting points in the interval $(\phat, \phat + 2^{w})$. However we know from \eqref{eq:aiterbiter} that 
\[ \phat \leq \Aiter(\tsp) < \phat + 2^{w}. \]
Then according to Lemma \ref{lem:MP-jump}  there must be a time $t'\in (t, \tsp)$ such that Alice jumps to a point above $\phat$, $\Aiter(t') \leq \phat$. As a result of this Alice must cross the depth $\phat$ one more time to reach depth $\Aiter(\tsp)$; but this contradicts the fact that $\tphat$ is the last time before $\tsp$ that Alice has reached $\phat$.
This proves that $\phat \in \M_A$. The same logic follows for Bob showing that $\phat\in \M_B$. Thus the proof is complete. 
\end{proof}

\begin{claim}\label{cl:phat_av_jump}
    If $b \geq \phat$, then at time $\tsp$, $\phat$  is a  available jumpable point  for scale $\jsp -1$. Thus $\BVC_{AB} \geq \constbvc k$.  
\end{claim}
\begin{proof}
We will first show that $\phat$ is an available jumpable point for scale $j_p -1$. Since $j_p$ is the next available jumpable point, thus it must be that $j_p -1 = \jsp -1$ and we will use Lemma~\ref{lem:skipsCostAdv} to show that this means that $\BVC_{AB}$ must be large, and the proof will be complete.

Since $j_p = w-1$, we have $j_p -1 = w-2$. At this scale:
\begin{itemize}
    \item if $\Aiter\in [\phat , \phat + 2^{w-2})$ then $\MPthree(j_p-1, \Aiter) = \phat$;
    \item if $\Aiter \in [\phat + 2^{w-2}, \phat + 2^{w-1}) $ then $ \MPtwo(j_p-1, \Aiter) = \phat$;
    \item if $\Aiter \in [\phat + 2^{w-1} , \phat + 2^{w})$ then $\MPone(j_p-1, \Aiter) = \phat$.
\end{itemize}
In any case, $\phat \in \MPall(j_p -1,\Aiter),$ and
the same logic holds for Bob. 
Further, our assumption in Item 2 that $b \geq \phat$ implies that $\ms{\phat}_A = \ms{\phat}_B$. 
We conclude that $\phat$ is indeed a available jumpable point for scale $w-2=j_p-1$.

 However,
we know that $j_p$ is the next available jumpable scale. As a result it must be that $j_p-1 < \jsp$ and by Lemma~\ref{lem:skipsCostAdv}, the adversary has invested in corruption to create bad vote counts and skip this scale. Then it must be that $ j_p = \jsp$. Thus at time $\tsp$, $BVC_{AB} \geq (0.6) 2^{j_p -1} = (0.6) 2^{\jsp -1} \geq \constbvc k$. 
This completes the proof of the claim.
\end{proof}

Finally, we address the case that $\ell^-(\tsp) =0$.  
In this case, all of the arguments from the case that $b \geq \phat$ go through.  Indeed, \eqref{eq:aiterbiter} holds because now $\Aiter(\tsp) = \Biter(\tsp)$; since $p$ is a scale-$w$ MP for Alice, it is also a scale-$w$ MP for Bob, and \eqref{eq:aiterbiter} follows.
From this, the proof of Claim~\ref{claim:phat-avaialble} follows unchanged.  The only part of the proof of Claim~\ref{cl:phat_av_jump} that requires $b \geq \phat$ (and in particular that $b$ exists) is in order to establish that $\ms{\phat}_A = \ms{\phat}_B$; but if $\ell^-(\tsp) = 0$, then this holds by definition.  Thus, the proof of Claim~\ref{cl:phat_av_jump} goes through as well.

Finally, we put together the logic in \textbf{Item 1} and \textbf{Item 2} as described in the \textbf{Organization} paragraph above, and conclude that:
\begin{enumerate}
    \item If $b < \phat$, then either $\BVC_{AB}$ is large or a sneaky attack is in progress.
    \item If either $b \geq \phat$ or if $\ell^-(\tsp) = 0$, then $\BVC_{AB}$ is large.
\end{enumerate}
In either case, we conclude that either $\BVC_{AB}$ is large or a sneaky attack is in progress, which proves Lemma~\ref{lem:maintech}.